\documentclass[10pt,a4paper]{article}
\usepackage[utf8]{inputenc}
\usepackage[english]{babel}
\usepackage{amsmath,amsthm,enumerate}
\usepackage{amssymb}
\usepackage[dvipsnames]{xcolor}
\usepackage{tikz}
\usetikzlibrary{calc}
\usetikzlibrary{patterns}
\usepackage{caption}
\usepackage{subcaption}
\usepackage[hmargin=2.5cm,vmargin=3cm]{geometry}
\usepackage[color=green!30]{todonotes}
\usepackage{graphicx}
\newtheorem{theorem}{Theorem}

\newtheorem{lemma}[theorem]{Lemma}

\newtheorem{corollary}[theorem]{Corollary}

\newtheorem{definition}[theorem]{Definition}

\DeclareMathOperator{\Plc}{Pl(C)}

\usetikzlibrary{arrows, decorations.markings,shapes,arrows,fit}
\tikzset{box/.style={draw, minimum size=0.5em, text width=0.5em, text centered}}

\usepackage{pdfcomment}

\title{DNA tile self-assembly for 3D-surfaces: Towards genus identification}

\author{Florent Becker and Shahrzad Heydarshahi}

\begin{document}
	
\maketitle
	
\begin{abstract}
We introduce a new DNA tile self-assembly model: the Surface Flexible Tile Assembly Model (SFTAM), where 2D tiles are placed on host 3D surfaces made of axis-parallel unit cubes glued together by their faces, called polycubes.
The bonds are flexible, so that the assembly can bind on the edges of the polycube. 
We are interested in the study of SFTAM self-assemblies on 3D surfaces which are not always embeddable in the Euclidean plane, in order to compare their different behaviors and to compute the topological properties of the host surfaces.

We focus on a family of polycubes called \emph{order-1 cuboids}.
\emph{Order-0 cuboids} are polycubes that have six rectangular faces, and order-1 cuboids are made from two order-0 cuboids by substracting one from the other. Thus, order-1 cuboids can be of genus~0 or of genus~1 (then they contain a tunnel). 
We are interested in the genus of these structures, and we present a SFTAM tile assembly system that determines the genus of a given order-1 cuboid. 
The SFTAM tile assembly system which we design, contains a specific set $Y$ of tile types with the following properties.
If the assembly is made on a host order-1 cuboid $C$ of genus~0, no tile of $Y$ appears in any producible assembly, but if $C$ has genus~1, every terminal assembly contains at least one tile of $Y$.

Thus, for order-1 cuboids our system is able to distinguish the host surfaces according to their genus, by the tiles used in the assembly. This system is specific to order-1 cuboids but we can expect that the techniques we use to be generalizable to other families of shapes.
\end{abstract}
	
\section{Introduction} 

In this paper, we introduce a new tile self-assembly model in order to perform self-assembly on 3-dimensional surfaces.
The field of tile self-assemby originates in the work of Wang~\cite{wang}, who introduced in 1961 \emph{Wang tiles}, that is, equally sized 2-dimensional unit squares with labels/colors on each edge (later called \emph{glues}) and designed a Turing universal computation model based on these tiles.
In 1998, inspired by Wang tiles and DNA complexes from Seeman's laboratory~\cite{seeman}, Winfree introduced in his PhD thesis~\cite{WinfreeThesis} the \emph{abstract Tile Assembly Model} (aTAM).
This model uses Wang tiling with an extra information: he associated a non-negative integer strength for each glue label. 
The power of DNA self-assembly enables to compute using this model. 
We refer to the survey~\cite{Survey} for more details on the literature, and to the online bibliography of Seeman's laboratory~\cite{SeemanBIB}.

Most of the early work in the DNA tile self-assembly literature deals with rigid assemblies in the Euclidean plane~\cite{Survey,2D} (since the assemblies are discrete, the Euclidean plane is usually seen as the lattice $\mathbb{Z}^2$), which is a natural and simple setting for this model.
However, it can be interesting to use self-assembly in richer settings. 
This has been investigated experimentally for instance in \cite{Hydrogel,Shape-assisted, Fabrication} where the assembly takes place on a preexisting surface and changes according to the surface. 
On the theoretical side, there have been some recent works on DNA tile self-assembly outside the Euclidean plane, such as tile self-assembly in mazes~\cite{maze}, where the tile placement is done on the walls of a certain maze. 
Other types of self-assembly  exist that also do not use the Euclidean plane, for example a model of 
cross-shaped origami tiles~\cite{origami}.
Another type of self-assembly not in the plane is 3D assemblies of complex molecules like crystals~\cite{crystals,crystals2}. Inspired by this, a recent model called \emph{Flexible Tile Assembly Model} (FTAM) was introduced by Durand-Lose et al. in 2020~\cite{FTAM}, as an extension of earlier work~\cite{3D2001}.
Here, we have Wang tiles but they self-assemble (without an input surface) in 3D space (modeled by the lattice $\mathbb{Z}^3$) as they can have, in addition to standard \emph{rigid} bonds, \emph{flexible} bonds that allow tiles to bind at any angle along the tile edges. The goal of the FTAM model is to construct complex 3D structures called \emph{polycubes} (3D shapes made of unit cubes)~\cite{polycubes}.

In 2010, Abel et al.~\cite{Abel} used a variant of the aTAM to implement shape replication, where tiles react to the shape of a preexisting pattern to reproduce it. However, this setting is on the 1 dimensional border of a 2D pattern instead of the 2D surface of 3D objects in our setting. 
In this setting (like in the current work), the main challenge is that the system must react to the shape of the space around, rather than to an external input it can read as it wants.

We are interested in studying what happens if we put the tiles on a given 3D surface that is not necessarily topologically equivalent to the Euclidean plane.
The intuition is that this could modify the computational behaviour of the tile self-assembly model, and we believe it will be interesting for practical systems, as in some practical settings, self-assembly could be performed on complex surfaces.
	
Inspired by the FTAM, we introduce a new model, called \emph{Surface Flexible Tile Assembly Model} (SFTAM).
In the SFTAM, we are given a 3D surface, on which the tiles of the self-assembly get placed.
The SFTAM is an intermediate between aTAM and FTAM. 
Unlike in the FTAM, our aim for introducing the SFTAM is \emph{not} for building 3D structures or surfaces: we assume that the host surface already exists. 
In the SFTAM, tile bonds are all flexible and the tiles can bind along the edges of the surface.

This model enables to use self-assembly on surfaces other than $\mathbb{Z}^2$. The aim of this article is to introduce the SFTAM model, and to demonstrate its usefulness by showing how it can be used on various types of surfaces. One of the most important properties of a surface is its \emph{genus}, which, intuitively, is the number of ``holes'' in the surface. The Euclidean plane has genus~0. We are interested in using the SFTAM on surfaces with different values of genus. For that, we study the problem of characterizing the surface of the assembly, according to its genus, using the SFTAM. It is quit easy to devise a system which can behave in some way only on the torus, but it is harder to make sure that it has always this behavior when it is in fact on a torus.

We focus on a family of 3D surfaces called \emph{cuboids}, which are special types of polycubes.
Polycubes can form complex surfaces, and their genus can be high. We focus on a simple family of polycubes that can have genus~0 or genus~1. 
More specifically, we define an \emph{order-0 cuboid} $C_0$ as a polycube which has only six faces. 
An \emph{order-1 cuboid} $C_1=C_0 \setminus C'_0$ is a polycube that is made from the difference of two order-0 cuboids $C_0$ and $C'_0$. Thus, an order-0 cuboid is a simple surface with genus~0, but an order-1 cuboid can either have genus~0 (potentially with a pit or concavities) or genus~1, if it has a tunnel.

In this paper, we will suppose that the SFTAM self-assembly is performed on the surface of an order-1 cuboid $C$. 
We design an SFTAM system whose assemblies differ when $C$ is of genus~0 and of genus~1 and thus, can be used to detect the genus of the surface $C$ of the assembly it is used on. 
The goal of this study is to show that performing self-assembly on surfaces of higher genus can be helpful.
We also demonstrate some techniques which may prove useful in characterizing the topological properties of a wide range of surfaces.

A \emph{tile assembly system} (TAS) in the SFTAM is defined in a natural way as an extension of the aTAM: tile types are made of four glue labels, each has a strength, there is a  seed assembly and a temperature (more formal definitions will be given later). An \emph{assembly} is a placement of tiles on facets of the surface of the cuboid $C$. Two tiles bind if they are adjacent (i.e. their placements on the surface share an edge) and their glue labels are the same. In particular, edges are flexible and as a result the tiles can be placed on the border of orthogonal faces of $C$. The assembly is \emph{producible} if it can be obtained by successfully binding tiles, starting from a  seed. It is \emph{terminal} if no additional tile can be bound to an existing tile.

 Let $C=C_0 \setminus C'_0$ be an order-1 cuboid with its three dimensions 
 at least 10 for $C'_0$. Our main result is to describe an SFTAM (TAS) $\mathcal{S_G}$ with a subset $Y$ of its tile types such that the following holds:

\begin{itemize}
\item if the order-1 cuboid $C$ has genus $0$, then no tile of $Y$ appears in any producible assembly of $\mathcal{S_G}$ on $C$, and
\item if $C$ has genus $1$, every terminal assembly of $\mathcal{S_G}$ on $C$ contains at least one tile of $Y$.
\end{itemize}

In other words, the genus of $C$ can be determined using the assemblies of $\mathcal{S_G}$ on $C$.  The assemblies of $\mathcal{S_G}$ consist of two phases:  a skeleton  forms on the cuboid and separates it into several regions, then the regions are partially filled by inner tiles.  For  a sketch of the skeleton and its inner filling for an order-0 cuboid see Fig.~\ref{fig:plan1}. 
The skeleton of the assembly forms in $3$ steps $R_X$ (in red), $R_Y$(in green) and $R_Z$(in blue). After the formation of the skeleton the tiles of type $t_{even}$ and $t_{odd}$ partially fill inside the skeleton.

When the cuboid has genus 1, we show that there must be some parts of the skeleton or the inner filling which intersect in a way that is not possible on a genus-0 cuboid. The tile types of $Y$ stick at the place where this happens. See Fig.~\ref{plan2}.

We start with basic definitions and notations in Section~\ref{def}, 
where we introduce and formalize our SFTAM model. 
Next, we introduce the family of order-1 cuboids and we show how SFTAM behaves on the family of order-1 cuboids as an assembly model in three dimensions.
In Section~\ref{sec-lemmas} we develop technical lemmas that will be necessary for the proof of our main result.
In Section~\ref{problem} we present our main result: a SFTAM tile assembly system that identifies the genus of  order-1 cuboids using specific tiles from that system. We conclude in Section~\ref{future}.

\begin{figure}
    \centering

\begin{tikzpicture}[scale=1.5]
  \draw[thick](2,2,0)--(0,2,0)--(0,2,2)--(2,2,2)--(2,2,0)--(2,0,0)--(2,0,2)--(0,0,2)--(0,2,2);
  \draw[thick](2,2,2)--(2,0,2);
  \draw[gray,->](0,0,0)--(3,0,0);
  \draw[gray,->](0,0,0)--(0,3,0);
  \draw[gray,->](0,0,0)--(0,0,3);

\node at (3.2,0,0) {$X$};
\node at (0,3.2,0) {$Y$};
\node at (0,0,3.3) {$Z$};

  \draw[red, dotted](1,0,2)--(1,0,0)--(1,2,0);
  
   \draw[green,thick](0,1,2)--(2,1,2)--(2,1,0);
  \draw[green, dotted](2,1,0)--(0,1,0)--(0,1,2);
  
    \draw[blue,thick](0,2,1)--(2,2,1)--(2,0,1);
  \draw[blue, dotted](2,0,1)--(0,0,1)--(0,2,1);

\node[line width=0.1mm,rectangle, minimum height=2mm,minimum width=2mm,fill=white!70,rounded corners=1mm,draw, label]  at (2,2,0) {$\textcolor{black}{t_{even}}$};
\node[line width=0.1mm,rectangle, minimum height=2mm,minimum width=2mm,fill=white!70,rounded corners=1mm,draw, label]  at (2,0,2) {$\textcolor{black}{t_{even}}$};
\node[line width=0.1mm,rectangle, minimum height=2mm,minimum width=2mm,fill=white!70,rounded corners=1mm,draw, label]  at (0,2,2) {$\textcolor{black}{t_{even}}$};
\node[line width=0.1mm,rectangle, minimum height=2mm,minimum width=2mm,fill=white!70,rounded corners=1mm,draw, label]  at (0,0,0) {$\textcolor{black}{t_{even}}$};

\node[line width=0.1mm,rectangle, minimum height=2mm,minimum width=2mm,fill=black,rounded corners=1mm,draw, label]  at (0,2, 0) {$\textcolor{white}{t_{odd}}$};
\node[line width=0.1mm,rectangle, minimum height=2mm,minimum width=2mm,fill=black,rounded corners=1mm,draw, label]  at (2,0, 0) {$\textcolor{white}{t_{odd}}$};
\node[line width=0.1mm,rectangle, minimum height=2mm,minimum width=2mm,fill=black,rounded corners=1mm,draw, label]  at (2,2, 2) {$\textcolor{white}{t_{odd}}$};
\node[line width=0.1mm,rectangle, minimum height=1mm,minimum width=2mm,fill=black,rounded corners=1mm,draw, label]  at (0,0, 2) {$\textcolor{white}{t_{odd}}$};

 \draw[red,thick](1,2,0)--(1,2,2)--(1,0,2);
   \node[inner ysep=10pt,rectangle,draw,fill=yellow,scale=0.15mm,label]  at (1,1,2) {$\textcolor{black}{Seed}$};

\end{tikzpicture}

    \caption{The skeleton of a $\mathcal{S_G}$ assembly on an order-0 cuboid is showed in color. It is started from a seed in yellow and after the formation of the skeleton, the regions partially fill by $t_{odd}$ and $t_{even}$ tile types. }
    \label{fig:plan1}
\end{figure}
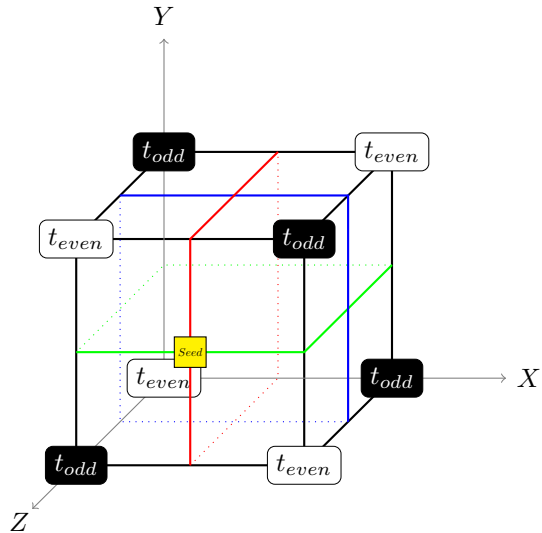

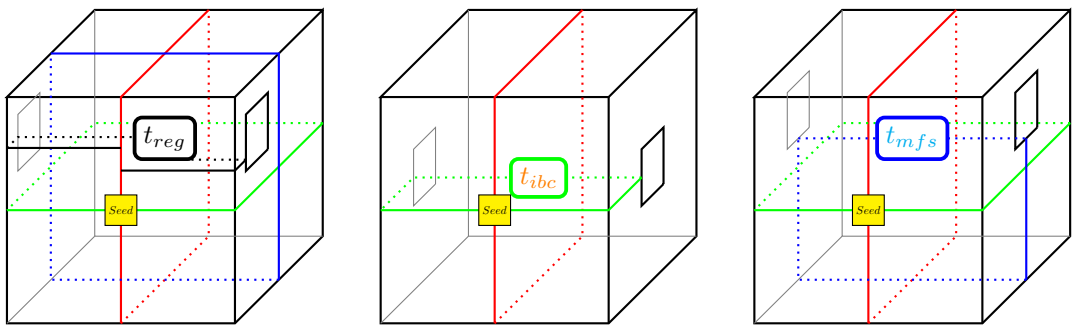
\begin{figure}[htpb!]
    \centering

\begin{subfigure}[t]{0.3\textwidth}
\begin{tikzpicture}[scale=1.5]
  \draw[thick](2,2,0)--(0,2,0)--(0,2,2)--(2,2,2)--(2,2,0)--(2,0,0)--(2,0,2)--(0,0,2)--(0,2,2);
  \draw[thick](2,2,2)--(2,0,2);
  \draw[gray](2,0,0)--(0,0,0)--(0,2,0);
  \draw[gray](0,0,0)--(0,0,2);

  \draw[red,thick](1,2,0)--(1,2,2)--(1,0,2);
  \draw[red, dotted,thick ](1,0,2)--(1,0,0)--(1,2,0);
  
   \draw[green,thick](0,1,2)--(2,1,2)--(2,1,0);
  \draw[green, dotted,thick ](2,1,0)--(0,1,0)--(0,1,2);

  \draw[thick](2,1.25,1.25)--(2,1.75, 1.25)--(2,1.75,1.75)--(2,1.25,1.75)--(2,1.25,1.25);
 \draw[gray](0,1.25,1.25)--(0,1.75, 1.25)--(0,1.75,1.75)--(0,1.25,1.75)--(0,1.25,1.25);
  
    \draw[blue,thick](0,2,1)--(2,2,1)--(2,0,1);
  \draw[blue,dotted,thick](2,0,1)--(0,0,1)--(0,2,1);
  
   \draw[black,thick](0,1.55,2)--(1,1.55,2);
  \draw[black,thick](1,1.35,2)--(2,1.35,2)--(2,1.35,1.75);
  \draw[black, thick, dotted](0,1.55,2)--(0,1.55,1.75)--(1,1.55,1.75);
    \draw[black, thick, dotted](2,1.35,1.75)--(1.5,1.35,1.75);

   \node[line width=0.5mm,rectangle, minimum height=0.5cm,minimum width=0.5cm,fill=white!70,rounded corners=1mm,draw=black, label]  at (1,1.25,1) {$\textcolor{black}{t_{reg} }$};

   \node[inner ysep=10pt,rectangle,draw,fill=yellow,scale=0.15mm,label]  at (1,1,2) {$\textcolor{black}{Seed}$};
 
\end{tikzpicture}
\caption{
The case where the skeleton does not meet the tunnel. In this case, the tile types $t_{odd}$ and $t_{even}$ of regions where  the entrances of the tunnel is located, pass inside the tunnel. In their collision a tile of type $t_{reg}$ appears in the assembly.}
\end{subfigure}
\begin{subfigure}[t]{0.3\textwidth}
\begin{tikzpicture}[scale=1.5]
  \draw[thick](2,2,0)--(0,2,0)--(0,2,2)--(2,2,2)--(2,2,0)--(2,0,0)--(2,0,2)--(0,0,2)--(0,2,2);
  \draw[thick](2,2,2)--(2,0,2);
  \draw[gray](2,0,0)--(0,0,0)--(0,2,0);
  \draw[gray](0,0,0)--(0,0,2);

  \draw[thick](2,0.75,0.75)--(2,1.25, 0.75)--(2,1.25,1.25)--(2,0.75,1.25)--(2,0.75,0.75);
 \draw[gray](0,0.75,0.75)--(0,1.25, 0.75)--(0,1.25,1.25)--(0,0.75,1.25)--(0,0.75,0.75);

  \draw[red,thick](1,2,0)--(1,2,2)--(1,0,2);
  \draw[red, dotted,thick](1,0,2)--(1,0,0)--(1,2,0);
  
   \draw[green,thick](0,1,2)--(2,1,2)--(2,1,1.25);
  \draw[green, thick, dotted](2,1,1.25)--(0,1,1.25)--(0,1,2);

   \node[line width=0.5mm,rectangle, minimum height=0.5cm,minimum width=0.5cm,fill=white!70,rounded corners=1mm,draw=green, label]  at (1,0.90,1) {$\textcolor{orange}{t_{ibc} }$};

   \node[inner ysep=10pt,rectangle,draw,fill=yellow,scale=0.15mm,label]  at (1,1,2) {$\textcolor{black}{Seed}$};

\end{tikzpicture}
\caption{The tunnel intersects along the width of plane $P_X$ and length of plane $P_Y$. The $t_{ibc}$ is a tile type from $T_{ibc}$}
\end{subfigure}
\begin{subfigure}[t]{0.3\textwidth} \begin{tikzpicture}[scale=1.5]
  \draw[thick](2,2,0)--(0,2,0)--(0,2,2)--(2,2,2)--(2,2,0)--(2,0,0)--(2,0,2)--(0,0,2)--(0,2,2);
  \draw[thick](2,2,2)--(2,0,2);
  \draw[gray](2,0,0)--(0,0,0)--(0,2,0);
  \draw[gray](0,0,0)--(0,0,2);

  \draw[red,thick](1,2,0)--(1,2,2)--(1,0,2);
  \draw[red, dotted,thick](1,0,2)--(1,0,0)--(1,2,0);
  
   \draw[green,thick](0,1,2)--(2,1,2)--(2,1,0);
  \draw[green,thick, dotted](2,1,0)--(0,1,0)--(0,1,2);

  \draw[thick](2,1.25,0.75)--(2,1.75, 0.75)--(2,1.75,1.25)--(2,1.25,1.25)--(2,1.25,0.75);
 \draw[gray](0,1.25,0.75)--(0,1.75, 0.75)--(0,1.75,1.25)--(0,1.25,1.25)--(0,1.25,0.75);
  
    \draw[blue,thick](2,1.25,1)--(2,0,1);
     \draw[blue,thick,dotted](0,1.25,1)--(2,1.25,1);
  \draw[blue,thick, dotted](2,0,1)--(0,0,1)--(0,1.25,1);

       \node[line width=0.5mm,rectangle, minimum height=0.5,minimum width=0.5cm,fill=white!70,rounded corners=1mm,draw=blue, label]  at (1,1.25,1){$\textcolor{cyan}{t_{mfs} }$};

   \node[inner ysep=10pt,rectangle,draw,fill=yellow,scale=0.15mm,label]  at (1,1,2) {$\textcolor{black}{Seed}$};
  \end{tikzpicture}
  \caption{The case where the tunnel of an order-1 cuboid is shown by a tile of type $t_{mfs}$, located at the intersection of the skeleton.}
  \end{subfigure}

    \caption{
   According to the relative position of the seed  and the tunnel, the detection of the tunnel is done by different tile types of $\mathcal{S_G}$ in the assembly. The seed in indicated in yellow and the skeleton is in color.
    }
    \label{plan2}
\end{figure}

\section{Definitions and notations} \label{def}
	
In this section, we present the notions and definitions that are used throughout this paper. First, we explain what we mean by a tile assembly model for 3D surfaces in $\mathbb{Z}^3$ and we present our model: the Surface Flexible Tile assembly model or \emph{SFTAM}. Then we present the family of surfaces that we work on: order-1 cuboids.

\subsection{The tile assembly model on surfaces in $\mathbb{Z}^3$: SFTAM}\label{SFTAM}
	
We now define the Surface Flexible Tile assembly Model, SFTAM.  We work in 3-dimensional space, on the integer lattice $\mathbb{Z}^3$. 
	
\begin{definition}[Tile type in SFTAM] 
Let $\Sigma$ be a finite label alphabet and  $\epsilon$ represent the null label.
A \emph{tile type} $t$ is a 4-tuple $ t = (t_1, t_2, t_3, t_4)$ with $t_i \in \Sigma \cup \{ \epsilon \}$  for each $i=\{ 1, 2, 3, 4 \}$. Each copy of a tile type is a \emph{tile} and $t_1, t_2, t_3, t_4$ are the glues of $t$. 
\end{definition}

Tiles are $2D$ unit squares whose sides are assigned the labels of the tile type. These squares are allowed to translate and rotate (unlike in aTAM), but they can not be mirrored 
(unlike in FTAM). In fact,  since the tiles stick to a given surface, we can assume that they have an inner face and an outer face and that they always attach with the inner face in contact with the surface.
In the definition of tile types, we show labels by numbers rather than cardinal directions. 
However, often, the orientation of a tile dictates the orientation of the tiles around it. Then, we use the expression ``northern label'' to refer to the label which will end up on the northern side (and similarly for east, west, south).	
	
\begin{definition}[Facet] A \emph{facet} is a face of the lattice $\mathbb{Z}^3$, i.e. a unit square whose vertices have integer coordinates.
\end{definition}

\begin{definition}[Polycube]
A \emph{polycube} is a 3D structure that is a subset of $\mathbb{Z}^3$ and is formed by the union of unit cubes that are attached by their faces.
\end{definition}

For a  facet of a polycube, there are four possibilities for placing a  tile. We formalize this as follows.
	
\begin{definition}[Placement]
		Let $C$ be a polycube. A \emph{placement} $p = (f , o)$ on $C$ consists of a facet $f$ on the surface of $C$, and a side $o$ of $f$, called its \emph{orientation}.

		We denote the set of all placements in $C$ by $Pl(C)$. 
		
		Given a tile type $t=(t_1,t_2,t_3,t_4)$ and a placement $p=(f,o)$, \emph{placing} $t$ at the placement $p$ defines a mapping from the edges of $f$ to the label alphabet $\Sigma$.
	 The $i$-th side of $f$ (starting from the orientation $o$ and going in clockwise direction, looking from the exterior of the surface of $C$) is associated with $t_i$.
	\end{definition}
	
Notice that the normal vector $n$ of the placement in the FTAM is not needed in the SFTAM. Indeed, in the SFTAM, the tiles do not reflect. The reason is that we fix that the normal vector (as used in the FTAM) always starts inside of the polycube and  points to the outside of the structure. So, the order of the tiles' sides is uniquely determined by the orientation of its placement.

Now we can define a \emph{tile assembly system} in the SFTAM.

\begin{definition}[Tile assembly system (TAS) on a polycube in SFTAM]
		
		A \emph{tile assembly system}, or TAS,  over the surface of a given polycube $C$ is a quintuple $\mathcal S = (\Sigma, T, \sigma, str, \tau)$, where :
		
		\begin{itemize}
			\item $\Sigma$ is a finite label alphabet, 
			\item $T$ is a finite set of tile types on $\Sigma$,
			\item $\sigma$ is called the \emph{seed} and can be a single tile or several tiles 
			\item $ str$ is a function from $ \Sigma\cup \{ \epsilon \}$ to non-negative integers called \emph{strength function} such that $str( \epsilon) = 0$, and
			\item $\tau \in \mathbb{N}$ is called the \emph{temperature}.
		\end{itemize}
	\end{definition}
While the System SFTAM is a theoretical system, its components have an analogy with elements of practical DNA settings. The labels are the single strands of DNA, the function $str$ show the strength of their connections and  the $\tau$  is the temperature. 
	
We present the definitions and notations of SFTAM assemblies  that we will use throughout the article. They define similar to the ones for the aTAM~\cite{Survey}.

	\begin{definition}
		An \emph{assembly} $\alpha$ of a SFTAM TAS $\mathcal S$ on a polycube $C$ is a partial function $\alpha:\Plc\dashrightarrow T$ defined on at least one placement such that for each facet $f$ of $C$, there is at most one placement $(f,o)$ where $\alpha$ is defined.

		For placements $p=(f,o)$, $p'=(f',o')$ of $\Plc$ with $\alpha(p)=t$ and $\alpha(p')=t'$ such that $f$ and $f'$ are distinct but have a common side $s$, we say that $t$ and $t'$ \emph{bind} together with the strength $st$ if the glues of $t$ and $t'$ placed on $s$ are equal and have the strength $st$. 
		
		The \emph{assembly graph} $G_{\alpha}$ associated to $\alpha$ has as its vertices, the placements of $\Plc$ that have an image by $\alpha$, and two placements $p$ and $p'$ are adjacent in $G_{\alpha}$ if the tiles $\alpha(p)$ and $\alpha(p')$ bind.
		
		An assembly $\alpha$ is \emph{$\tau$-stable} if for breaking $G_\alpha$ to any smaller assemblies, the sum of the strengths of disconnected edges of $G_\alpha$ needs to be at least $\tau$. 
	\end{definition}
	
 We start with a seed and then we add tiles one by one. This is formally described as follows.
	
\begin{definition}
Let $C$ be a polycube and
$\mathcal S = (\Sigma, T, \sigma, str, \tau)$  a SFTAM TAS with $\sigma$ positioned on a placement of $C$. 
An assembly $\alpha$ of $\mathcal S$ is \emph{producible} on $C$ if either $dom(\alpha)=\{(p\}$ and $\alpha(p)=\sigma$   where $p\in \Plc$, or if $\alpha$ can be obtained from a producible assembly $\beta$ by adding a single tile from $T\setminus\sigma$ on $C$, such that $\alpha$ is $\tau-stable$.	
We denote the set of producible assemblies of $\mathcal S$ by $A^{C}[\mathcal S]$.
An assembly is \emph{terminal} if no tile can be $\tau$-stably attached on $C$. The set of producible, terminal assemblies of $\mathcal S$ is denoted by $ A^{C}_\square [\mathcal S]$.
	
\end{definition}

	\subsection{Cuboids}
	
	In this part, we introduce the structures that we work on: \emph{order-1 cuboids}, which are sepcial types of polycubes. 
	We work in 3-dimensional space, on the integer lattice $\mathbb{Z}^3$.  
	We start with some definitions.
	
		\begin{definition}[Order-$0$ cuboid]
An \emph{order-$0$ cuboid} $C = (s_C, x_C, y_C, z_C)$ where $s_C=(s_x,s_y,s_z)\in\mathbb{Z}^3$ is the point of $C$ with smallest coordinates and $x_C, y_C, z_C$ are integers representing the length, width and height of $C$ is a 3D structure containing all points $(x,y,z)$ of $\mathbb{Z}^3$ such that $s_x\leq x\leq s_x + x_C$, $s_y\leq y\leq s_y + y_C$ and $s_z\leq z\leq s_z + z_C$. We denote the set of all cuboids by $O_0$.
	\end{definition}
	
Note that we work on the boundary surface of cuboids.

We are interested in 3D structures that are more complicated than order-0 cuboids, in particular 3D structures that can have \emph{tunnels}, that is, ``holes''. 
To this aim, we  introduce a family of polycubes called \emph{order-1 cuboids}.
	
\begin{definition}[Order-1 cuboid] 
		An \emph{order-1 cuboid} $C_1$ is the difference between two elements of $O_0$. 
		Given $C_0=(s_{C_0},x_{C_0},y_{C_0},z_{C_0})$ and $ C'_0=(s_{C'_0},x_{C'_0},y_{C'_0},z_{C'_0}) $ in $O_0$.  $ C_1= C_0 \setminus C'_0$  is an order-1 cuboid if  there is a $i\in\{x,y,z\}$ such that  $i_{C_0}\leq i_{C'_0} $. 
		We note $O_1$ the set of all order-1 cuboids.
		
	\end{definition}

	Note that an order-0 cuboid is a classic ``cuboid", i.e. it has six rectangular faces. 
		An order-1 cuboid can have an asymmetric surface, including a hole or a concavity.
		Moreover, the condition on the dimension of $C_0$ and $C'_0$ ensure that the surface of $C_1$ is connected.

	The genus of an order-1 cuboid is at most $1$. Note that the set of order-0 cuboids is a subset of the set of order-1 cuboids, that is, $O_0 \subseteq  O_1$. An order-1 cuboid $C_1=C_0\setminus C'_0$ can be of three different types, depending on how $C_0$ and $C'_0$ interact: (i) $C_0$ and $C'_0$ have no intersection, and $C_1$ is an order-0 cuboid, (ii) $C'_0$ cuts a hole in $C_0$, and $C_1$ is named as an order-1 cuboid with a \emph{tunnel} and has genus 1; and (iii) $C_0$ and $C'_0$ intersect but $C'_0$ does not cut a hole in $C_0$. If the cut is in the inner face of $C_0$, $C_1$ is an order-1 cuboid with a \emph{pit} and if the cut is in the side of a face of $C_0$, $C_1$ is an order-1 cuboid with a \emph{concavity}. In both cases of order-1 cuboid with a pit or with a concavity, the genus is $0$.
	
	We denote by $O^*_1$ the set of order-1 cuboids that are not order-0 cuboids, that is, the ones from items (ii) and (iii) above.
	
	The set of cuboids of $O^*_1$ with a pit are denoted by $O_1^p$, the ones with a tunnel, by $O_1^t$, and the ones with a concavity, by $O_1^c$.

\section{Technical definitions and lemmas : counters, U-turns and middle finding system}\label{lemma}\label{sec-lemmas}

Our main results uses some arithmetic and geometric computations, which are defined in $\mathbb{Z}^2$. A transfer lemma guarantees that they also work on any surface, if it is regular enough.

Given an assembly on $\mathbb{Z}^2$, we call the smallest axis-parallel rectangle containing the assembly, its \emph{underlying rectangle}.
If an SFTAM assembly is on a 3D surface, it is permitted to fold along the tiles' edges. 
The underlying rectangle is then the smallest subset of the surface which contains the assembly and is isomorphic to a rectangle of $\mathbb{Z}^2$, if it exists. 
See Fig.~\ref{fig:rectangle} for an illustration.
In this section, we present some lemmas about the SFTAM, but for simplicity we assume to be on $\mathbb{Z}^2$. 
However, due to the following lemma, the results of this section are applicable to assemblies on polycube surfaces as well.

\begin{lemma}
Let $\alpha$ be a producible assembly of an SFTAM TAS $S$ on $\mathbb{Z}^2$ with underlying rectangle $R$, and let $C$ be a polycube. If there exists a function $i:\mathbb{Z}^2\to C$ such that the restriction of $i$ to $R$ is a graph isomorphism, then the image of $\alpha$ under $i$ is producible on $C$.
\end{lemma}
\begin{proof}
If the seed of $S$ is placed at $p_s$ in $\mathbb{Z}^2$, it is placed at $i(p_s)$ on $C$. Since the tile bonds can fold along edges of $C$, the assembly on $C$ proceeds exactly as it proceeds on $\mathbb{Z}^2$, and each tile placed at a point $p$ in $\mathbb{Z}^2$ is placed at point $i(p)$ on $C$.
\end{proof}

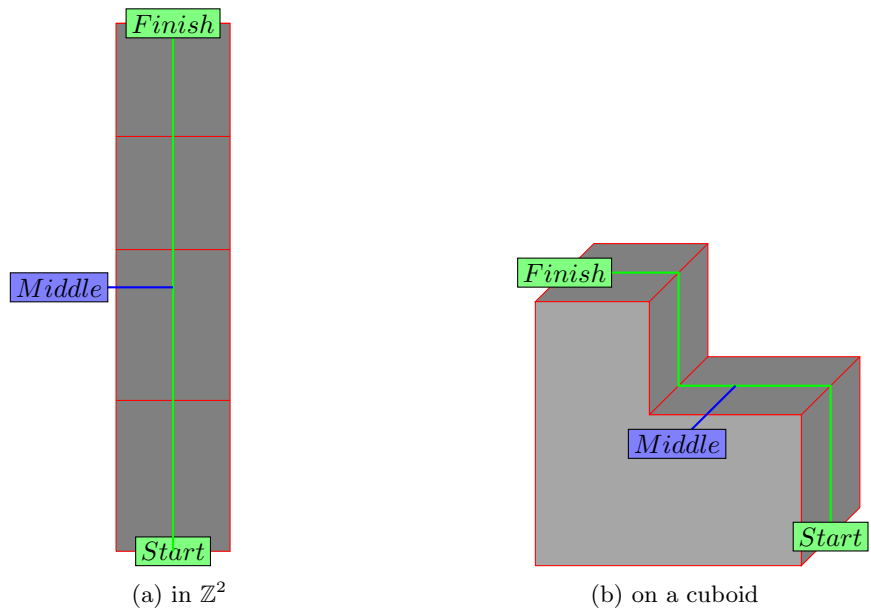
\begin{figure}[h!]
    \centering
\begin{subfigure}
[t]{0.4 \textwidth}
\begin{tikzpicture}[scale=0.5]
\filldraw[fill=gray, draw=red]
(0,0,0)--(3,0,0)--(3,14,0)--(0,14,0)--(0,0,0);

\draw[draw=red](0,4,0)--(3,4,0);
\draw[draw=red](0,8,0)--(3,8,0);
\draw[draw=red](0,11,0)--(3,11,0);
\draw[draw=red](0,14,0)--(3,14,0);

\node[inner sep=2pt,draw,fill=green!50, label]  at (1.5,0,0) 
{$\textcolor{black}{ Start}$};
\draw[draw=green, thick](1.5,0,0)--(1.5,14,0);
\draw[blue,thick, ->](1.5,7,0)--(-1.5,7, 0);
\node[inner sep=2pt,draw,fill=blue!50, label]  at (-1.5,7, 0) 
{$\textcolor{black}{ Middle}$};

\node at (-4.5,7, 0) {};

\node[inner sep=2pt,draw,fill=green!50, label]  at (1.5,14,0)
{$\textcolor{black}{ Finish}$};

\end{tikzpicture}\qquad
\caption{in $\mathbb{Z}^2$}
\end{subfigure}
\begin{subfigure}[t]{0.4 \textwidth}
\begin{tikzpicture}[scale=0.5]

\filldraw[fill=gray, draw=red] (3,2,-4)--(3,6,-4)--(-1,6,-4)--(-1,9,-4)--(-4,9,-4)--(-4,9,0)--(-1,9,0)--(-1,6,0)--(3,6,0)--(3,2,0)--(3,2,-4);

\draw[fill=gray!70,draw=red](3,2,0)--(-4,2,0)--(-4,9,0)--(-1,9,0)--(-1,6,0)--(-1,6,0)--(3,6,0)--(3,2,0);

\draw[draw=red](3,6,-4)--(3,6,0);
\draw[draw=red]((-1,6,-4)--(-1,6,0);
\draw[draw=red](-1,9,-4)--(-1,9,0);

\draw[draw=green, thick](3,2,-2)--(3,6,-2)--(-1,6,-2)--(-1,9,-2)-- (-4,9,-2);

\draw[blue,thick, ->](0.5,6,-2)--(0.5,6, 2);
\node[inner sep=2pt,draw,fill=blue!50, label]  at (0.5,6, 2) {$\textcolor{black}{ Middle}$};


\node[inner sep=2pt,draw,fill=green!50, label]  at (-4,9,-2)
{$\textcolor{black}{ Finish}$};
\node[inner sep=2pt,draw,fill=green!50, label]  at (3,2,-2) 
{$\textcolor{black}{ Start}$};

\node at (-6.5,7, 0) {};

\end{tikzpicture}
\caption{on a cuboid}
\end{subfigure}
    \caption{Finding the middle of a surface. The underlying rectangle is in dark gray.}
    \label{fig:rectangle}
\end{figure}

We introduce two increasing and decreasing counter TASs.
Then, by combining a TAS called \emph{U-turn system} with the counter systems, we establish a lemma called \emph{middle finding system lemma} to find the middle length of a given surface. 
This lemma paves the way to the proof of our main result.

There is another essential point before starting: we will represent a number via tiles, in a way that is classic in the literature of tile self assembly, see for example~\cite{BinaryCounter}.
To this aim, we use the binary representation of numbers and we associate a number to a row of tiles such that every bit of the number is assigned to a tile with the appropriate label. 

\begin{definition}[Row tile number]\label{rowtilenumber}
Let $T_0$ and $T_1$ be two sets of tiles with labels $0$ and $1$, respectively.
Let $N$ be an integer and $a_1 ... a_2 a_n$ be its binary representation with $n=\lceil log_2 N\rceil$. 
A row of tiles with labels $a_{ 1}^*, a_2, ..., a_n$ is the \emph{row tile number} representation of $N$ such that the distinct tile $a_{ 1}^*$, represents the most significant bit of the number. 
\end{definition}

We see an example of the row tile number of the number $12$ in Fig.~\ref{fig:12}. 

\begin{figure}[h!]
\centering
\includegraphics[scale=0.5]{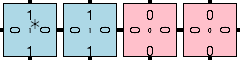}
\caption{Representation of the number $12$ by its row tile number.}
\label{fig:12}
\end{figure}

\subsection{The increasing binary counter system}\label{sec:counters}

In this section, we introduce the \emph{Increasing Binary Counter system} (IBC).
First, we explain what the IBC system is, then we present its tile types.
Afterwards, we show how it works by explaining the assembly process. 	
This system is a base for the middle finding system.

A different binary counter system was introduced in~\cite{BinaryCounter}. The advantage of our IBC system is that we do not need the support tiles at the bottom of the production. Especially, there is no infinite (or arbitrarily large) line of $0$ tiles at the left of the binary number of each row. Thus, the counter does not fill the whole plane and instead we will have a counter ribbon that is able to \emph{measure} the length of a given support on the side of a surface. 

In the production of the IBC system, the tiles of each row (excluding the tiles $\sigma_+$ and $t_{++}$) form a row tile number representing the index of the row in which they are located. 
Consequently, whenever the assembly stops by getting to the end of the given surface, the length of the assembly corresponds to the last row tile number. This number is a measure of the length of the underlying rectangle of the assembly.

\begin{lemma}[The increasing binary counter system, IBC]\label{IBC}

The \emph{increasing binary counter system} or \emph{IBC system}, is a SFTAM tile assembly system over $\mathbb Z^{2}$ such that if the bottom of the seed is placed at the origin, each row (excluding the tiles of type $\sigma_+$ and $t_{++}$) denotes its position in binary on the y axis and the leftmost tiles of all rows whose number is $2^k$ for some $k\geq 1$ is of type $t_{O}$.

The system is described by a quintuple $\mathcal S = (\Sigma, T, \sigma_{I}, str, \tau)$, where:
	
\begin{itemize}
		\item$\Sigma= \{0 , 1, s , ++ , over , \epsilon\}$ 
		\item $T=\{\sigma_{+}, t_{++} , t_{0+0},t_{0+1}, t_{1+0},t_{1+1},t_{S}, t_{O}, t_{OS} \} $
		\item $\sigma_{I}$ is a block of tiles made of a tile of type $\sigma_+=(1'' , \epsilon , \epsilon , 1'')$ (the $\epsilon$ labels are arbitrary) together with a column of \emph{support tiles} placed on the north side of $\sigma_+$, all of whose west labels are $1$ (the other labels can be arbitrary).
		\item $ str: \Sigma \cup \{ \epsilon \} \longrightarrow \{0,1,2\} $  such that $str( \epsilon) = 0$. The values of the str function for $\Sigma$ are $str( 0) = str( 1) =  str( ++) = str( s) = 1$ and $str(+)=str(over)=2$.
		\item $\tau=2$
\end{itemize}

\end{lemma}

The IBC tile types are shown in Fig.~\ref{ fig:Increment tiles }.

\begin{figure}[h!]
	\centering
	\begin{subfigure}[t]{0.4 \textwidth}
	\centering
	\includegraphics[scale=0.5]{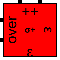}

	\caption{The initial seed tile $\sigma_+ $}
		\end{subfigure}
		\begin{subfigure}[t]{0.4 \textwidth}
	\centering
	\includegraphics[scale=0.5]{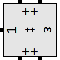}
	\caption{The support incrementor tile $t_{++}$}
		\end{subfigure}
		\begin{subfigure}[t]{0.4\textwidth}
		\centering
	\includegraphics[scale=0.5]{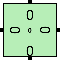}	\includegraphics[scale=0.5]{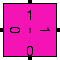}
	\includegraphics[scale=0.5]{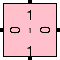}
	\includegraphics[scale=0.5]{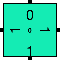}
\caption{The rule tile types $t_{0+0}$, $t_{0+1}$, $t_{1+0}$ and $t_{1+1}$}
		\end{subfigure}
	\begin{subfigure}[t]{0.4\textwidth}
	\centering
	\includegraphics[scale=0.5]{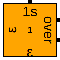}
	\includegraphics[scale=0.5]{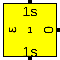}
	\includegraphics[scale=0.5]{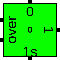}
	\caption{The tile types  $t_{O}$, $t_{S}$ and $t_{OS}$ related to the most significant bit}
		\end{subfigure}
	\caption{The IBC tile types.}
	\label{ fig:Increment tiles }
\end{figure}

\begin{proof}
We start by describing the IBC tiles.
\medskip

\noindent\emph{The seed of the IBC.}

There is one tile type $\sigma_+ = (++, \epsilon, \epsilon, over)$, and one incrementor tile type $t_{++}=(++, \epsilon , ++ , 1)$ (the gray tile in Fig.~\ref{ fig:Increment tiles }) called the support tiles.  
Together, one tile of type $\sigma_+$ with a column of tiles of type $t_{++}$ form the block seed of the IBC. Note that the tile type of the support is allowed to be different from $t_{++}$, as long as its western labels are $1$.

\medskip

\noindent\emph{The rule tile types of IBC.}

There are also seven sum rule tile types: four "standard" rule tile types and three rule tile types related to the most significant bit. However, from now on when we mention rule tiles, we refer to standard rule tiles.

\medskip

First, we explain how the binary sum is done by these tiles.
In each rule tile type, the inputs are the labels at the east and south of the tile, for example $a$ and $b$. 
The two outputs are: the classic binary sum $a+b$ in the northern label, and the carry bit, which is displayed on the western label of the tile. 
The rule tile types of the IBC are: 
$ t_{0+0}=(0,0,0,0)
	, t_{0+1}=(1,0,1,1)
	, t_{1+0}=(1,1,0,0)
	$
	and $t_{1+1}=(0,1,1,0)
	$ 
with the side labels as mentioned, and where the inner label is the value of the sum, the same as the northern label.

\medskip

\noindent\emph{The most significant bit tile types of the IBC.}

There are three tile types for the most significant bit. 
They use the same logic as the rule tile types.
The tile type $t_{S}=(1_S, 0, 1, \epsilon)$ represents the most significant bit of the row number when that number is not a power of $2$. 
In case $a = b = 1$ in the currently most significant bit, there is an overflow after the sum operation $a+b$. 
In this case, the two rule tiles  $t_{OS}, t_{O}$ are used to change the most significant bit.
The tile $t_{OS}=(0, 1 , 1_s, over)$ is the connecting tile between $t_{O}$ and $t_{S}$ in case of the overflow and $t_{O}=(1_s, over, \epsilon, \epsilon)$ is the new most significant bit when the row number is a power of 2.

\medskip

\noindent\emph{The process of the assembly in the IBC TAS.}

The process starts with the seed.
The tiles of this column start the process of incrementing the row below.
A tile of type $t_{O}$ binds to the west of the seed tile by label "over" with strength~$2$ in the east. 
Note that the label of $t_{O}$ is $1$, which corresponds to row number 1, where it is located. 
Now the incrementation process starts. Each tile is bound by its eastern and southern neighbors, i.e. for having a tile $t_*$ with the label $a+b$, there is a tile at the south with north label $a$ and there is a tile at the east with west label $b$. The first tile to bind in each row is the easternmost, it binds with the support. 
The carry bit of the sum appears in the west side of the tile $t_*$. 
The northern label of $t_*$ is $a+b$ too. 
See an example of an assembly for the number $16$ in Fig.~\ref{fig:16}.	

\begin{figure}[h!]
	\centering
	\includegraphics[scale=0.5]{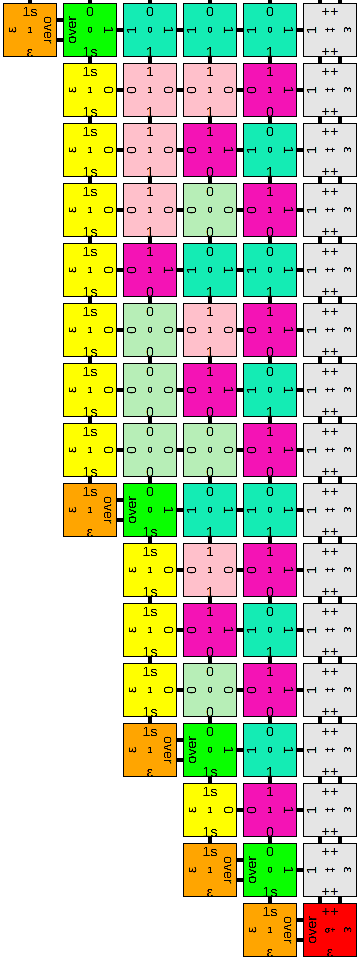}
	
	\caption{Assembly of 16 in the IBC.}
	
	\label{fig:16}
\end{figure}

Note again that each row represents the row number in which it is located along the y axis. This process can continue for the whole height of the support $\sigma_I$.

In the terminal assembly, the top row is the row tile number $N$, where $N$ is the length of the support. Moreover, the tiles of type $t_O$ only appear at the left of a row when the row is a power of $2$, and the highest one is the row number $\max\{c~|~2^c\leq N\}$.
\end{proof}

It can seen that  the IBC  counts up to a number, while distinguishing rows which are a power of two.

\subsection{The Decreasing Binary Counter System}

The \emph{decreasing binary counter system} is a tile assembly system that implements a reverse counter. 
This process is within the framework of \emph{Decreasing binary counter system} or \emph{DBC system} for short, formally defined as follows.

\begin{lemma}[Decreasing Binary Counter system, DBC system]\label{DBC}
The \emph{Decreasing binary counter system} or \emph{DBC system} for short, is a tile assembly system over $\mathbb Z^{2}$. It is defined by a quintuple $\mathcal S = (\Sigma, T, \sigma_{D}, str, \tau)$, where:
\begin{itemize}
		\item$\Sigma= \{0, 0',  1,1',  -- ,\epsilon \}$ 
		\item $T=\{ \sigma_{-}, t_{--}, t_{0-0},t_{0-1}, t_{1-0},t_{1-1}, t_{0-0}^*,t_{0-1}^*, t_{1-0}^*,t_{1-1}^*\}$, see Fig.~\ref{fig:DecrementTiles}.
		\item $\sigma_D$ is a block made of a binary row tile number $N$, a tile of type $\sigma_-=(--,\epsilon,\epsilon,0)$ (where the $\epsilon$ labels are arbitrary) at the east of it, and a column of support tiles that have western label $1$ (for example of type $t_{--}$) placed on the north of the $\sigma_-$ tile. The column has height at least $N$.
		\item $ str: \Sigma \cup \{ \epsilon \} \longrightarrow \{0,1,2\} $  such that $str( \epsilon) = 0$, $str(--)=2$ and all other values are 1.
		\item $\tau=2$
\end{itemize}

The terminal assembly of the DBC system is a rectangle with exactly one tile of type $t_{1-0}^*$ located at the left of the $(N+1)$-st row.

\end{lemma}

\begin{proof}
The DBC tiles are shown in Fig.~\ref{fig:DecrementTiles}. 

\begin{figure}[h!]
	\centering
	\begin{subfigure}[t]{0.5 \textwidth}
	\centering
	\includegraphics[scale=0.5]{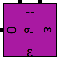}
	\caption{The seed tile $\sigma_-$}
		\end{subfigure}
\begin{subfigure}[t]{0.5 \textwidth}
\centering
	\includegraphics[scale=0.5]{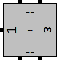}
		\caption{The support decrementor tile $t_{--}$}
		\end{subfigure}
\begin{subfigure}[t]{0.5 \textwidth}
\centering
	\includegraphics[scale=0.5]{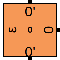}
	\includegraphics[scale=0.5]{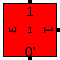}
	\includegraphics[scale=0.5]{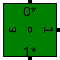}
	\includegraphics[scale=0.5]{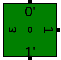}
		\caption{The tiles $t_{0-0}^*$, $t_{0-1}^*$, $t_{1-0}^*$, and $t_{1-1}^*$}
		\end{subfigure}
\begin{subfigure}[t]{0.5 \textwidth}
\centering
	\includegraphics[scale=0.5]{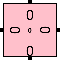}
	\includegraphics[scale=0.5]{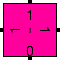}
	\includegraphics[scale=0.5]{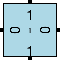}
	\includegraphics[scale=0.5]{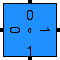}
			\caption{The tiles $t_{0-0}$, $t_{0-1}$, $t_{1-0}$ and $t_{1-1}$}
\end{subfigure}
	
\caption{The DBC tile types.}
	
	\label{fig:DecrementTiles}
\end{figure}

\begin{figure}[h!]
\centering

\includegraphics[scale=0.4]{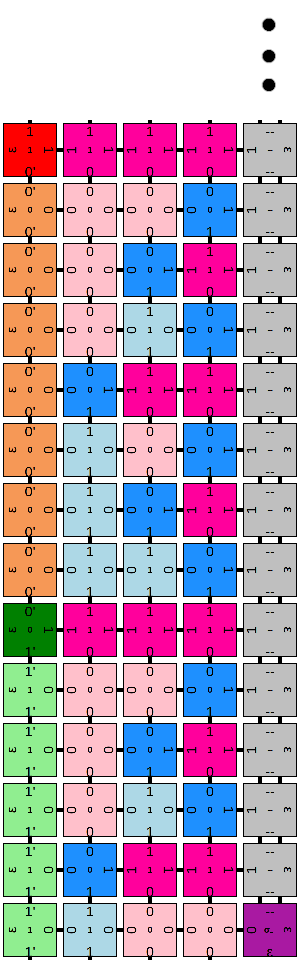}
\qquad\includegraphics[scale=0.25]{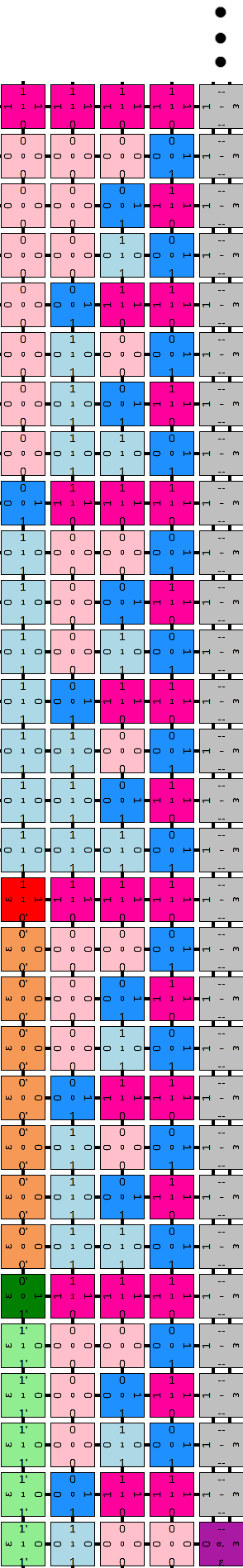}

\caption{The assembly of the DBC System for number 12 to negative 15}
	
	\label{ fig:12Decrement }
\end{figure}
\medskip

\noindent\emph{The seed of the DBC.}

An assembly of the DBC TAS starts when a binary row number binds to the decrement operation tile of type $\sigma_{D}=(--, \epsilon, \epsilon, 0)$ to form a DBC seed for the decrementation operation.

There is a decrementor tile of type $t_{--}=(-1, \epsilon , -1 , 1)$ which triggers the decrementation of the row below it into a new row to the left of the next tile of type $t_{--}$.
\medskip

\noindent\emph{The rule tile types of the DBC.}

There are eight rule tiles, four for the most significant bit and four for the others. The inputs are the labels of the east and south of the tile, for example $a$ and $b$, and there are two outputs.
The first one, that is located in the northern label, is the result of the binary subtraction $a-b$. 
The second output is the western label of the tile such that its value is needed to borrow for the next operation.
The label of the new tile itself is also $a-b$.
With this description, the rule tiles are $t_{a-b}=(a-b , b , a ,   borrow)$, i.e. $t_{0-0}=(0,0,0,0)$, $t_{0-1}=(0,1,1,1)$, $t_{1-0}=(1,0,1,0)$, $t_{1-1}=(1,1,0,0)$,
$t_{0-0}^*=(0',0,0',\epsilon)$, $t_{0-1}^*=(0',1,1',\epsilon)$, $t_{1-0}^*=(1',0,1', \epsilon)$, $t_{1-1}^*=(1',1,0',\epsilon)$.
See Fig.~\ref{fig:DBCrule} for an illustration.

\medskip

\noindent\emph{The process of the assembly in the DBC TAS.}

The counter starts at an arbitrary value $N$ and grows along a vertical column of support tiles. 
A decrementation process starts by relying on the support tiles that trigger the decrementation of the row below it into a new row to the left of the support tile. Each row $i$ shows the number $N-i$ modulo $min \{2^k~\vert~N \leq 2^{k}\}$.
Indeed, when the row's value counts down to $0$, the value of the next row will be the bits' maximum possible capacity of the number that the process was started with. 
Moreover, note that the most significant bit tiles are the tiles of types $t_{0-0}^*$, $t_{0-1}^*$, $t_{1-1}^*$ and
the row of the number $0$ for the first time is exposed  whenever a tile $t_{1-0}^*$ appears in the most significant bit in the row after the one of $0$, i.e. the row of $-1$. An example of a DBC assembly for $N=12$ is presented in Fig.~\ref{ fig:12Decrement }, where the tile of type $t_{1-0}^*$ is presented in red. After the appearance of $t_{1-0}^*$, the tiles of types $t_{0-0}^*$, $t_{0-1}^*$, $t_{1-0}^*$, $t_{1-1}^*$ do not appear in the assembly anymore and the process continues with tiles of types $t_{0-0}$, $t_{0-1}$, $t_{1-0}$ and $t_{1-1}$.
After that, the process repeats for the whole length of the support. Since the length of the support is at least $N$, the terminal assembly is a rectangle with exactly one tile of type $t_{1-0}^*$ located at the left of the $(N+1)$-st row.
\end{proof}
Hence the DBC  counts down from a number, while distinguishing the row where 0 is reached. 

	\begin{figure}[h!]
		\centering
		\scalebox{0.7}{\begin{tikzpicture}[node distance=7mm]
			
			\node[line width=0.5mm,rectangle, minimum height=2cm,minimum width=2cm,fill=white!70,rounded corners=1mm,draw] (Seed) at (0,0) {$a - b$};
			\draw[line width=0.5mm] ($(Seed.east) + (0,0.1)$) -- ($(Seed.east) + (0.2,0.1)$);
			\node (b) at ($(Seed.east) + (-0.3,0)$) {$b $};
			
			\draw[line width=0.5mm] ($(Seed.south) + (0,0)$) -- ($(Seed.south) + (0,-0.2)$);
			\node (a) at ($(Seed.south) + (0,0.3)$) {$a$};
			
			\draw[line width=0.5mm] ($(Seed.west) + (-0.3,0)$) -- ($(Seed.west) + (0,0)$);
			\node[rotate=90] (c) at ($(Seed.west) + (0.3,0)$) {$borrow$};
			
			\draw[line width=0.5mm] ($(Seed.north) + (0,0.2)$) -- ($(Seed.north) + (0,0)$);
			\node (d) at ($(Seed.north) + (0,-0.3)$) {$a - b$};

			\end{tikzpicture}}
		
		\caption{DBC rule tiles}
		\label{fig:DBCrule}
	\end{figure}
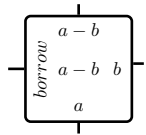

\begin{figure}[h!]
\centering
\includegraphics[scale=0.3]{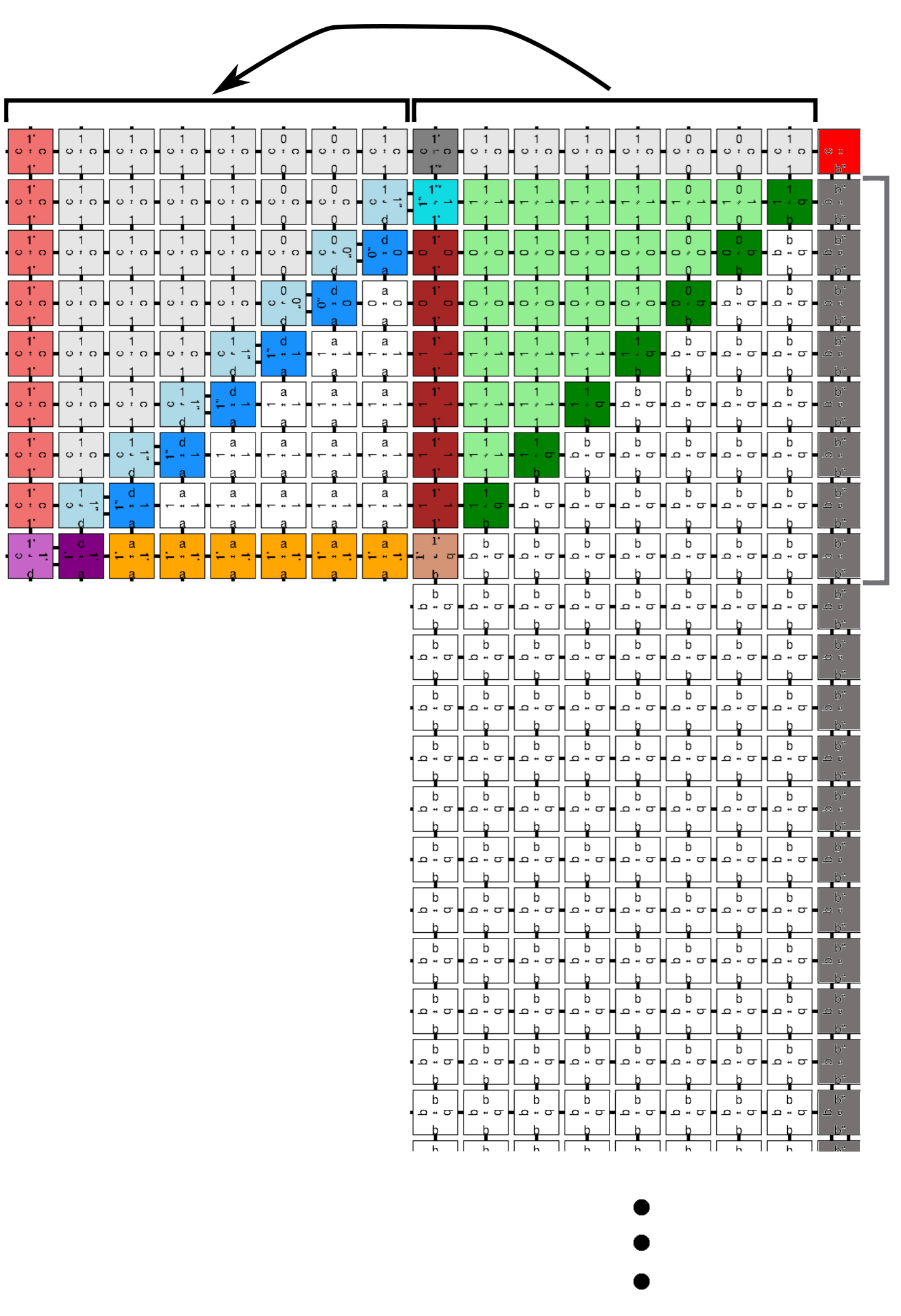}
\caption{Copying 11111001 in the U-turn system. The gray bracket on the right shows the minimum number of support tiles that are necessary for this assembly.}
\label{fig:11111001}
\end{figure}

\subsection{The U-turn system}

We now introduce a TAS which copies a row tile number to its left by doing a U-shape rotation of each of its tiles. The seed is analogous to that of IBC and DBC. The support must be at least as long as the width of the input.

\begin{lemma}[U-turn system lemma]\label{UTURN}
There is a TAS $\mathcal S_U =(\Sigma, T_U, \sigma_U, str, \tau)$ such that for a given row tile number $N$, $S_U$ makes a copy of $N$ to the left of its most significant bit. 
More precisely, if $N$ was in positions $[(x,y), ..., (x+k, y)]$, in the terminal assembly, there is a copy of $N$ in positions $[(x-k-1,y),\ldots, (x-1, y)]$.

The system is defined as follows.

\begin{itemize}
		\item$\Sigma= \{0 , 0'' , 1 , 1'' , 1' , 1'^* , a , b , b'' , c , d , \epsilon \}$ 
		\item $T_U=\{\sigma_u, t_{U}, t_{b},t_{\leftarrow},t^{0}_{\lrcorner}
		, t^{1}_{\lrcorner},t_{\swarrow},t_{\uparrow}, t_{\leftleftarrows}, t_{\llcorner} \} $. 
		\item The  seed $\sigma_U$ is formed by a tile of type $\sigma_u=(\epsilon,\epsilon,b'',c)$ (where $\epsilon$ can be arbitrary labels) which is bound to the right of a row tile number, and to the north of the $\sigma_u$ tile, a vertical column of support tiles (with at least as many tiles as in the row tile number) with western label $b$.
		\item $ str: \Sigma \cup \{ \epsilon \} \longrightarrow \{0,1,2\} $  such that $str( \epsilon) = 0$.
		The values of the str function in $\Sigma$ are  $str( 0'') =str( 1'') =str( b'') =  2$, and $1$ for the others.
		\item $\tau=2$
\end{itemize}
\end{lemma}

\begin{figure}[h!]
    \begin{subfigure}[t]{0.3 \textwidth}
    \centering
   
		\includegraphics[scale=0.5]{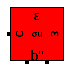}
       \caption{The initial seed tile $\sigma_u$}  
    \end{subfigure}
\begin{subfigure}[t]{0.3 \textwidth}
\centering
   
		\includegraphics[scale=0.5]{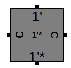}
		\includegraphics[scale=0.5]{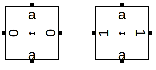}
          \caption{The row tile number tile types: $t^{1*}_{\uparrow}$, $t^0_{\uparrow}$ and $t^1_{\uparrow}$}   
    \end{subfigure}
\begin{subfigure}[t]{0.3 \textwidth}
\centering
   		\includegraphics[scale=0.5]{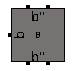}
           \caption{The support tile $t_u$}  
    \end{subfigure}
\begin{subfigure}[t]{0.3 \textwidth}
\centering
   		\includegraphics[scale=0.5]{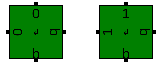}
	          \caption{The tile types $t^{0}_{\lrcorner}$ and 
		$t^{1}_{\lrcorner}$}   
    \end{subfigure}
\begin{subfigure}[t]{0.3 \textwidth}
\centering
   
		\includegraphics[scale=0.5]{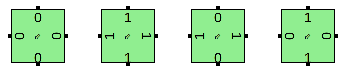}
           \caption{The tile types $t^{0,0}_{\swarrow}$, $t^{1,1}_{\swarrow}$, $t^{0,1}_{\swarrow}$ and $t^{1,0}_{\swarrow}$}  
    \end{subfigure}
\begin{subfigure}[t]{0.3 \textwidth}
	\centering
   
		\includegraphics[scale=0.5]{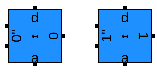}
           \caption{The tile types $t^{0''}_{\leftarrow}$ and $t^{1''}_{\leftarrow}$}  
    \end{subfigure}	
\begin{subfigure}[t]{0.3 \textwidth}
\centering
		\includegraphics[scale=0.5]{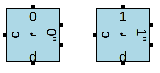}
           \caption{The tile types $t^{0''}_{\llcorner}$ and $t^{1''}_{\llcorner}$}  
    \end{subfigure}	
\begin{subfigure}[t]{0.3 \textwidth}
\centering
		\includegraphics[scale=0.5]{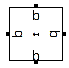}
         \caption{The tile type $t^{b}_{\leftarrow}$}    
    \end{subfigure}	
\begin{subfigure}[t]{0.3 \textwidth}
\centering
   
		\includegraphics[scale=0.5]{uturn-a.png}
          \caption{The tile types $t^{a0}_{\leftarrow}$ and $t^{a1}_{\leftarrow}$}   
    \end{subfigure}\\
\begin{subfigure}[t]{\textwidth}
\centering
   
		\includegraphics[scale=0.5]{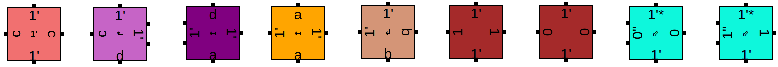}
         \caption{The marked tile types for copying the most significant bit tile (which is of type $t^{1*}_{\uparrow}$)}
    \end{subfigure}
    \caption{The U-turn system and support tile types.}
    \label{fig:UTURN_Tiles}
    \end{figure}

\begin{proof} Let $T_U$ be the set of tiles shown in Fig.~\ref{fig:UTURN_Tiles}, its temperature is $2$.
\medskip

\noindent\emph{The seed and the block tiles in the U-turn system.}

The U-turn system is a TAS that starts with a multiple seed such that the number of support tiles is at least the number of tiles of the row tile number.
Let the tile $\sigma_u=(\epsilon,\epsilon,b'',c)$ bind to a row tile number $N$  with $n=\lceil log_2 N\rceil$. We assume that the assembly is oriented as in Fig.~\ref{fig:11111001}, where the initial row tile number is at the top right.
From the south of $\sigma_u$, support tiles of type $t_u=(b'',\epsilon, b'', b)$ with northern and southern labels $b''$ whose strength are $2$, are placed one after another. These tiles form a ribbon of tiles of type $t_u$ that is perpendicular with respect to the seed. Once the number of support tiles becomes $n$, all these tiles together form the seed of the U-turn system, and the assembly begins.

After this step, the tiles below the seed start to grow and the assembly continues row by row. The overview of the assembly's process of the U-turn system is as below.

\medskip

\noindent\emph{Overview of the assembly's process in the U-turn system.}

The word "U-turn" refers to the process of the assembly in this system: each bit tile is copied along a U-shaped path to be copied to the left of the initial row tile number.
More precisely, each tile $a_k$ that is placed at the $k$-th bit of the row tile number, is copied $k$ times to the south, then it is copied $n$ times to the west, and at the end $k$ times to the north.
Furthermore, the assembly progresses row by row under the seed. 
Indeed, the last tile of each row has strength $2$ at the west and increments the number of row tiles. Thus, the next tile is placed at the $(n+k)$-th column (starting from the rightmost column of Fig.~\ref{fig:11111001}).
This tile makes a base, and using the support of the last tiles of each row, a vertical column is constructed that sends the value of $a_k$ to the north.

\medskip

\noindent\emph{The first stage of the assembly in the U-turn system.}

\begin{figure}[h!]
    \centering
    \includegraphics[scale=0.5]{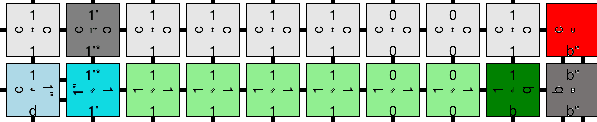}
    \caption{In the first stage of the assembly of the U-turn system, the value of the least significant bit is transferred $n$ times to the left (here $n=8$).}
    \label{fig:leastbittransfering}
\end{figure}

First, we explain the tiles of the very first row under the seed. For an illustration during this stage see Fig.~\ref{fig:leastbittransfering}.
This row is made of the tiles that transfer the value of the rightmost tile i.e. the least significant bit $a_1$ of the row tile number, to the column $n+1$.
At first, a tile of type $t^{0}_{\lrcorner}=(0,b,b,0)$ or $t^{1}_{\lrcorner}= (1,b,b,1)$ (depending on the value of the least significant bit, for short we use $t_{\lrcorner}$) appears in the assembly.
The tiles of type $t_{\lrcorner}$
transfer the value of the north label to the west label. These tiles are colored in dark green in Fig.~\ref{fig:11111001}.

Then, tiles of type $t^{1,1}_{\swarrow}=(1,1,1,1)$, $t^{1,0}_{\swarrow}=(1,0,1,0)$, $t^{0,1}_{\swarrow}=(0,1,0,1)$ and $t^{0,0}_{\swarrow}=(0,0,0,0)$ (for short $t_{\swarrow}$, light green in Fig.~\ref{fig:11111001}), that copy the north label to the south label, and the east label to the west label, appear under the row tile number $N$, at the west of the tile of type $t_{\lrcorner}$ (dark green in Fig.~\ref{fig:11111001}) except the leftmost one, which corresponds to the most significant bit.
Notice that in the U-turn system, all the tile types that transfer the value of the most significant bit tile are distinct from the other tiles.

The tile below the most significant bit is a special tile of type $t^{1*}_{\longleftarrow}=(1'^*, 1 , 1' , 1'')$ or $t^{0*}_{\longleftarrow}=(1'^*, 0 , 1',0'')$ (cyan in Fig.~\ref{fig:11111001}) that copies the north label to the south label and the east label to the west label with strength $2$.
Thus, this tile of type $t^{1*}_{\longleftarrow}$ or $t^{0*}_{\longleftarrow}$ allows to open a new column. The tile after it is of type
$t^{0''}_{\llcorner}=(0'', d, c,0)$ or $t^{1''}_{\llcorner}=(1'', d,c,1)$ (light blue in Fig.~\ref{fig:11111001})
with strength $2$ at the east label and brings its value to the north.
Now, the first copied tile, which corresponds to the least significant bit tile, appears at the left of the most significant bit tile of $N$.

Note that the tile of type $t_{\lrcorner}$ 
also copies the east label, i.e. block label $b$ from the support tiles $t_u$, to the south label.
Thus, at the south of the tile of type $t_{\lrcorner}$, the tiles of types $t^{b}_{\longleftarrow}=(b,b,b,b)$ (named block tiles) appear vertically with the support of $t_u$. 
\medskip

\noindent\emph{The overall transfer of the value in the U-turn assembly.}

Now, we explain the general process of the assembly for transferring the value of the $k$-th least significant bit of the row tile number. In Fig.~\ref{fig:kthbittransferig} this stage is marked with yellow-filled rectangles on the left side of the assembly. In this part, each new tile can be placed if and only if their northern and eastern labels correspond to the previous tiles.

\begin{figure}[htpb!]
    \centering
    \includegraphics[scale=0.6]{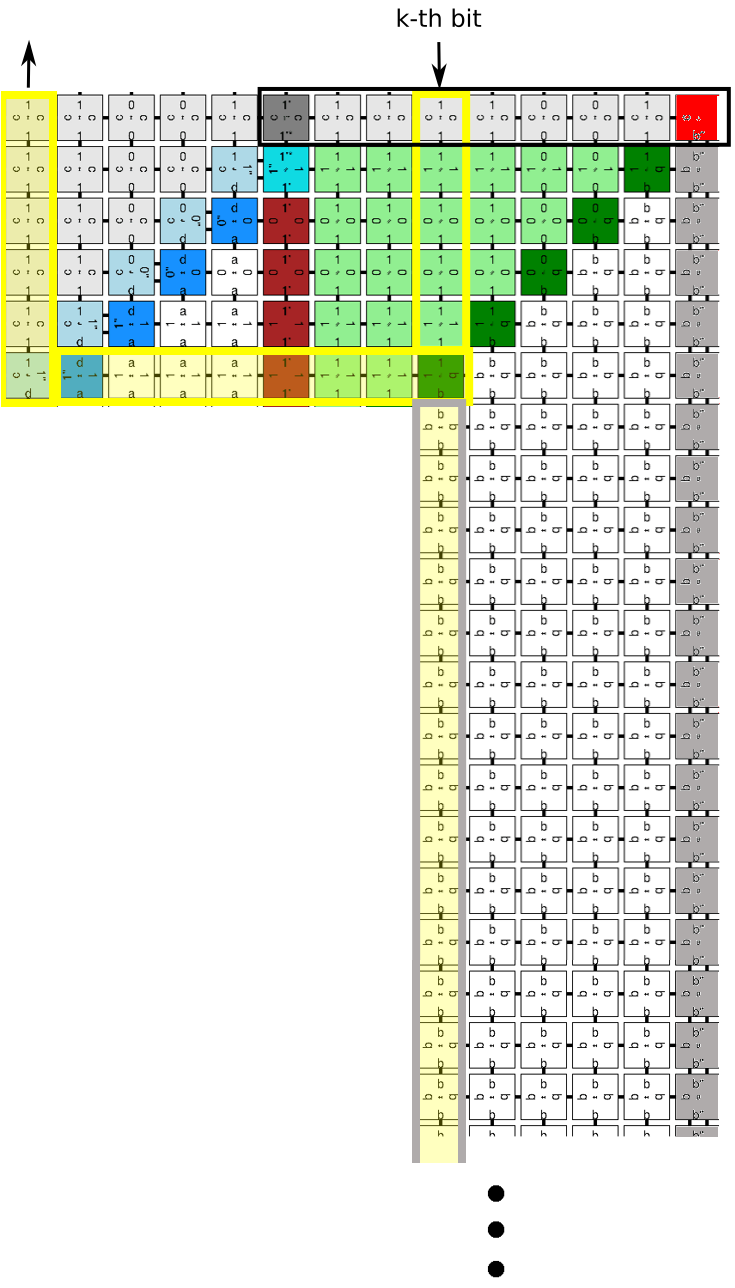}
    \caption{The $k$-th stage of the assembly in the U-turn system is shown by yellow-filled rectangles. The value of the $k$-th significant bit is copied down by $k-1$ rows during the previous stages. The $k$-th stage copies the value one time to the south and $n$ times to the left and finally $k$ times to the top. Here, $k=5$ and $n=8$. In addition, in the $k$-th stage, the tiles of type $t^{b}_{\longleftarrow}$ in the gray rectangle appear below the tile of type $t_{\lrcorner}$, and they will be the supports for the $(k+1)$-th stage. The seed is highlighted in the black rectange.}
    \label{fig:kthbittransferig}
\end{figure}

Assume that we are in the $k$-th stage i.e. the value of the $k$-th least significant bit is being transferred.

In the previous stages before $k$, the block tiles of types $t^{b}_{\longleftarrow}$ grow vertically from the south of the tile of type $t_{\lrcorner}$.
Thus, the tiles of type $t^{b}_{\longleftarrow}$ are placed between the support  tiles of type $t_u$ and $t_{\lrcorner}$ in each row and  copy the block label $b$ from the east to the west.
Therefore, the tile $t_{\lrcorner}$ appears in the $k$-th row and the $k$-th column with help of the tile of types $t_{\swarrow}$ in the north and $t^{b}_{\longleftarrow}$ at its east.
That being said, the tile that is placed at the $k$-th bit of the row number is copied $k-1$ times to the south before it meets the block tile, the place where the tile of type $t_{\lrcorner}$ (that transfers the north label to the west label) appears.

Then, the value of the $k$-th bit tile is copied one time to the south, and $n$ times to the west. 
Before arriving at the $(k+n)$-th column, the value passes through the $n$-th column, i.e. the column of the most significant bit tiles. 

As mentioned before, the path that the most significant bit goes through is created with marked tiles to keep it distinguished from the other bit tiles.
From the column of the most significant tile types, the tiles of type $t^{a}_{\longleftarrow}$ transfer their eastern value to the western value until arriving to the $(n+k)$-th column at the left. 
Note that, at the $(n+k-1)$-th bit towards the left, a tile of type $t_{\leftarrow}^{1''}=(d,0, a, 0'')$ or $t_{\leftarrow}^{0''}=(d,1, a, 1'')$ (dark blue in Fig.~\ref{fig:11111001}) opens a new column that starts with $t^{0''}_{\llcorner}=(0'', d, c,0)$ or $t^{1''}_{\llcorner}=(1'', d,c,1)$ (light blue in Fig.~\ref{fig:11111001}) that transfers the value from the west to the north.
Observe in Fig.~\ref{fig:11111001} that  these tiles form a light blue diagonal bisector between the column of the most significant bit below it, and the copied row tile number of $N$. 
From the time that we pass these tiles, the values of each tile are copied to the north using $t_{\uparrow}$-type tiles. 
Note that the tile types of the $(n+k)$-th column are placed by matching with their southern and eastern tiles. 

This process continues until the value of each tile bit is placed in a new tile bit in the same row as $N$ but with a shift of $n$ columns.
In the terminal assembly, there is a copy of of $N$ (which was in positions $[(x,y), ..., (x+k, y)]$) in positions $[(x-k-1,y),\ldots, (x-1, y)]$.
\end{proof}

Therefore,  the U-turn System makes a copy of a number from position $[(x,y), ..., (x+k, y)]$ to position $[(x-k-1,y) , ..., (x-1, y)]$.

\subsection{The middle finding system}
Let $R$ be an \emph{explicitly bounded rectangle} on $\mathbb{Z}^2$ i.e. a rectangle with horizontal sides that are bounded by specially marked tiles. The following lemma uses the IBC, DBC and U-turn systems to present a SFTAM TAS in order to find the middle of the height of $R$. See Fig.~\ref{fig:mfs}.

\begin{lemma}[Middle finding system lemma]\label{mfs}
Let $R$ be an explicitly bounded rectangle of height $N$ and width at least $3\log (N)$. There is a TAS $S_{1/2}=(\Sigma, T_{1/2}, \sigma_{+}, str, \tau)$ such that for all assemblies with a seed located at the start, a tile of type $t_m$ appears at coordinate $(x,\lfloor\frac{N}{2}\rfloor)$ with no other tiles to its left, and $t_m$ does not appear anywhere else.
Formally the system is:
\begin{itemize}
		\item$\Sigma= \Sigma_{IBC1} \cup \Sigma_{IBC2} \cup \Sigma_{U} \cup \Sigma_{DBC}  \cup  \{  org , c\}$ 
		where the labels in $\Sigma$ are the same as in each system, except when a modification is explicitly mentioned. However, any two labels from two distinct systems are distinct from each other.
		\item $T_U= T_{IBC1} \cup  T_{IBC2} \cup T_{U} \cup T_{DBC}   \cup  \{  t_v, t_{copy}  \} $
		\item The seed is $\sigma_{+}$ from the IBC System.
		\item $ str: \Sigma \cup \{ \epsilon \} \longrightarrow \{0,1,2\} $  such that $str( \epsilon) = 0$ and $str(org)=str(c)=1$.
		The values of the str function in $\Sigma$ are the same as each system.
		\item $\tau=2$
\end{itemize}
\end{lemma}	

\begin{proof}
Let $R$ be an explicitly bounded rectangle and $N=2^n+k$ with $k < 2^n$ be the  height of $R$. Without loss of generality, we assume that the specially marked horizontal sides of $R$ being "Start" at the bottom and the other, "Finish", at the top, as in Fig.~\ref{fig:mfs}.

We use the IBC, DBC and U-turn systems from the previous sections to define our middle finding SFTAM TAS that finds $\frac{N}{2}= 2^{(n-1)}+\frac{k}{2}$. There are two copies of the IBC, named $IBC1$ and $IBC2$.
To avoid any confusion regarding tile types of IBC, DBC and U-turn systems each label is marked with its system name.
For example, we use $1_{IBC}$ for the IBC system, and $1_{DBC}$ for the DBC system. 
However, whenever there is no ambiguity,  we only use label $1$ for simplicity.
See an overview of the assembly in  Fig.~\ref{fig:mfs} to follow the structure of the proof. 
Now we describe the steps for the middle finding system.

  \medskip

\noindent\emph{0. Growing  a column of tiles of type $t_{++}$ until "Finish".}

The assembly starts to grow when $\sigma_{+} $ from the $IBC1$ System is placed besides the starting tiles. The tiles of type $t_{++}$ pave the way until becoming blocked by "Finish" tiles, creating a column of tiles. These tiles are the support for the $IBC1$ system.

  \medskip

\noindent\emph{1. Using the IBC System $IBC1$ to find the height $N=2^n+k$ of $R$.}

Recall from Definition~\ref{IBC} of IBC systems, that each row tile number represents the height of the assembly along the $y$ axis.
At the end, the assembly is blocked at the endpoint tiles of $R$ in the row tile number $N$. 

Here we modify the labels of the most significant bit tiles.
We set $t_{S}=(1_S, 0, 1, 1)$ instead of $t_{S}=(1_S, 0, 1, \epsilon)$,  i.e the last label is changed to $1$. This change will be used for starting  the growth of the second IBC system.
In addition, instead of $t_{O}=(1_s, over, \epsilon, \epsilon)$ we have $t_{O}=(1_s, over, org, org)$ ("org" stands for "orange"). The reason for this modification is that it will be useful in the $3$rd step where we want to copy the value of $k$.

\medskip

Here by this new modifications,  immediately two types of tiles will grow. These tiles  will be used later.
Firstly, each time that in the IBC system a new most significant bit tile is added, a vertical tile bond grows in its south by $t_v=(org,1, org,1)$, where org is the new southern label of $t_O$, and the label 1 is $1_{copy}$. (They will be used in the copy system.)
Secondly, when $t_v$ tiles collide with previous $t_v$ tiles, a double column of tiles of type $t_b$ appears. (They will be used as block tiles in the U-turn system.) 

 \medskip

\noindent\emph{2. Returning to row number $2^n$ using the second IBC system $IBC2$, which then outputs the value of $k$ }

Afterwards, at the collision of $IBC1$ with the finish block tiles, $IBC2$ appears next to the most significant bit tile of $IBC1$. 
This is where our modification of $IBC1$ becomes useful. In fact, changing the fourth label of the most significant bit tile of the $IBC1$ system to $1$, enables this new IBC system $IBC2$ to start growing, using the most significant column of $IBC1$ as its support.

Note that due to the SFTAM rotating tiles, the two IBC systems are able to grow in opposite directions. 
Thus, the assembly returns towards the row tile number of $2^n$ using the $IBC2$ system.

consequently, the value of the last row tile number of $IBC2$ is $k$.
 \medskip

\noindent\emph{3. Copying $k$ until $2^{n-1}$ (the middle of $2^n$).}

Now, the topmost column of tiles of type $t_v$ from the first step will act as the support for the third step.
Thanks to the vertical tile bond that is located below the most significant bit tile of the $2^n$ row tile number, the $k$-row tile number is copied $2^{n-1}$ times until the $2^{n-1}$-row tile number of $IBC1$  (this is half of $2^n$).
To this aim, we use tiles of types $t_{copy}=(x, copy,x,copy)$ and $t_{copy}^*=(x, copy,x,1_{DBC})$ (for $x\in\{0,1\}$) such that $x$ is the label of the IBC system, and $t_{copy}^*$ is used for copying the most significant bit tiles of the IBC system.

Since the column of tiles of type $t_v$ is the support of this step, the copy process stops after $2^{n-1}$ copies.

 \medskip

\noindent\emph{4. Halving $k$ by eliminating its least significant bit.}

Now, the copied row tile number with value $k$ is halved by getting rid of its least significant bit. The block tiles that already appeared when two vertical tile bonds of type $t_v$ met each other, discard the least significant bit tile of the $k$-row tile number and copy the $\frac{k}{2}$-row tile number to the left of the first block tiles. The tiles here are a tile of type $t_{b'}=(org,c,\epsilon,u)$ such that $u$ is a label with strength~2, and a seed $\sigma_u$ of the U-turn system that binds to it. These are block tiles. Also, from the south of $\sigma_u$, there are tiles of type $t_u$, the support tiles of the U-turn system. They are shown in black in Fig.~\ref{fig:mfs}.
    
  \medskip

\noindent\emph{5. Shifting $\frac{k}{2}$ to the left by a U-turn system.}

The column of tiles of type $t_u$ from the previous step form the support for this step. The binding of the block tiles to the $\frac{k}{2}$-row tile number form a U-turn block seed and the $\frac{k}{2}$-row tile number shifts to its left.
Remember that the U-turn system needs space for the $\frac{k}{2}$ tiles in the south of the U-turn seed; this condition is true at the row number $2^{n-1}$ along the y axis, since $\frac{k}{2} < 2^{n-1}$.

 \medskip

\noindent\emph{6. Going up by $\frac{k}{2}$ using the DBC system.}

After shifting $\frac{k}{2}$ to the left at the row $\frac{2^n}{2}$, a DBC system starts from $\frac{k}{2}$ along the side of the Copy system. Note that here the label $1_{DBC}$ of $t^*_{copy}$ (from the Copy system) plays the role of the support tile type in this DBC system. Moreover, the first row of the DBC system i.e. the seed, is the copied $\frac{k}{2}$ row tile number from the U-turn system. 
We add two final modifications to the DBC. We change the tile type $t_{0-1}^*$ to $t_m=(1', 0, 1', --)$ by changing the fourth label.

After passing the zero with the DBC system, a tile of type $t_m$ appears in row $\frac{N}{2}=2^{n-1}+\frac{k}{2}$.
The process of finding the middle finishes here by the appearance of $t_m$ in the assembly. Note that $t_m$ only appears once because the support of the DBC system is the left side of the copy system.

For an illustration, see Fig.~\ref{fig:mfs} (where $t_m$ is in red).

\end{proof}

	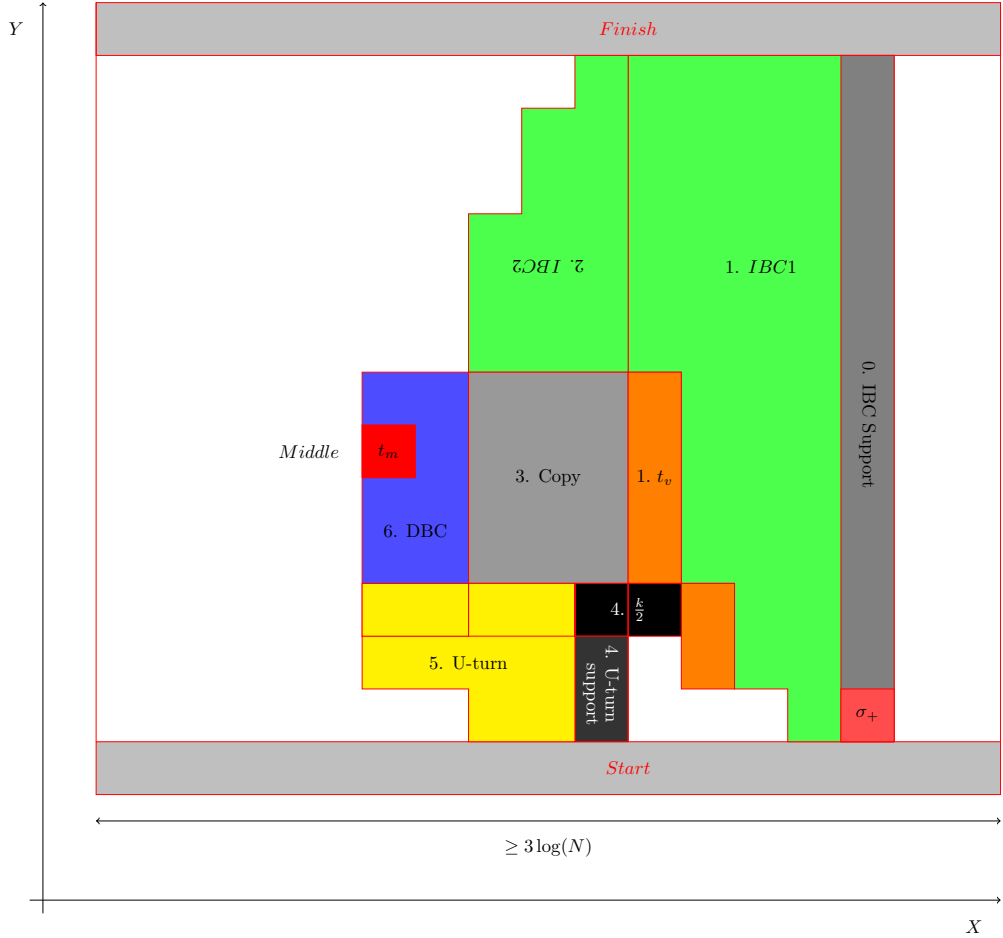
\begin{figure}[h!]
		\centering
		\scalebox{0.7}{\begin{tikzpicture}[node distance=7mm]
		
		\filldraw[fill=gray!50, draw=red]	(2,-1)--(2,0)--(-15,0)--(-15,-1)--(2, -1);
		\node[inner sep=2pt, label]  at (-5,-0.5){$\textcolor{red}{Start}$};
		
			\filldraw[fill=gray!50, draw=red]	(2,13)--(2,14)--(-15,14)--(-15,13)--(2, 13);
			
		\node[inner sep=2pt, label]  at (-5,13.5){$\textcolor{red}{Finish}$};
		
			\draw[red] (-15,0)--(-15,14);
			\draw[red] (2,0)--(2,14);
			
			\draw[->] (-16,-3.25)--(-16,14);
		\node[inner sep=2pt, label]  at (-16.5,13.5){$Y$};

			\draw[->] (-16.25,-3)--(2,-3);
			\node[inner sep=2pt, label]  at (1.5,-3.5){$X$};

\filldraw[fill=green!70, draw=red](0,0) --(0,13) --(-5,13) --(-5,7)--(-4,7)--(-4,3)--(-3,3)--(-3,1)--(-2,1)--(-2, 0)--(0,0);

	\filldraw[fill=gray, draw=red](0,0)--(-1,0)--(-1,13)--(0,13)--(0,0);
	\node[rotate=270,inner sep=2pt, label]  at (-0.5,6){0. IBC Support};

	\filldraw[fill=red!70, draw=red](0,0)--(-1,0)--(-1,1)--(0,1)--(0,0);

	\node[inner sep=2pt, label]  at (-0.5,0.5) {$\textcolor{black}{\sigma_{+}}$};

\node[inner sep=2pt, label]  at (-2.5,9) {\textcolor{black}{1. $IBC1$}};

\filldraw[fill=green!70, draw=red](-5,13)--(-6,13)--(-6,12)--(-7,12)--(-7,10)--(-8,10)--(-8,7)--(-5,7)--(-5,13);
\node[inner sep=2pt, label, rotate=180]  at (-6.5,9) {\textcolor{black}{2. $IBC2$}};

\filldraw[fill=gray!80, draw=red](-5,3)--(-8,3)--(-8,7)--(-5,7)--(-5,3);
\node[inner sep=2pt, label]  at (-6.5,5) {\textcolor{black}{3. Copy}};

\filldraw[fill=orange, draw=red, ](-5,7)--(-4,7)--(-4,3)--(-5,3)--(-5,7);

\node[inner sep=2pt, label]  at (-4.5,5) {$\textcolor{black}{1.~ t_v}$};

\filldraw[fill=orange, draw=red](-4,3)--(-3,3)--(-3,1)--(-4,1)--(-4, 3);

\filldraw[fill=black, draw=red](-4,3)--(-4,2)--(-5,2)--(-5,3)--(-6,3)--(-6,2)--(-5,2)--(-5,3)--(-4,3);
\node[inner sep=2pt, label]  at (-5,2.5) {\textcolor{white}{4. $\frac{k}{2}$}};

\filldraw[fill=black!80, draw=red](-6,2)--(-5,2)--(-5,0)--(-6,0)--(-6,3);
\node[rotate=270, inner sep=2pt, label]  at (-5.3,1.1) {\textcolor{white}{4. U-turn}};
\node[rotate=270, inner sep=2pt, label]  at (-5.7,0.9) {\textcolor{white}{support}};

\filldraw[draw=red](-5,0)--(-5,2);

\filldraw[fill=yellow, draw=red](-6,3)--(-6,2)--(-8,2)--(-8,3)--(-10,3)--(-10,2)--(-8,2)--(-8,3)--(-6,3)--(-6,0)--(-8,0)--(-8,1)--(-10,1)--(-10,2);
\filldraw[fill=yellow, draw=red](-10,3)--(-10,2)--(-8,2)--(-8,3)--(-10,3);
\node[inner sep=2pt, label]  at (-8,1.5) {\textcolor{black}{5. U-turn}};

\filldraw[fill=blue!70, draw=red](-10,3)--(-10,7)--(-8,7)--(-8,3)--(-10,3);
\node[inner sep=2pt, label]  at (-9,4) {\textcolor{black}{6. DBC}};

\filldraw[fill=red, draw=red](-10,5)--(-9,5)--(-9,6)--(-10,6)--(-10,5);

 \node[inner sep=2pt, label]  at (-9.5,5.5) {$\textcolor{black}{t_m}$};

 \node[inner sep=2pt, label]  at (-11,5.5) {$\textcolor{black}{Middle}$};

            \draw[<->] (-15,-1.5)--(2,-1.5);
            	        \node (log) at (-6.5,-2) {$\geq 3\log(N)$};
			\end{tikzpicture}}
		
		\caption{The steps of the middle finding system process, starting from the seed that is in red on the right. The tile $t_m$ (the  red tile on the left) appears in the middle of two rows of tiles that are shown by start and finish.}
		\label{fig:mfs}
	\end{figure}

\section{Distinguishing order-1 cuboids by their genus using SFTAM}\label{problem}

We now introduce a SFTAM tile assembly system named $\mathcal{S_G}$ which can distinguish the order-1 cuboids $C \in O_1^t$ (of genus 1) from the others in $O_1$ using the definitions and lemmas of Sections~\ref{def} and~\ref{lemma}.

Our main result is stated in the following theorem.

\begin{theorem}\label{maintheorem}
There is a SFTAM tile self-assembly system $\mathcal {S_G} = (\Sigma, T, \sigma, str, \tau)$ and a subset of tile-types $Y = \{t_{reg} , t_{mfs} \}\cup T_{ibc} \subseteq T$  such that for any order-1 cuboid $C=C_0 \setminus C'_0$ with the dimensions 
at least 10 for $C'_0$, if $\mathcal{S_G}$ assembles on $C$ starting from a seed which is placed in a normal placement, the following holds:

\begin{itemize}

\item if $C$ has genus $1$, every terminal assembly of $\mathcal{S}$ on $C$ contains at least one tile of $Y$, and

\item if $C$ has genus $0$, then no tile of $Y$ appears in any producible assembly of $\mathcal{S}$ on $C$.

\end{itemize}
\end{theorem}

There is a crucial point to consider: the system presented here works for the case where the assemblies' seed is placed on a \emph{normal placement} $p$, i.e. the place $p$ is ``far enough" from the borders on the surface of the cuboid.
Therefore first we define ``normal placement". 
Then, we explain a sketch of the construction of the assemblies before diving into the presentation of the system $\mathcal{S_G}$. 

\begin{definition}[Normal placement]
Let $C$ be an order-1 cuboid such that $(x_1, y_2, z_2)$,  $(x_2, y_2, z_2)$, ... ,$(x_n, y_n, z_n)$ are its vertices. A placement $p \in Pl(C)$ with position $(x, y, z)$ is a \emph{normal placement} of $C$ if and only if for all $i \in \mathbb{N}$, two of the following inequalities hold:
$\mid x_i - x \mid \geq 3\log (N)+6$, $\mid y_i - y \mid \geq 3\log (N)+6$ and $ \mid z_i - z \mid \geq 3\log (N)+6$, where $N$ is the largest of the three dimensions of the cuboid.
The set of all normal placements is denoted by $Pl_N(C)$.
\end{definition}

The simplest example to demonstrate the concept of normal placement is on a cuboid $C \in O_0$. 
In this case, normal placements consist of the cuboid's surface minus its ``frame" i.e. the border of the cuboid's edges with a thick margin. Hence there are $6$ disconnected areas on $C$'s faces where the normal placements are.
The normal placements on order-1 cuboids can be described similarly.
It should be noted that in this case there are more than $6$ disconnected areas. 
Recall the order-1 cuboid's type: the order-0 cuboid $(C\in O_0)$, the order-1 cuboid with a pit $(C\in O_1^p)$ or with a concavity $(C\in O_1^c)$ whose genus are $0$; and the order-1 cuboid with tunnel $(C\in O_1^t)$ whose genus is $1$. 
There are up to $9$ disconnected areas with normal placements for $C \in C_1^O$, $10$ for $C \in O_1^t$ and $11$ for $C \in O_1^p$, depending on the new faces that are generated by the nature of the order-1 cuboid's construction.
We note that, by definition, when the smallest dimension is large enough with respect to the others, almost all the placements on order-1 cuboids are normal placements. From now on, assume that the seed of assemblies in $\mathcal{S_G}$ is placed on a normal placement of order-1 cuboids.

\subsection{Region partition  on order-1 cuboids }\label{reg_part}

Let $C = C_0 \setminus C'_0$ be an order-1 cuboid.  In order to detect a potential tunnel whose entrances are on parallel faces, the construction separates these faces. For this purpose we use three planes, one for each pair of parallel faces of $C$, located between them. Let $P_X$, $P_Y$, $P_Z$ be three planes in this way: take $p\in Pl_N(C)$. The plane $P_X$ is passing on $p$ and is parallel to the plane formed by the $Y$-axis and $Z$-axis. The plane $P_Y$ is parallel to the plane formed by the $X$-axis and $Z$-axis  and is passing  on $p$. The plane $P_Z$ is parallel to the plane formed by the $X$-axis and $Y$-axis and contains the center of $C_0$. 
In Fig.~\ref{planes} the seed in yellow is in the point $p$  and the plane $P_X$, $P_Y$ and $P_Z$ are framed respectively by the ribbons $R_X$ (in red), $R_Y$ (in green) and $R_Z$ (in blue) on $C$. 
For $i \in \{ X, Y \}$,  $R_i $ is the connected component of $ \partial C \cap P_i $  that contains $p$.
If $R_X$ and $R_Y$ intersect in one point, $R_Z$ in the empty set. If they intersect in two points, $R'_Z = P_Z \cap \partial C$ and $R_Z$ is the connected component of $R'_Z$ that has an intersection with $R_Y$.
The difference $C \setminus \{ R_X, R_Y, R_Z \}$
consists up to $8$ connected components is called \textit{regions}. 
They are noted by $R_{XYZ}$ such that $X , Y , Z\in\{0, 1\}$ where $0$ represents the left, down and back sides, and $1$ represents the right, up and front sides. 
For example, $R_{101}$ refers to the region at the right, down and front side of $C$.
This way of partitioning $C$ helps to define the graph $G_C$, the \emph{region
graph} of $C$:
\begin{definition}[Region graph]
Let $ C = C_0 \setminus C_{0}^{'}$  be an order-1 cuboid with $p$ as a position it, the planes $P_X$, $P_Y$ two  perpendicular planes passing on $p$, and $P_Z$  a plane perpendicular on both planes passing through the middle of the $P_Y$. Also, let and  $R_Z$ for $R_i = \partial C_0 \cap P_i $ for $i \in \{ X, Y, Z \}$.  There is a graph assigned to $C$ named the  \emph{region
graph} $G_C(p)$ whose vertices are the regions separated by $R_X$, $R_Y$ and an edge is added between two regions if and only if they share  $P_X$, $P_Y$ or $P_Z$. 
\end{definition}
For an order-0 cuboid $C$, $G_C(p)$ is a bipartite graph and therefore it is $2$-colorable.
The region graph for an order-0 cuboid is presented in Fig.~\ref{fig:RegionGraph}. 

If $C$ be an order-1 cuboid with a tunnel, the number of disconnected regions can be less than~8 depending on intersection of the tunnel with
 three planes.  The \emph{axis} of the tunnel is the direction orthogonal to its entrances. The three planes can intersect the tunnel in two ways: along the width of the tunnel when the plane is perpendicular to the axis of the tunnel, or along the length of the tunnel when the plane is parallel to the axis of the tunnel. Thus, a tunnel may have an intersection with up to three perpendicular planes, one along the width, and up to two other planes along the length.
Based on this, three types of partitions into regions are possible and the possible numbers of regions are: 7 regions when one plane intersects along the width of the tunnel, 5 regions when one plane intersects along the length of the tunnel and one along the width (See Fig~\ref{fig:5region}), and 1 region when  three perpendicular planes intersect along the tunnel, one along the width and the others along the length.

\begin{figure}[h!]
    \centering

    \begin{subfigure}[t]{.45\textwidth}
      \begin{tikzpicture}[scale=2]
  \draw[thick](2,2,0)--(0,2,0)--(0,2,2)--(2,2,2)--(2,2,0)--(2,0,0)--(2,0,2)--(0,0,2)--(0,2,2);
  \draw[thick](2,2,2)--(2,0,2);
  \draw[gray](2,0,0)--(0,0,0)--(0,2,0);
  \draw[gray](0,0,0)--(0,0,2);
  \node[inner sep=2pt,circle,draw,fill=white, label]  at (0,0,0) {$\textcolor{black}{R_{000}}$};
   \node[inner sep=2pt,circle,draw,fill=white,label]  at (2,0,2) {$\textcolor{black}{R_{101}}$};
  \node[inner sep=2pt,circle,draw,fill=white,label]  at (0,2,2) {$\textcolor{black}{R_{011}}$};
   \node[inner sep=2pt,circle,draw,fill=white,label]  at (2,2,0) {$\textcolor{black}{R_{110}}$};
  
\node[inner sep=2pt,circle,draw,fill=white, label]  at (0,2,0){$\textcolor{black}{R_{010}}$};
\node[inner sep=2pt,circle,draw,fill=white, label]  at (0,0,2){$\textcolor{black}{R_{001}}$};
\node[inner sep=2pt,circle,draw,fill=white, label]  at (2,0,0){$\textcolor{black}{R_{100}}$};
\node[inner sep=2pt,circle,draw,fill=white, label]  at (2,2,2){$\textcolor{black}{R_{111}}$};
\end{tikzpicture}
      \caption{The region graph}
    \end{subfigure}\qquad
    \begin{subfigure}[t]{.45\textwidth}
      \begin{tikzpicture}[scale=2]
  \draw[thick](2,2,0)--(0,2,0)--(0,2,2)--(2,2,2)--(2,2,0)--(2,0,0)--(2,0,2)--(0,0,2)--(0,2,2);
  \draw[thick](2,2,2)--(2,0,2);
  \draw[gray](2,0,0)--(0,0,0)--(0,2,0);
  \draw[gray](0,0,0)--(0,0,2);
  \node[inner sep=2pt,circle,draw,fill=white, label]  at (0,0,0) {$\textcolor{black}{R_{000}}$};
   \node[inner sep=2pt,circle,draw,fill=white,label]  at (2,0,2) {$\textcolor{black}{R_{101}}$};
  \node[inner sep=2pt,circle,draw,fill=white,label]  at (0,2,2) {$\textcolor{black}{R_{011}}$};
   \node[inner sep=2pt,circle,draw,fill=white,label]  at (2,2,0) {$\textcolor{black}{R_{110}}$};
  
\node[inner sep=2pt,circle,draw,fill=black, label]  at (0,2,0){$\textcolor{white}{R_{010}}$};
\node[inner sep=2pt,circle,draw,fill=black, label]  at (0,0,2){$\textcolor{white}{R_{001}}$};
\node[inner sep=2pt,circle,draw,fill=black, label]  at (2,0,0){$\textcolor{white}{R_{100}}$};
\node[inner sep=2pt,circle,draw,fill=black, label]  at (2,2,2){$\textcolor{white}{R_{111}}$};

\node[line width=0.1mm,rectangle, minimum height=2mm,minimum width=2mm,fill=white!70,rounded corners=1mm,draw, label]  at (2,2,0) {$\textcolor{black}{t_{even}}$};
\node[line width=0.1mm,rectangle, minimum height=2mm,minimum width=2mm,fill=white!70,rounded corners=1mm,draw, label]  at (2,0,2) {$\textcolor{black}{t_{even}}$};
\node[line width=0.1mm,rectangle, minimum height=2mm,minimum width=2mm,fill=white!70,rounded corners=1mm,draw, label]  at (0,2,2) {$\textcolor{black}{t_{even}}$};
\node[line width=0.1mm,rectangle, minimum height=2mm,minimum width=2mm,fill=white!70,rounded corners=1mm,draw, label]  at (0,0,0) {$\textcolor{black}{t_{even}}$};

\node[line width=0.1mm,rectangle, minimum height=2mm,minimum width=2mm,fill=black,rounded corners=1mm,draw, label]  at (0,2, 0) {$\textcolor{white}{t_{odd}}$};
\node[line width=0.1mm,rectangle, minimum height=2mm,minimum width=2mm,fill=black,rounded corners=1mm,draw, label]  at (2,0, 0) {$\textcolor{white}{t_{odd}}$};
\node[line width=0.1mm,rectangle, minimum height=2mm,minimum width=2mm,fill=black,rounded corners=1mm,draw, label]  at (2,2, 2) {$\textcolor{white}{t_{odd}}$};
\node[line width=0.1mm,rectangle, minimum height=1mm,minimum width=2mm,fill=black,rounded corners=1mm,draw, label]  at (0,0, 2) {$\textcolor{white}{t_{odd}}$};
\end{tikzpicture}
\label{InnerTiling}
      \caption{Tile coloring of the region graph, and
filling up the regions via inner filling tiles of types $t_{even}$ and $t_{odd}$.}
    \end{subfigure}

\caption{The region graph $G_C$:
$V(G_{S_G})$ are the distinct regions, and
$E(G_{S_G})$ contains the edges between neighboring regions.}
\label{fig:RegionGraph}
\end{figure}
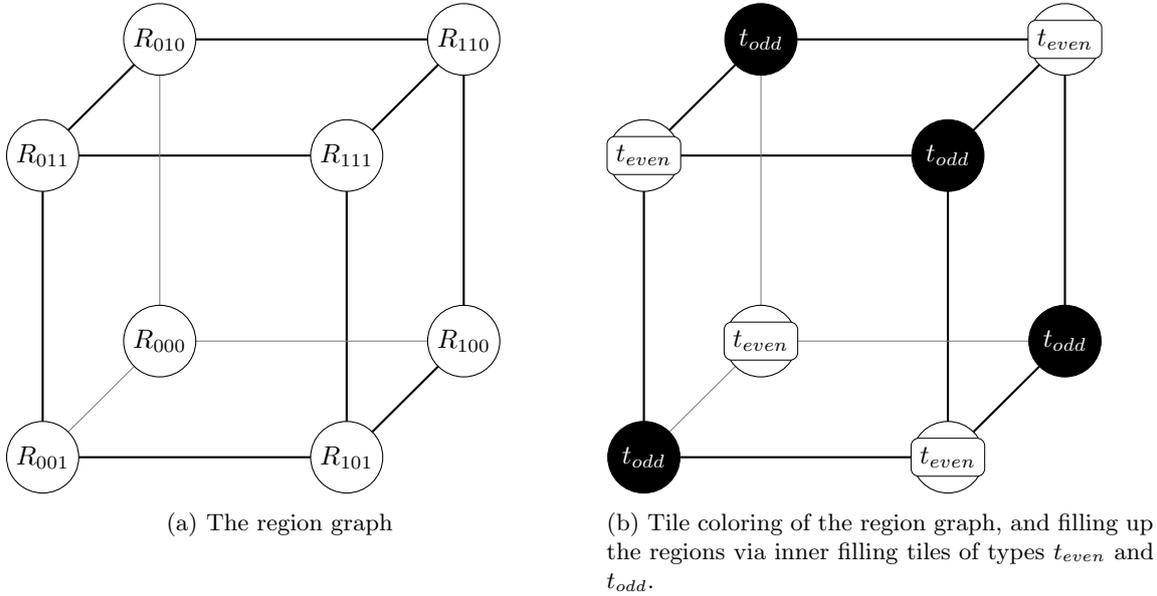

\subsection{Description of the TAS $\mathcal{S_G}$}

We now describe $\mathcal{S_G}$ formally.

\medskip

\noindent\emph{The tile types of $P_X$:}

\begin{itemize}

\item  The seed $\sigma=(\epsilon, \epsilon, x', \epsilon)$
\item  $x'=(x',x''_e,x,x''_w)$ and $x=(x, x_e, x, x_w)$
\item $x'_e=(x'_e, x'_e , x_e , x''_e)$ and  $x'_w=(x'_w, x''_w,  x_w, x'_w)$ 
\item $x_e=(x_e, x_e, x_e, x_e)$ and $x_w=(x_w , x_w , x_w , x_w)$

    \end{itemize}

	\begin{figure}[h!]
		\centering
        \includegraphics[scale=0.5]{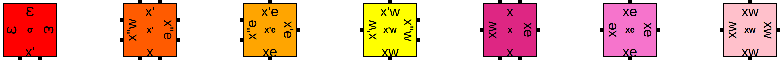}
		\caption{The tile types for the $R_X$.}
		\label{fig:tuiles_PX}
	\end{figure}

    \medskip

\noindent\emph{The tile types for the ribbons framing $P_Y$:}
 
    \begin{itemize}
   
    \item The starter tile types : 
    $y_e1=(x_e, y_e1, x'_e, \epsilon)$ and $y_w1=(x_w,\epsilon, x'_w, y_w1)$, 
    $y_e2=(\epsilon, ,y_e2,\epsilon, y_e1)$ and
    $y_w2=(\epsilon, y_w1,\epsilon, y_w2)$, 
    $y_e=(\epsilon, ,y_e,\epsilon, y_e2)$ and
    $y_w=(\epsilon, y_w2,\epsilon, y_w)$.

    \item The tile types of four middle finding systems (presented in \ref{mfs}) of $MFS^u_e$, $MFS^d_e$ (east) with shared seed $\sigma^+_e=(over,++,over,y_e)$ and support tiles of type $t_{++}=(1,++,1,++)$, and $MFS^u_w$, $MFS^d_w$ (west) with shared seed $\sigma^+_w=(over,y_w,over,++)$ and support $t_{++}$. Apart from that, each system has its own tile types.
        \end{itemize}

    	\begin{figure}[h!]
		\centering
		\begin{subfigure}[t]{0.5\textwidth}
		\includegraphics[scale=0.25]{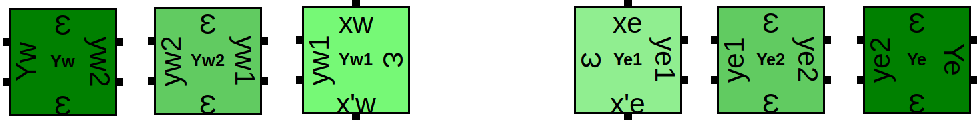}
		\caption{Starter tiles}
		\end{subfigure}\\
		\begin{subfigure}[t]{0.2\textwidth}
		\centering
		\hspace{1cm}\includegraphics[scale=0.5]{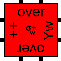}\hspace{1cm}
		\caption{West seed}
		\end{subfigure}\qquad
		\begin{subfigure}[t]{0.2\textwidth}
		\centering\hspace{1cm}\includegraphics[scale=0.5]{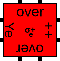}\hspace{1cm}
			\caption{East seed}
		\end{subfigure}\qquad
		\begin{subfigure}[t]{0.2\textwidth}
		\centering\hspace{1cm}\includegraphics[scale=0.5]{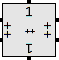}\hspace{1cm}
			\caption{Support tile}
		\end{subfigure}
		\caption{The starter tile types, seeds and support tile type for the ribbon framing $P_Y$.}
		\label{fig:tuiles_PY_starter}
	\end{figure}

        \medskip

\noindent\emph{The tile types of $P_Z$:}
    
    \begin{itemize}
   
   \item $t_{mr}=( 1, 1 ,0' , z"r)$
   , $z_{1r}=( ze , z"r , z"r , zr)$
   , $z_{2r}=(z"r , \epsilon , zo , zo)$
   , $z_{r}=(ze , zr , r , zr)$
   , $z_{or}=( r , zo , zo , zo)$.
   
   \item $t_{ml}=(1 , z"l , o', 1)$
   , $z_{1l}=( z"l , zl , ze , z"l)$
   , $z_{2l}=( zo , zo , zl", \epsilon)$
   , $z_{l}=( l , zl , ze , zl)$
   , $z_{ol}=( zo , zo , l , zo)$.

    \end{itemize}
    \begin{figure}[h!]
  \centering
  \begin{subfigure}[t]{\textwidth}
   \centering
    \includegraphics[scale=0.5]{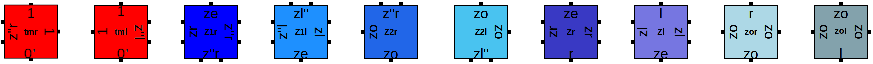}
    \caption{The tile types for the $R_Z$ ribbon.}
  \end{subfigure}\\
  \begin{subfigure}[t]{\textwidth}
   \centering
    \includegraphics[scale=0.5]{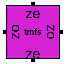}\hfill
    \caption{
    The tile type $t_{mfs}$, which may only appear if $C$ has genus~1, when the $R_Z$ ribbons meet each other.}
  \end{subfigure}
  
  \caption{The tile types for the ribbon framing $R_Z$.}
  \label{fig:tuiles_PZ}
\end{figure}

        \medskip

\noindent\emph{The tile types of inner filling:}
\begin{itemize}
\item  $t_{ee}=(z_{e1},x_{ev},z_e,x_e)$,
 $t_{ee1}=(z_{e2},\epsilon,z_{e1},x_e)$,
 $t_{ee2}=(z_{e3},\epsilon,z_{e2},x_e)$,
 $t_{ee3}=(z_{e4},\epsilon,z_{e3},x_e)$,
 $t_{ee4}=(z_{e},\epsilon,z_{e4},x_e)$.
\item $t_{eo}=(z_o,x_{od},z_{o1},x_e)$,
$t_{eo1}=(z_{o1},\epsilon,z_{o2},x_e)$,
 $t_{eo2}=(z_{o2},\epsilon,z_{o},x_e)$.
\item$t_{we}=(z_e,x_{w},z_{e1},x_{ev})$, 
 $t_{we1}=(z_{e1}, x_w,z_{e2},\epsilon)$,
 $t_{we2}=(z_{e2}, x_w,z_{e3},\epsilon)$,
 $t_{we3}=(z_{e3}, x_w,z_{e4},\epsilon)$,
 $t_{we4}=(z_{e4}, x_w,z_{e},\epsilon)$.
\item $t_{wo}=(z_{o1},x_{w},z_o,x_{od})$,
$t_{wo1}=(z_{o2}, x_w,z_{o1},\epsilon)$,
 $t_{wo2}=(z_{o}, x_w,z_{o2},\epsilon)$.
\item $t_{odd}=(z_o,x_{od},z_o,x_{od})$ and $t_{even}=(z_e,x_{ev},z_e,x_{ev})$.
        
    \end{itemize}
    
     \medskip

\noindent\emph{The tile types which demonstrate a tunnel on the cuboid:}
    \begin{itemize}
        \item $t_{reg}=(z_e,\epsilon,z_o,\epsilon)$
        \item $T_{ibc}$ is not a single tile type, but it is a set of tile types: $T_{ibc}=\{(i_1,i_2,i_3,i_4)$\} where $i_1,i_3\in\{0,1,over\}$, $i_2$ is a label from $IBC1^u_e \cup IBC1^d_e$ and $i_4$ is a label from  $IBC1^u_w \cup IBC1^d_w$ systems of the first part of the middle finding systems. 
        \item $t_{mfs}=(z_e,z_o,z_e,z_o)$
    
    \end{itemize}

\subsection{Overview of the assemblies of $\mathcal{S_G}$ on $O_1$}

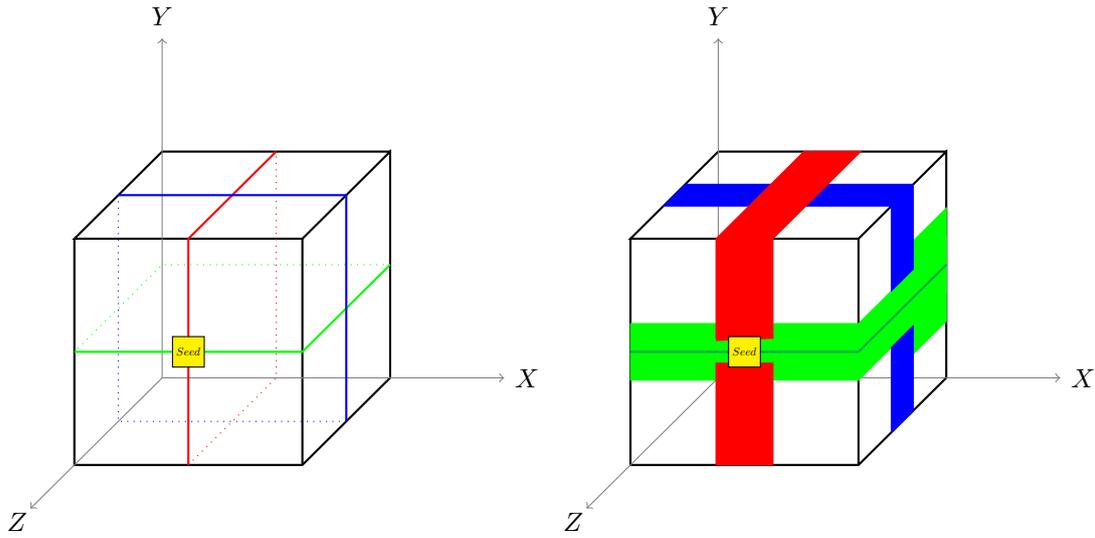
\begin{figure}[h!]
  \centering
  \begin{subfigure}[t]{.45\textwidth}
    \begin{tikzpicture}[scale=1.5]
      \draw[thick](2,2,0)--(0,2,0)--(0,2,2)--(2,2,2)--(2,2,0)--(2,0,0)--(2,0,2)--(0,0,2)--(0,2,2);
      \draw[thick](2,2,2)--(2,0,2);
      \draw[gray,->](0,0,0)--(3,0,0);
      \draw[gray,->](0,0,0)--(0,3,0);
      \draw[gray,->](0,0,0)--(0,0,3);

      \node at (3.2,0,0) {$X$};
      \node at (0,3.2,0) {$Y$};
      \node at (0,0,3.3) {$Z$};

      \draw[red,thick](1,2,0)--(1,2,2)--(1,0,2);
      \draw[red, dotted](1,0,2)--(1,0,0)--(1,2,0);
      
      \draw[green,thick](0,1,2)--(2,1,2)--(2,1,0);
      \draw[green, dotted](2,1,0)--(0,1,0)--(0,1,2);
      
      \draw[blue,thick](0,2,1)--(2,2,1)--(2,0,1);
      \draw[blue, dotted](2,0,1)--(0,0,1)--(0,2,1);
      
      \node[inner ysep=10pt,rectangle,draw,fill=yellow,scale=0.15mm,label]  at (1,1,2) {$\textcolor{black}{Seed}$}; 
      
    \end{tikzpicture}
  \end{subfigure}
  \begin{subfigure}[t]{.45\textwidth}
    \begin{tikzpicture}[scale=1.5]
      \draw[thick](2,2,0)--(0,2,0)--(0,2,2)--(2,2,2)--(2,2,0)--(2,0,0)--(2,0,2)--(0,0,2)--(0,2,2);
      \draw[thick](2,2,2)--(2,0,2);
      \draw[gray,->](0,0,0)--(3,0,0);
      \draw[gray,->](0,0,0)--(0,3,0);
      \draw[gray,->](0,0,0)--(0,0,3);

      \node at (3.2,0,0) {$X$};
      \node at (0,3.2,0) {$Y$};
      \node at (0,0,3.3) {$Z$};

      \filldraw[fill=green,draw=green] (0,1.25,2)--(2,1.25,2)--(2,1.25,0)--(2,0.75,0)--(2,0.75,2)--(0,0.75,2)--(0,1.25,2);
      \filldraw[fill=green,draw=green] (2,1.5,1)--(2,1.5,0)--(2,0.5,0)--(2,0.5,1)--(2,1.5,1);
      \draw[Green,thick](0,1,2)--(2,1,2)--(2,1,0);

      \filldraw[fill=blue,draw=blue] (2,1.25,0.75)--(2,2,0.75)--(0,2,0.75)--(0,2,1.25)--(2,2,1.25)--(2,1.25,1.25)--(2,1.25,0.75);
      \filldraw[fill=blue,draw=blue] (2,0.75,0.75)--(2,0.75,1.25)--(2,0,1.25)--(2,0,0.75)--(2,0.75,0.75);

      \filldraw[fill=red,draw=red] (0.75,0,2)--(1.25, 0,2)--(1.25, 0.90,2)--(0.75, 0.9,2) --(0.75,0,2);
      \filldraw[fill=red,draw=red] (0.75,1.1,2)--(1.25,1.12,2)--(1.25,2,2)--(1.25,2,0)--(0.75,2,0)--(0.75,2,2)--(0.75,1.12,2);

      \node[inner ysep=10pt,rectangle,draw,fill=yellow,scale=0.15mm,label]  at (1,1,2) {$\textcolor{black}{Seed}$}; 
      
    \end{tikzpicture}
  \end{subfigure}

  \caption{
   The skeleton of a terminal  assembly of $\mathcal{S_G} $ on an order-0 cuboid starting from a seed (in yellow) in a normal placement.
     On the left, the traces of the ribbons $R_X$ (in red), $R_Y$ (in green) and $R_Z$ (in blue). On the right, the shape of the skeleton on the cuboid.
 }
  \label{planes}
\end{figure}

Let $C$ be an order-1 cuboid. An assembly of $\mathcal{S_G}$  starts from a seed in an arbitrary normal placement on $C$.
In the TAS $\mathcal{S_G}$, the seed acts like a compass for the assemblies.  Without loss of generality, we assume that the side on which the seed is located is the face parallel to the $XY$-plane and intersects the $Z$ axis, and the north label of the seed's tile points towards the $Y$ axis.
The process of assemblies' growth in $S_G$ consists two phases, a phase for forming a skeleton and a phase for filling up the skeleton:

\begin{enumerate}
    \item  Constructing  the skeleton of the assembly's structures by at most $7$ perpendicular ribbons on $C$.
    Here, the planes  $P_X$, $P_Y$ and $P_Z$ are located from being framed by several ribbons  of tiles ($R_X$, $R_Y$ and $R_Z$ ) during the assembly and each step starts only when the previous step is finished.
    \begin{itemize}
        \item $R_X$ including one ribbon for framing the first plane $P_X$
        \item $R_Y$ including two ribbons for framing the second plane $P_Y$
        \item $R_Z$ including zero or four ribbons constitute the frame of the third plane $P_Z$ (depending on the intersection of the two previous planes, details will be given later)
    \end{itemize}

    \item  Filling the inside of the assembly's skeleton by distinctive tiles. In this step  the interior of the regions is partially filled by their distinctive tiles in a way that no connected component has a neighbor with the same inner filling tile. 
\end{enumerate}

In order to simplify the explanation of the process of the assemblies, first phase one is presented:
\begin{itemize}
    \item how the skeleton grows depends on the placement of the seed
    \item how the skeleton partitions $C$ into distinct connected components
    \item what its assigned region graph is.
\end{itemize}
Next, we study the phase of inner filling. Afterwards, we conclude the proof of the main theorem.

\subsubsection{Terminal assemblies on order-0 cuboids: $ A^{C_0}_\square [\mathcal{S_G}]$ }\label{genus0}

For the study of the shape of the productions in $O_1^t$, the productions on $O_0$ will be useful as a reference. We show that $\mathcal{S}_\mathcal{G}$ partitions cuboids into eight distinct regions as presented in Section~\ref{reg_part}. Later in the next section we study the case of order-1 cuboids. 

\medskip

\noindent\emph{1. The structure of the skeleton.}
\begin{lemma}\label{skeleton}
Let $C \in O_0$ be an order-0 cuboid and assume that the seed $\sigma$ is placed at a normal placement $p \in Pl_N(C)$. Every terminal assembly of $\mathcal{S_G}$ on $C$ includes a "3-step skeleton" noted by $R_X \cup R_Y \cup R_Z$ where each part is located on the corresponding ribbon defined in Section~\ref{reg_part}.
\end{lemma}

\begin{proof}

 The assembly start from the seed $\sigma=(\epsilon, \epsilon, x', \epsilon)$ has label $x'$ with strength~2 at the south, and its other labels are $\epsilon$ with strength~0. In the first step tiles make a vertical segment ribbon of tiles around $C$ as the $R_X$, starting from the south of the seed and finishing at its north. More precisely, a tile of type $x'=(x', x''_e , x , x''_w)$ sits at $\sigma$'s south and the label $x$ gives rise to a vertical ribbon of tiles of type $x=(x, x_e, x, x_w)$ such that it grows around $C$ and returns to the seed from the north. 
Moreover, two tiles of types $x'_e=(x'_e, x'_e , x_e , x''_e)$ and  $x'_w=(x'_w, x''_w,  x_w, x'_w)$ come at the east and the west of $x'$. By supporting the tile of type $x$, ribbons of tiles of types $x_e$and $x_w$ form respectively in the south of $x'_e$ and $x'_w$. 
Now, there is a ribbon made of three columns of tiles
that build the  $R_X$. In Fig.~\ref{fig:PX-ribbons} the seed in red is in the center of $R_X$.
The  $R_X$ divides $C$ into two regions, the \emph{right side} and the \emph{left side} of $C$ with respect to the $\sigma$. (We always assume that we view the cuboid from the point of view of the $Z$-axis, as in Fig.~\ref{planes}, and thus \emph{left}, \emph{right}, \emph{up}, \emph{down}, \emph{back}, \emph{front} refer to this point of view.)

\begin{figure}[h!]
    \centering
    \includegraphics[scale=0.4]{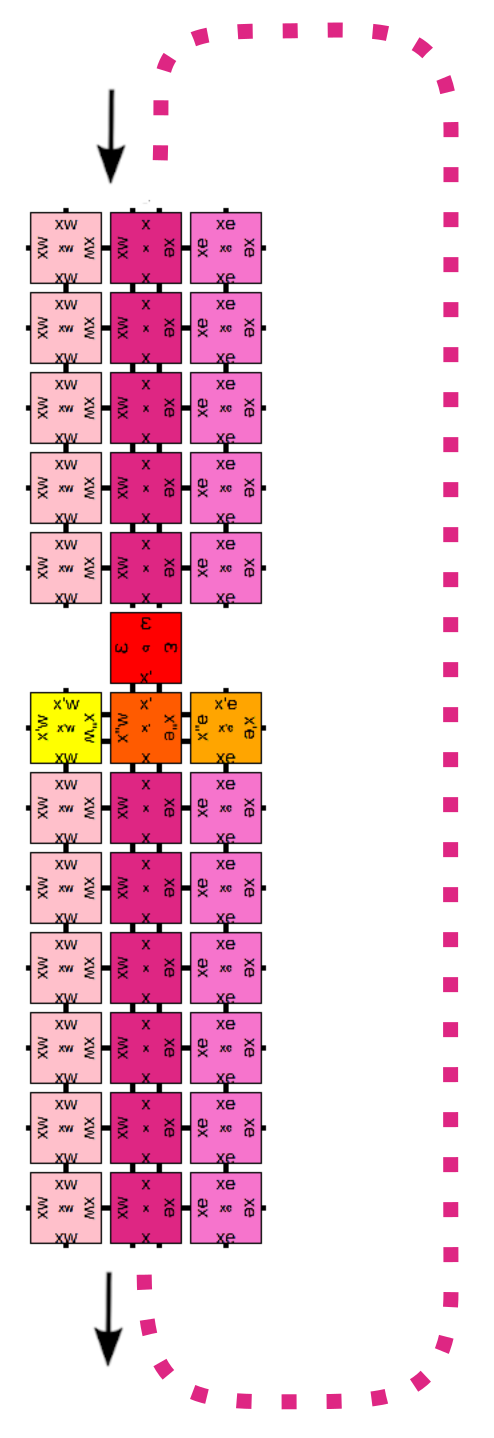}
    \caption{The ribbon $R_X$. The assembly starts from the south of the seed (in red at the center) and rounds the order-1 cuboid.}
    \label{fig:PX-ribbons}
\end{figure}

Therefore, the  $R_Y$ start to form only when $R_X$ rebounds to the north of the seed using tiles of type  $y_e1$ and $y_w1$.
Two segment ribbons starting from both right and left sides of $\sigma$ develop perpendicular to $R_X$ by using two middle finding systems (Definition~\ref{mfs}), one on each side.  
To elaborate, the tiles of type $y_e1=(x_e, y_e1, x'_e, \epsilon)$ and $y_w1=(x_w,\epsilon, x'_w, y_w1)$ sit in the assembly respectively between tiles of type $x'_e$ and $x_e$ at the east of $R_X$, and between tiles of type $x'_w$ and $x_w$ at the west of $R_X$. 
The tiles of type $y_e2=(\epsilon,y_e2,\epsilon, y_e1)$ and $y_w2=(\epsilon, y_w1,\epsilon, y_w2)$ bind to $y_e1$ and $y_w1$; and the tiles of type $y_e=(\epsilon, ,y_e,\epsilon, y_e2)$ and $y_w=(\epsilon, y_w2,\epsilon, y_w)$ bind to the tiles of type $y_e2$ and $y_w2$, respectively.
The tile types $y_e$ and $y_w$ are starting points for the eastern and western two-side IBC systems (Definition~\ref{IBC}), that is,  the seeds $\sigma^+_e=(over,++,over,y_e) $ and $\sigma^+_w=(over,y_w,over,++)$ make four IBC systems $IBC1^u_e$, $IBC1^d_e$ (east), and $IBC1^u_w$, $IBC1^d_w$ (west) with the support $(1, ++, 1, ++)$. See Fig.~\ref{fig:PX-PY} for an illustration of growing of $R_Y$ after finishing $R_X$. Note that each system has its own distinguished tile types.
\begin{figure}[h!]
    \centering
    \includegraphics[width=\textwidth]{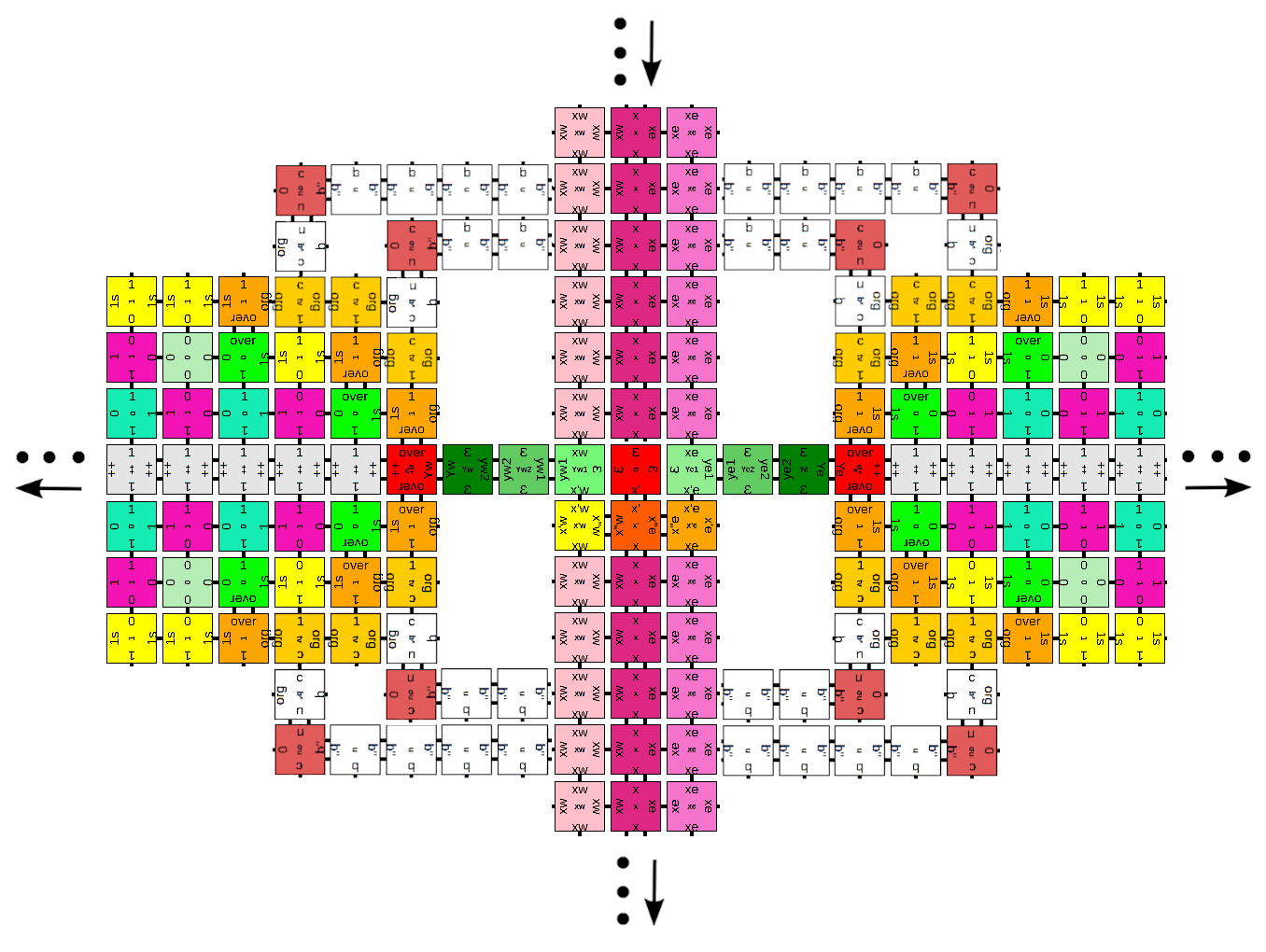}
    \caption{The formation of $R_Y$ (horizontal), out of $R_X$ (vertical). The initial seed of the assembly is in the center(in red). When $R_X$ is finished, from the est and the west of the seed the  two-side middle finding  systems grow to form $R_Y$.}
    \label{fig:PX-PY}
\end{figure}
After that, when the two-side $IBC1$ ribbons of the plane $P_Y$ meet the $P_X$ ribbon, the tiles of types $x_e$ and $x_w$ of $P_X$ act as finish block tiles for the middle finding system.

Then, the four tiles of types 
$t_{eu}=(1_s,\epsilon, x_e, 1)$, 
$t_{ed}=(1_s, 1 , x_e,\epsilon)$, $t_{wu}=(1_s,1, x_w,\epsilon)$ and  $t_{wd}=(1_s, \epsilon, x_w,1)$ respectively, appear as the row tile number of $1$ of the second IBC systems $IBC2^u_e$, $IBC2^d_e$ (east), and $IBC2^u_e$, $IBC2^d_e$ (west) of the middle finding systems. Note that the tiles of the second $IBC2$ systems have different labels from the ones of the $IBC1$ system (except the exterior labels of the most significant bit tiles in the $IBC1$ systems).
On the same underlying rectangle as the underlying rectangle of $R_Y$, the ribbons grow and they stop in the middle of $R_Y$'s ribbons.

Once the $R_Y$ ribbons form, they separate $C$ into an \emph{up side} and a \emph{down side}. Thus, $C$ is now partitioned into four separate regions due to the first and second step ribbons.

Next, by finding the middle of each of $R_Y$'s ribbons on the right and left faces, four new perpendicular ribbons are generated from the tiles of type $t_{ml}=(1 ,z''_l,0',1)$ at the left-up and right-down sides, and $t_{mr}=(1 ,1,0',z''_r)$ at the right-up an left-down sides. They are the replacements of the tile of type $t_m$ of the middle finding systems. See Fig.~\ref{fig:PZZ-2ribbons} for an illustration.
They go on, until they reach $R_X$ on the upper and down faces (see Fig.~\ref{fig:PZ-PX}). The union of these four ribbon forms a frame for the plane $R_Z$. 
This step creates a separation between the \emph{front side} and the \emph{back side} of the cuboid $C$ with respect to $\sigma$.
Note that all the ribbons are able to form completely because by our hypothesis, the seed is located on a normal placement. This is an essential condition since we need enough space for using counters in the assembly.
We show the detailed assembly of $R_X$, $R_Y$ and $R_Z$ in Fig.~\ref{fig:PX-PY-PZ} that starts from the seed $\sigma$ (in red) at the left of the Figure where $R_X$ rebounds on $C$. The western two-side middle finding system and its assign parts of $R_Z$ is omitted for the sake of brevity, however they are the mirror image of the eastern ones. 

\begin{figure}[h!]
    \centering
    \includegraphics[scale=0.5]{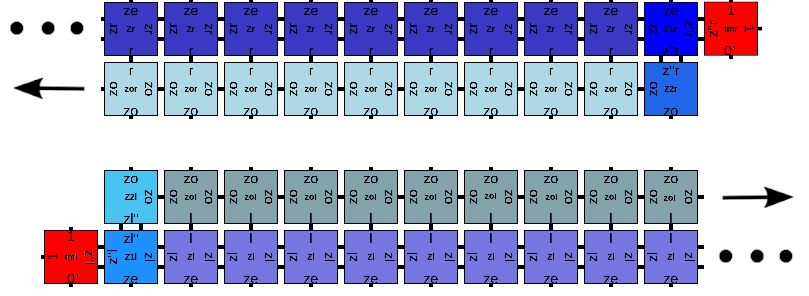}
    \caption{Two ribbons of $R_Z$.}
    \label{fig:PZZ-2ribbons}
\end{figure}

\begin{figure}[h!]
    \centering
    \includegraphics[width=\textwidth]{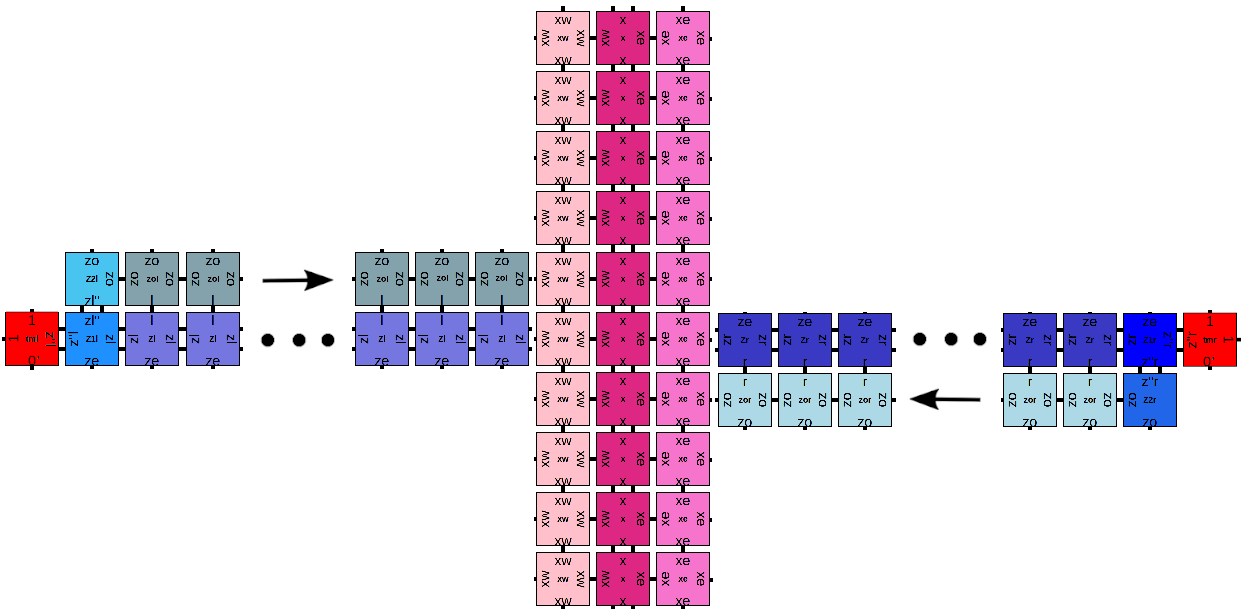}
    \caption{Two ribbons of $R_Z$ meeting the ribbons of $R_X$.}
    \label{fig:PZ-PX}
\end{figure}

\begin{figure}[h!]
    \centering
    \includegraphics[width=\textwidth]{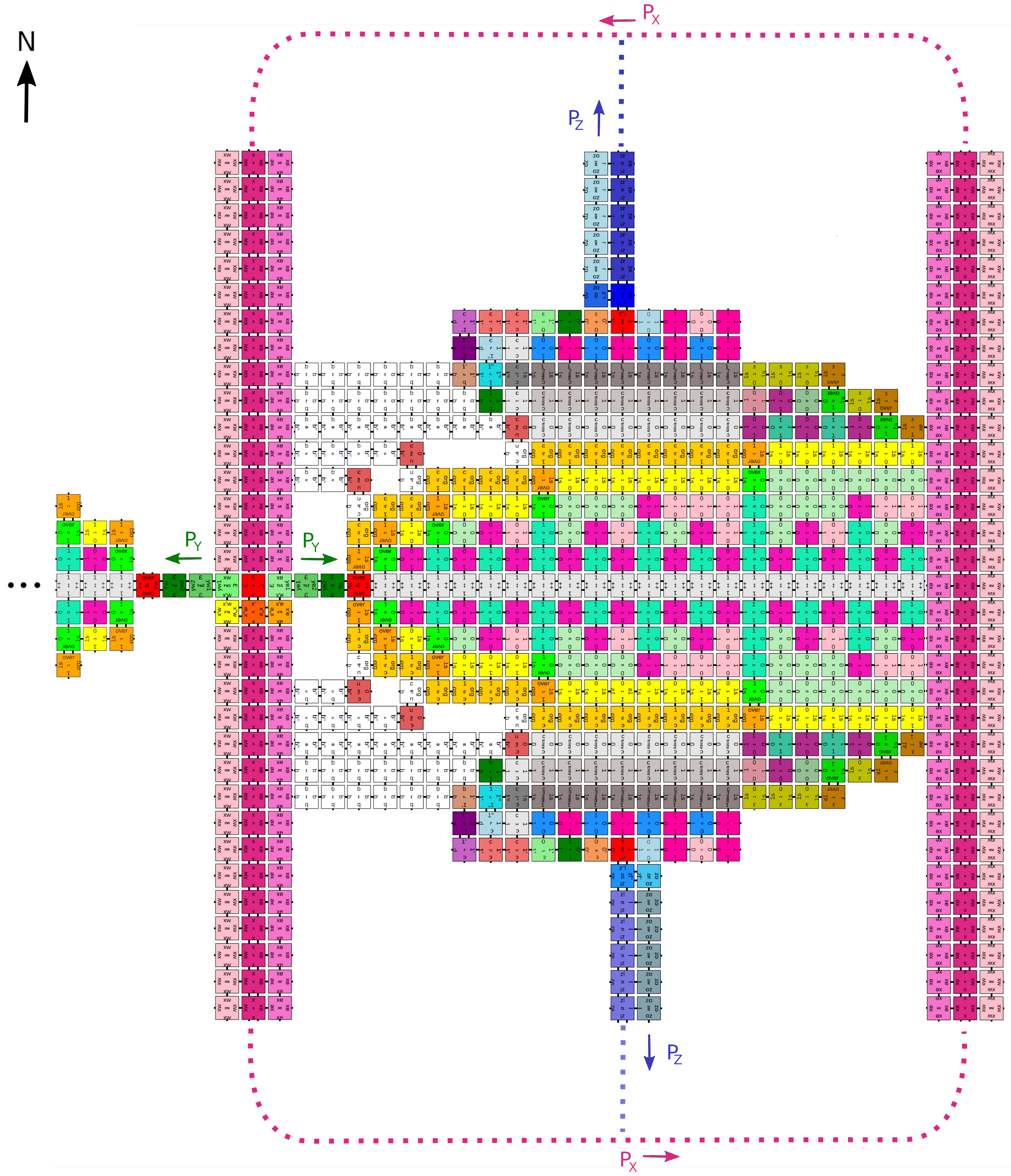}
    \caption{The assembly of $R_X$, $R_Y$ and $R_Z$ on an order-0 cuboid. The seed is located in the middle of $R_X$ at left. $R_X$ grows from the south of the seed and finishes at its north. Then, $P_Y$ starts growing by two two-side eastern and western middle finding systems. At the end, $P_Z$ starts to assemble from the found middle tile of $R_Y$ (in red) and finishes by arriving at $P_X$. Note that The western two-side middle finding system and its assigned parts of $R_Z$ are omitted for the sake of brevity, however they are the mirror image of the eastern ones.}
    \label{fig:PX-PY-PZ}
\end{figure}

Similar to Section~\ref{reg_part}, these three steps partition $C$ into $8$ distinct regions. 
\end{proof}

\medskip

\noindent\emph{2. Inner filling of the skeleton.}

After the formation of the skeleton, the second phase is to fill the eight regions by lines of interior tiles, once the $R_Z$ ribbons reach $R_X$. 
Since the region graph of an order-0 cuboid $C$ is a  $2$-colorable graph, we use two tile types to distinctly tile $C$. 
\begin{figure}[h!]
    \centering
    \includegraphics[scale=0.4]{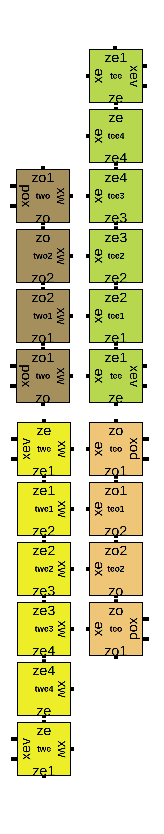}
    \caption{The tile types of four inner ribbons at the intersection of $R_X $ and the $R_Z$. }
    \label{fig:innerfillingtiles}
\end{figure}

\begin{lemma}
 Let $C\in O_0$ be an order-0 cuboid. For all the terminal assemblies $\alpha \in  A^{C_0}_\square [\mathcal{S_G}]$ started from a seed $\sigma \in Pl_N(C)$, the tile $t_{even}=(z_e,x_{ev},z_e,x_{ev})$  appears in the even regions and the tile $t_{odd}=(z_o,x_{od},z_o,x_{od})$ appears in the odd regions.
\end{lemma}

\begin{proof}

For an assembly that is started from a seed in a normal placement, by Lemma~\ref{skeleton} the $C$ is portioned by skeleton to odd and even regions.  
First, four ribbons of tiles types appear at the intersection of $R_X$ and the $R_Z$ ribbons. These ribbons form by tile types of Fig \ref{fig:innerfillingtiles}.
From the parts along $R_X$, straight lines of tiles start growing parallel to the $x$ axis using strength  $2$  glues $x_{ev}$ and $x_{od}$.
Thanks to modulo 5 (resp. 3) counters on the even (resp. odd) $R_X$ border tiles, there is one such line every other $5$ (resp. 3) position along that part of the border with tiles of type $t_{even}=(z_e,x_{ev},z_e,x_{ev})$  (resp. $t_{odd}=(z_o,x_{od},z_o,x_{od})$). These lines form the even (resp. odd) filling tiles and fill the partitioned regions. See Fig.~\ref{fig:inner-tile-types} for the tile types of partially inner filling of the regions and Fig.~\ref{fig:inner-filling} for an illustration of the assembly.

\end{proof}

\begin{figure}[h!]
\centering

\includegraphics[scale=0.5]{t_mfs.png}
\includegraphics[scale=0.5]{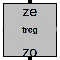}
\caption{The tile type $t_{mfs}$ and $t_{reg}$(which may only appear if $C$ has genus~1).}
\label{fig:inner-tile-types}
\end{figure}

The 2-coloring of $G_C$ indicates also where the tiles of type $t_{even}$ and of type $t_{odd}$ can be placed.
The region $R_{000}$ and the regions with even distance to it are tiled by tiles of type $t_{even}$, and the ones with odd distance to it, by $t_{odd}$ tiles. 
Therefore, the regions $R_{XYZ}$ ($X,Y,Z\in\{0,1\}$) are tiled with $t_{even}$ if and only if the sum $X + Y + Z $ is even, and the regions with even sum are covered by $t_{odd}$.
For more clarity, see Fig.~\ref{fig:RegionGraph} where the regions corresponding to $t_{even}$ are colored white and the ones with $t_{odd}$ are colored black in the region graph $G_C$ of the order-0 cuboid $C$.

\begin{figure}[h!]
    \centering
    \begin{subfigure}[t]{\textwidth}
    \includegraphics[width=\textwidth]{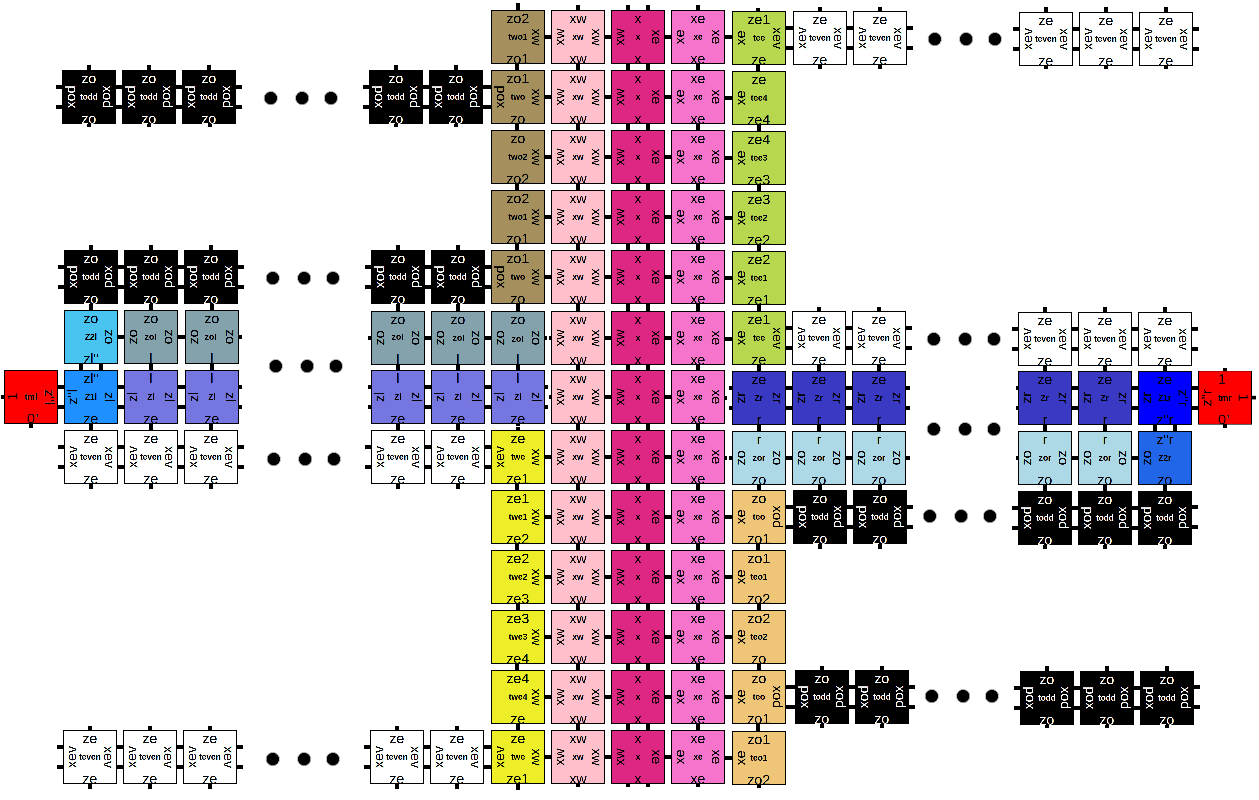}
    \caption{Up side collision of $P_X$ and $P_Z$.}
    \end{subfigure}\\
    \begin{subfigure}[t]{\textwidth}
    \includegraphics[width=\textwidth]{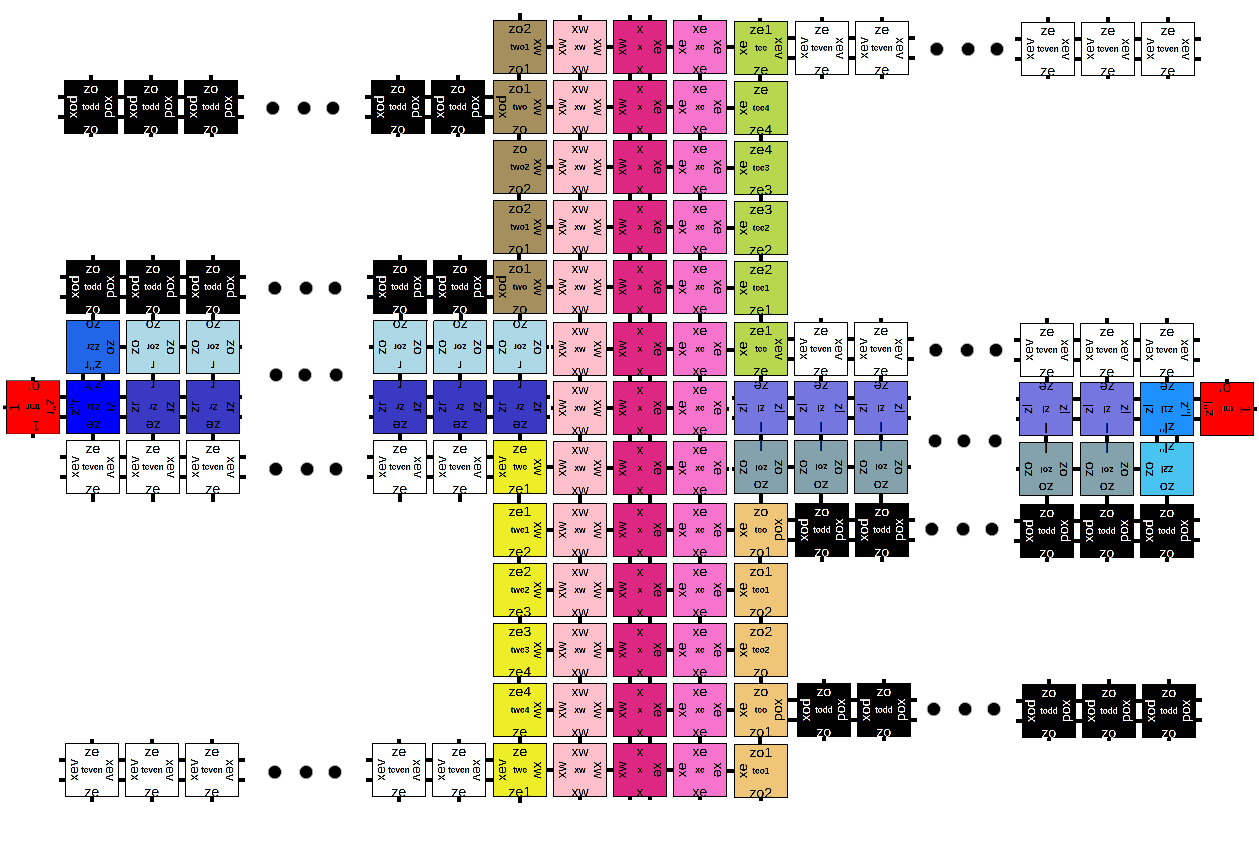}
    \caption{Down side collision of $P_X$ and $P_Z$.}
    \end{subfigure}
    \caption{The inner filling with tiles of types $t_{even}$ (white) and $t_{odd}$ (black) at the two places where the $P_Z$ ribbons meet $P_X$.}
    \label{fig:inner-filling}
\end{figure}

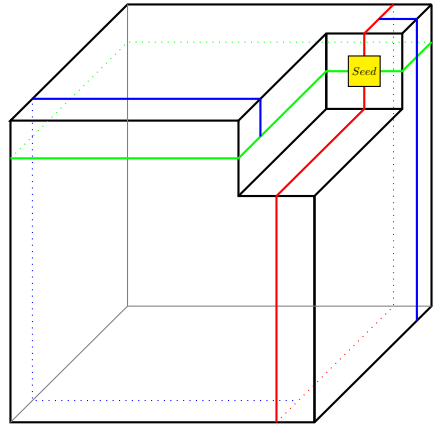
\begin{figure}[h!]
    \centering
  \begin{tikzpicture}[scale=2]
  \draw[thick](2,2,0)--(0,2,0)--(0,2,2)--(1.5,2,2)--(1.5,2,0.5)--(2,2,0.5)--(2,2,0);
  \draw[thick](2,1.5,2)--(2,0,2);
  \draw[thick](1.5,2,2)--(1.5,1.5,2)--(2,1.5,2)--(2,0,2);
    \draw[thick](1.5,1.5,2)--(1.5,1.5, 0.5)--(1.5,2,0.5);
    \draw[thick](2,1.5,0.5)--(1.5,1.5, 0.5)--(1.5, 2, 0.5);
    \draw[thick](0,2,2)--(0,0,2)--(2,0,2)--(2,0,0)--(2,2,0)(2,2,0.5)--(2,1.5,0.5)--(2,1.5,2);

  \draw[gray](2,0,0)--(0,0,0)--(0,2,0);
  \draw[gray](0,0,0)--(0,0,2);
  
   \draw[red,thick](1.75,2,0)--(1.75,2,0.5)--(1.75,1.75,0.5)--(1.75,1.5,0.5)--(1.75,1.5,2)--(1.75,0,2);
     \draw[red, dotted](1.75,0,2)--(1.75,0,0)--(1.75,2,0);
       \draw[green,thick](1.75,1.75,0.5)--(2,1.75,0.5)--(2,1.75,0);
        \draw[green, dotted](2,1.75,0)--(0,1.75,0)--(0,1.75,2);
        \draw[green,thick](0,1.75,2)--(1.5,1.75, 2)--(1.5,1.75, 0.5)--(1.75,1.75, 0.5) ;
         \draw[blue,thick](1.75,2,0.25)--(2,2,0.25)--(2,0,0.25);

    \draw[blue,thick](0,2,1.625)--(1.5,2,1.625);
    \draw[blue, dotted](0,1.75,1.625)--(0,2,1.625);
    \draw[blue, thick](1.5,2,1.625)--(1.5,1.75,1.625);
    
    \draw[blue, dotted](0,1.75,1.625)--(0,0,1.625)--(1.75,0,1.625);
  
    \node[inner ysep=10pt,rectangle,draw,fill=yellow,scale=0.15mm,label]  at (1.75,1.75,0.5) {$\textcolor{black}{Seed}$};

\end{tikzpicture}
    \caption{
    Cuboid with concavity: The three planes $P_X$ (red), $P_Y$ (green) and $P_Z$ (blue), which consists of two semi-planes.)
    }
    \label{concavity}
\end{figure}

\subsubsection{Terminal assemblies on order-1 cuboids with genus 1 : $ A^{C_t}_\square [\mathcal{S_G}]$}

We consider now in detail  the process of the assembly of $\mathcal{S_G}$ for order-1 cuboids with a tunnel. Let $C$ be an order-1 cuboid. 
The key element of the proof is the appearance of
some specific tile in each assembly when it has less than 8 regions.
The assemblies on $C$ have a skeleton with a different shape depending on the region graph associated with the placement of the seed.
Let $P_i$ and $R_i$ for $i \in \{ X, Y, Z \}$ be defined
as presented in Section~\ref{reg_part}. If a plane $P_i$ intersects along the width of the tunnel, it acts like a separator between the two parallel faces where the tunnel's entrances are located.  If a plane $P_i$ intersects along the length of the tunnel, the  tiles of $R_i$ enter and pass inside the tunnel. 
 Moreover, three types of partitions into regions are possible and the possible numbers of regions are: 7 regions when one plane intersects along the width of the tunnel, 5 regions when one plane intersects along the length of the tunnel and one along the width, and 1 region when  three perpendicular planes intersect along the tunnel, one along the width and the others along the length. Each case needs to be studied separately, we give the proof of the case with 5 regions and omit the proof of the other cases as they are analogous.

Note that in following $G_C(\sigma)$ refers to the region graph $G_C(p)$ such that $p$ is the position of the seed $\sigma$ on $C$.

\medskip

\noindent\emph{Case 1 (7 regions):  one plane intersects along the width of the tunnel.}

In this case, tiles of types $t_{odd}$ and $t_{even}$ touch, which enforces the attachment of $t_{reg}$ or $t_{mfs}$.

At least one of the planes $P_X$, $P_Y$ and $P_Z$ introduced in Section~\ref{reg_part} intersects with the tunnel of $C$, since its entrances are on parallel faces of the cuboid, and these planes are located between parallel faces.
When the tunnel has an intersection with only one of the three planes, the plane intersects along the width of the tunnel. For example, in Fig.~\ref{fig:7region}, the tunnel has an intersection with the plane $P_X$ only.

\begin{lemma}\label{tunnel1plane}
Let $C=C_0 \setminus C'_0 \in O_1^t$ be an order-1 cuboid  with the dimensions
at least 10 for $C'_0$. Assume that the seed $\sigma$ is placed in a normal placement $p \in Pl(C)$. In a terminal assembly of the system $S_{\mathcal{G}}$, if only one of the planes defined in Section~\ref{reg_part} intersect with the tunnel, $G_C(\sigma)$ has $7$ regions and a tile of type  $t_{reg}$ or $t_{mfs}$  appears in the assembly.
\end{lemma}

\begin{proof}
Let the seed be placed in a manner that only one of the planes $P_X$, $P_Y$ or $P_Z$ intersects along the width of the tunnel. The plane that intersects the tunnel is the separating buffer of two regions $R_{xyz}$ and $R_{x'y'z'}$ containing the two tunnel's entrances. 
In this case, the two regions $R_{xyz}$ and $R_{x'y'z'}$  get combined into a single region via the tunnel. Therefore, the number of distinct regions decreases to $7$ regions.
See~Fig.~\ref{fig:7region} for an illustration.

Without loss of generality, assume that $x+y+z$ is an odd number and $x'+y'+z'$ is an even number. When two regions $R_{xyz}$ and $R_{x'y'z'}$ are joined by the tunnel, tiles of type $t_{odd}=(z_o,x_{od},z_o,x_{od})$ from $R_{xyz}$ and of type $t_{even}=(z_e,x_{ev},z_e,x_{ev})$ from $R_{x'y'z'}$ both exist in the new unique region. it will be shown that the tile type $t_{reg}=(z_e,\epsilon,z_o,\epsilon)$ 
or $t_{mfs}$ must then occur in the assembly. The tile types $t_{reg}$ and $t_{mfs}$ are the only tile type of $\mathcal{S_G}$ with  labels $z_o$ and $z_e$ of  inner filling tiles $t_{odd}$ and $t_{even}$. 
In Fig.~\ref{fig:treg}, the places on $C$ that reveal a tunnel by $t_{reg}$ or $t_{mfs}$ are shown.
To conclude the proof, one needs to show that in a region with a disconnected border, there is a \emph{good} empty space, that is an empty space which sees both an even tile and an odd tile through strength 1 sides. Then, this space can be filled by neither type of filling tiles, but it must eventually be filled by a tile of type $t_{reg}=(z_e,\epsilon,z_o,\epsilon)$
or $t_{mfs}=(z_e,z_o,z_e,z_o)$.
In a region with a tunnel, on each side of the tunnel, the border of every $10 \times 10$ square must be crossed by either
\begin{itemize}
    \item at least two of the lines of tiles starting from $R_X$ on that side of the tunnel, or
    \item at least two of the lines exiting the tunnel.
\end{itemize}
In particular, because $C'_0 $ is at least $10 \times 10$ units wide, there are at least two lines crossing one of the edges the tunnel in the same direction. Each such line must either reach the opposite connected component of the border, be stopped orthogonally by a line from the opposite side of the tunnel, or run head-first into an opposite line.
Consider such a pair of lines, with minimal distance between them. In particular, that distance must be at most $10$.
\begin{itemize}
    \item If one of the lines reaches the opposite connected component of the border, either of the spaces next to its end is \emph{good} and in this case the tile of type $t_{reg}$ appears in the assembly;
    \item likewise, if one of them is stopped orthogonally by a line from the opposite side of the tunnel, one of the spaces next to the intersections is \emph{good} and a tile of type $t_{mfs}$ appears.
\end{itemize}
Moreover, if one of them runs head-first into an opposite line, the other cannot, because their distance cannot be at the same time divisible by 15, positive and less than 10. Hence the pair satisfies one the previous cases.
This concludes the proof of that case of our construction.
\end{proof}

\begin{figure}[h!]
    \centering
    \begin{subfigure}[t]{.45\textwidth}
     \centering
\begin{tikzpicture}[scale=2]
  \draw[thick](2,2,0)--(0,2,0)--(1,2,1)--(2,2,0)--(2,0,0)--(2,0,2)--(0,0,2)--(1,2,1);
  
  \draw[thick](1,2,1)--(2,0,2);
  \draw[gray](2,0,0)--(0,0,0)--(0,2,0);
  \draw[gray](0,0,0)--(0,0,2);
  \node[inner sep=2pt,circle,draw,fill=white, label]  at (0,0,0) {$\textcolor{black}{R_{000}}$};
   \node[inner sep=2pt,circle,draw,fill=white,label]  at (2,0,2) {$\textcolor{black}{R_{101}}$};

   \node[inner sep=2pt,circle,draw,fill=white,label]  at (2,2,0) {$\textcolor{black}{R_{110}}$};
  
\node[inner sep=2pt,circle,draw,fill=white, label]  at (0,2,0){$\textcolor{black}{R_{010}}$};
\node[inner sep=2pt,circle,draw,fill=white, label]  at (0,0,2){$\textcolor{black}{R_{001}}$};
\node[inner sep=2pt,circle,draw,fill=white, label]  at (2,0,0){$\textcolor{black}{R_{100}}$};

 \node[inner sep=0.1pt,circle,draw,fill=white,label]  at (1,1.75,1) {$\textcolor{black}{R_{011} = R_{111}}$};
\end{tikzpicture}
      \caption{The region graph}
    \end{subfigure}\qquad
    \begin{subfigure}[t]{.45\textwidth}
     \centering
      \begin{tikzpicture}[scale=2]
  \draw[thick](2,2,0)--(0,2,0)--(0,2,2)--(2,2,2)--(2,2,0)--(2,0,0)--(2,0,2)--(0,0,2)--(0,2,2);
  \draw[thick](2,2,2)--(2,0,2);
  \draw[gray](2,0,0)--(0,0,0)--(0,2,0);
  \draw[gray](0,0,0)--(0,0,2);

  \draw[red,thick](1,2,0)--(1,2,2)--(1,0,2);
  \draw[red, dotted, thick](1,0,2)--(1,0,0)--(1,2,0);
  
   \draw[green,thick](0,1,2)--(2,1,2)--(2,1,0);
  \draw[green, dotted, thick](2,1,0)--(0,1,0)--(0,1,2);

  \draw[thick](2,1.25,1.25)--(2,1.75, 1.25)--(2,1.75,1.75)--(2,1.25,1.75)--(2,1.25,1.25);
 \draw[gray](0,1.25,1.25)--(0,1.75, 1.25)--(0,1.75,1.75)--(0,1.25,1.75)--(0,1.25,1.25);
  
    \draw[blue,thick](0,2,1)--(2,2,1)--(2,0,1);
  \draw[blue,dotted,thick](2,0,1)--(0,0,1)--(0,2,1);

   \node[inner ysep=10pt,rectangle,draw,fill=yellow,scale=0.15mm,label]  at (1,1,2) {$\textcolor{black}{Seed}$};
 
\end{tikzpicture}
      \caption{The cuboid and the three planes.}
    \end{subfigure}

\caption{The case where $C\in O_1^t$ is partitioned into $7$ distinct regions. If there is a tunnel between two distinct regions, a tile of type $t_{reg}$ or $t_{mfs}$, which have common labels with both $t_{even}$ and $t_{odd}$, must appear in the assembly.}
    \label{fig:7region}
  \end{figure}
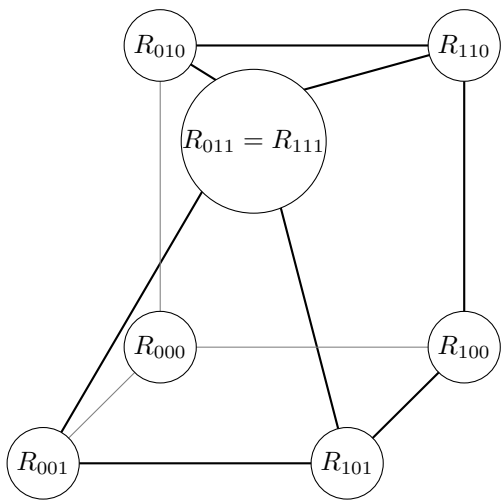
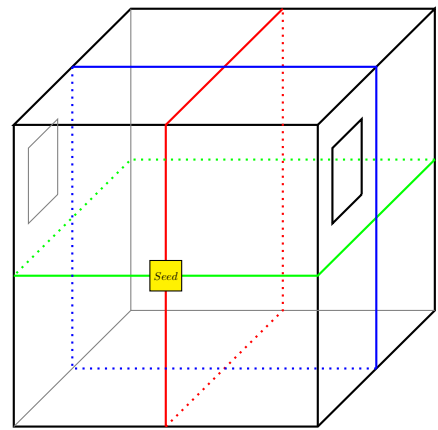

\medskip

\noindent\emph{Case 2 (5 regions): the tunnel intersects with $P_Z$, and exactly one of $P_X$ and $P_Y$.}

\begin{lemma}\label{tunnel2plane}
Let $C \in O_1^t$ be an order-1 cuboid and assume that the seed $\sigma$ is placed in a normal placement $p \in Pl(C)$. In a terminal assembly of the system $S_{\mathcal{G}}$, if the plane $P_Z$ and exactly one of the planes $P_X$ and $P_Y$ defined in Section~\ref{reg_part} have an intersection with the tunnel, there exist $5$ regions on the cuboid and a tile of type $t_{mfs}$ appears in the assembly.
\end{lemma}
\begin{proof}
If the seed is placed where the tunnel has intersection with two perpendicular planes, one of them intersects the tunnel along its width and the other one along its length.
If $P_Z$ intersects with the tunnel along the length, 
the ribbons of $R_Z$  meet each other inside the tunnel. However, if $P_Z$ intersects the tunnel along its width, they meet outside the tunnel.

In both cases, the tile $t_{mfs}=(z_e, z_o, z_e, z_o)$ appears in the assembly when two frame ribbons of $P_Z$ meet each other.
Note that when the tunnel has intersection with $P_Z$ and one of the planes $P_X$ or $P_Y$,  the cuboid is separated into two connected components such that one of them is a cuboid with genus $0$ and the other one is a cuboid with genus $1$.
The part with genus $0$ has $4$ distinct regions, and the part with genus $1$ (containing a tile of type $t_{mfs}$) has one single region. In total, there exist $5$ distinct regions on the cuboid $C$. For an illustration of the skeleton and its graph in this case, see Fig.~\ref{fig:5region}.
\end{proof}

\medskip

\noindent\emph{Case 3 (1 region): the tunnel intersects with $P_X$ and $P_Y$.} 

\begin{lemma}\label{tunnel3plane}
Let $C \in O_1^t$ be an order-1 cuboid and assume that the seed $\sigma$ is placed in a normal placement $p \in Pl(C)$. In a terminal assembly of the system $S_{\mathcal{G}}$, if two planes $P_X$ and $P_Y$ defined in Section~\ref{reg_part} intersect the tunnel, there exists $1$ regions on the cuboid and a tile of one of the $T_{ibc}$ types appears in the assembly.
\end{lemma}
\begin{proof}
In this case, the skeleton of the assembly is not the same as before. 
Recall the process of the assembly's skeleton:
The  $R_X$ are generated independently from $\sigma$. 
Two segment ribbons of $R_Y$ begin to grow after rebounding on the  $R_X$, regardless of passing through a tunnel or not. 
However, the ribbons of $R_Z$ start to grow only after finding the middle of $R_Y$ and they end by reaching the ribbon of $R_X$.
Considering this process, when the two planes $P_X$ and $P_Y$ intersect with the tunnel, the plane $P_Z$ is not able to form since there is a tunnel that does not permit to have the second collision of $R_Y$ and $R_X$. Therefore, the process of finding the middle of $R_Y$ is not able to continue and the ribbons of $R_Z$ is not able to form.

Moreover, two ribbons of $R_Y$ must meet each other at a tile of one of the $T_{ibc}$ types that comes between their $IBC1$ systems. This happens inside the tunnel if $P_Y$ intersects the tunnel along its length, and outside the tunnel if it intersects the tunnel along its width. 
In either cases a tile of one of the $T_{ibc}$ types appears.
See Fig.~\ref{fig:tibc} for an illustration of the places where the presence of a tunnel entrance implies that we have a tile of one of the $T_{ibc}$ types. 

Note that the skeleton consists of two closed loops of $R_X$ ribbons and $R_Y$ ribbons. 
 This phenomenon demonstrates that the genus of $C$ is 1.
In order to have a better overview, see Fig.~\ref{fig:PXPY}.
Furthermore, there is only one single region throughout the whole surface of $C$.
\end{proof}

 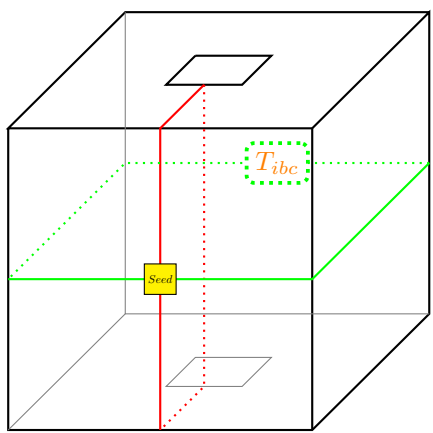
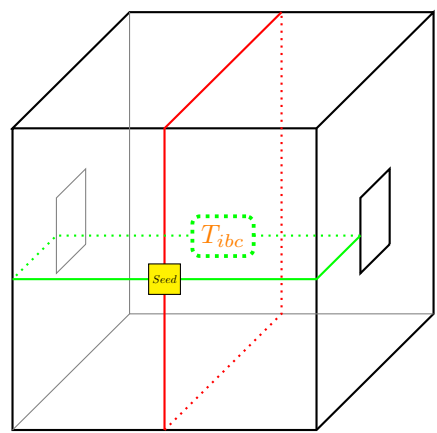
\begin{figure}[h!]
  \centering
  \begin{subfigure}[t]{.4\textwidth}
    \begin{tikzpicture}[scale=2]
      \draw[thick](2,2,0)--(0,2,0)--(0,2,2)--(2,2,2)--(2,2,0)--(2,0,0)--(2,0,2)--(0,0,2)--(0,2,2);
      \draw[thick](2,2,2)--(2,0,2);
      \draw[gray](2,0,0)--(0,0,0)--(0,2,0);
      \draw[gray](0,0,0)--(0,0,2);
      \draw[thick](0.75,2,0.75)--(1.25,2,0.75)--(1.25,2,1.25)--(0.75,2,1.25)--(0.75,2,0.75);
      \draw[gray](0.75,0,0.75)--(1.25,0,0.75)--(1.25,0,1.25)--(0.75,0,1.25)--(0.75,0,0.75);

      \draw[red,thick](1,2,1.25)--(1,2,2)--(1,0,2);
      \draw[red, dotted,thick](1,2,1.25)--(1,0,1.25)--(1,0,2);
      
      \draw[green,thick](0,1,2)--(2,1,2)--(2,1,0);
      \draw[green, dotted,thick](2,1,0)--(0,1,0)--(0,1,2);
      \node[line width=0.5mm,rectangle, minimum height=0.5cm,minimum width=0.5cm,fill=white!70,rounded corners=1mm,draw=green, label,dotted]  at (1,1,0) {$\textcolor{orange}{T_{ibc} }$};

      \node[inner ysep=10pt,rectangle,draw,fill=yellow,scale=0.15mm,label]  at (1,1,2) {$\textcolor{black}{Seed}$};
      
    \end{tikzpicture}
    
    \caption{The tunnel intersects along the length of plane $P_X$ and width of plane $P_Y$.}
  \end{subfigure}
  \begin{subfigure}[t]{.4\textwidth}

    \begin{tikzpicture}[scale=2]
      \draw[thick](2,2,0)--(0,2,0)--(0,2,2)--(2,2,2)--(2,2,0)--(2,0,0)--(2,0,2)--(0,0,2)--(0,2,2);
      \draw[thick](2,2,2)--(2,0,2);
      \draw[gray](2,0,0)--(0,0,0)--(0,2,0);
      \draw[gray](0,0,0)--(0,0,2);

      \draw[thick](2,0.75,0.75)--(2,1.25, 0.75)--(2,1.25,1.25)--(2,0.75,1.25)--(2,0.75,0.75);
      \draw[gray](0,0.75,0.75)--(0,1.25, 0.75)--(0,1.25,1.25)--(0,0.75,1.25)--(0,0.75,0.75);

      \draw[red,thick](1,2,0)--(1,2,2)--(1,0,2);
      \draw[red, dotted,thick](1,0,2)--(1,0,0)--(1,2,0);
      
      \draw[green,thick](0,1,2)--(2,1,2)--(2,1,1.25);
      \draw[green, thick, dotted](2,1,1.25)--(0,1,1.25)--(0,1,2);

      \node[line width=0.5mm,rectangle, minimum height=0.5cm,minimum width=0.5cm,fill=white!70,rounded corners=1mm,draw=green, label,dotted]  at (1,0.90,1) {$\textcolor{orange}{T_{ibc} }$};

      \node[inner ysep=10pt,rectangle,draw,fill=yellow,scale=0.15mm,label]  at (1,1,2) {$\textcolor{black}{Seed}$};
      
    \end{tikzpicture}
    \caption{The tunnel intersects along the width of plane $P_X$ and length of plane $P_Y$.}
  \end{subfigure}

  \caption{Intersection of tunnel with two planes $P_X$ (red) and $P_Y$ (green).}
  \label{fig:PXPY}
\end{figure}

Note that the situation when the seed is located inside the tunnel is similar to Case~3, up to topological isomorphism.  

 From Lemmas \ref{tunnel1plane}, \ref{tunnel2plane} and \ref{tunnel3plane}, the following corollary is obtained:
 
\input{FigureLatex/5region}

\begin{corollary}\label{thm1}

Let $C=C_0 \setminus C'_0 \in O_1$ be an order-1 cuboid  with the dimensions
at least 10 for $C'_0$ and  $\alpha$ be an assembly of the TAS $\mathcal {S_G} = (\Sigma, T, \sigma, str, \tau)$ such that its seed is placed at a normal placement.
 If  there is tunnel on $C$ (i.e. its genus is $1$), at least a tile type from $Y=\{t_{reg}\} \cup \{ t_{mfs} \}\cup T_{ibc} \subseteq T$ exists in all terminal assemblies of $\mathcal{S_G}$ on $C$ 

\end{corollary}
\begin{proof}

If there is a tunnel on $C$, at least one of the planes $P_X$, $P_Y$ and $P_Z$ defined in Section~\ref{reg_part} intersects with the tunnel since its entrances are on parallel faces of the cuboid, and these planes are located between parallel faces.

Firstly, if the tunnel of $C$ intersects with only one of the planes, due to Lemma~\ref{tunnel1plane},
a tile of type $t_{reg}$ or $t_{mfs}$, which are the only tile types of $\mathcal{S_G}$ with labels in common with both inner filling tile types $t_{odd}$ and $t_{even}$, appears in the assembly. In Fig.~\ref{fig:treg}, the places where the presence of these tile types displays the presence of the tunnel is shown.

Nextly, if two planes (among them $P_Z$) intersect with the tunnel on $C$, a tile of type $t_{mfs}$ appears in all terminal assemblies on $C$ by Lemma~\ref{tunnel2plane}. See Fig.~\ref{fig:tmfs} for the places where the presence of a tile of type $t_{mfs}$ displays the presence of the tunnel.

At the end, if two planes $P_X$ and $P_Y$ intersect with the tunnel, Lemma~\ref{tunnel2plane} implies that a tile of one of the $T_{ibc}$ types is present in the assembly. See Fig.~\ref{fig:tibc} for the places where the presence of a tunnel implies the presence of a tile of one of the $T_{ibc}$ types.

The places where a tunnel implies the presence of a tile of $Y$ are shown in Fig.~\ref{fig:treg-tmfs-tibc}.
\end{proof}

\subsection{Detecting the genus of the order-1 cuboids via $\mathcal {S_G}$ }

Before proving Theorem~\ref{maintheorem}, we need to prove following lemma:

\begin{figure}[h!]
   
    \begin{center}
\begin{tikzpicture}[scale=4.2]
  \draw[thick](2,2,0)--(0,2,0)--(0,2,2)--(2,2,2)--(2,2,0)--(2,0,0)--(2,0,2)--(0,0,2)--(0,2,2);
  \draw[thick](2,2,2)--(2,0,2);
 
  \draw[red,thick](1,2,0)--(1,2,2)--(1,0,2);
 
    \draw[blue,thick](0,2,1)--(2,2,1)--(2,0,1);

    \draw[green,thick](0,1,2)--(2,1,2)--(2,1,0);

  \node[line width=0.5mm,rectangle, minimum height=0.5cm,minimum width=0.5cm,fill=yellow,rounded corners=1mm,draw=red, label]  at (1,1,2) {$\textcolor{black}{Seed }$};
  
              \node[line width=0.5mm,rectangle, minimum height=0.5cm,minimum width=0.5cm,fill=white!70,rounded corners=1mm,draw=black, label]  at (0.5,2,0.40) {$\textcolor{black}{t_{reg/ mfs} }$};
              
               \node[line width=0.5mm,rectangle, minimum height=0.5cm,minimum width=0.5cm,fill=white!70,rounded corners=1mm,draw=black, label]  at (1.5,2,0.40) {$\textcolor{black}{t_{reg/ mfs} }$};
               
               \node[line width=0.5mm,rectangle, minimum height=0.5cm,minimum width=0.5cm,fill=white!70,rounded corners=1mm,draw=black, label]  at (1.5,2,1.60) {$\textcolor{black}{t_{reg/ mfs} }$};
               
                 \node[line width=0.5mm,rectangle, minimum height=0.5cm,minimum width=0.5cm,fill=white!70,rounded corners=1mm,draw=black, label]  at (0.5,2,1.60) {$\textcolor{black}{t_{reg/ mfs} }$};

          \node[line width=0.5mm,rectangle, minimum height=0.5mm,minimum width=0.5cm,fill=white!70,rounded corners=1mm,draw=black, label]  at (2,0.5,0.5) {$\textcolor{black}{t_{reg / mfs}}$};
          
          \node[line width=0.5mm,rectangle, minimum height=0.5mm,minimum width=0.5cm,fill=white!70,rounded corners=1mm,draw=black, label]  at (2,1.5,0.5) {$\textcolor{black}{t_{reg / mfs}}$};
          
            \node[line width=0.5mm,rectangle, minimum height=0.5mm,minimum width=0.5cm,fill=white!70,rounded corners=1mm,draw=black, label]  at (2,1.5,1.5) {$\textcolor{black}{t_{reg / mfs}}$};
            
             \node[line width=0.5mm,rectangle, minimum height=0.5mm,minimum width=0.5cm,fill=white!70,rounded corners=1mm,draw=black, label]  at (2,0.5,1.5) {$\textcolor{black}{t_{reg / mfs}}$};

         \node[line width=0.5mm,rectangle, minimum height=0.5,minimum width=0.5cm,fill=white!70,rounded corners=1mm,draw=black, label]  at (0.5,1.5,2) {$\textcolor{black}{t_{reg/ mfs} }$};
         
           \node[line width=0.5mm,rectangle, minimum height=0.5,minimum width=0.5cm,fill=white!70,rounded corners=1mm,draw=black, label]  at (1.5,1.5,2) {$\textcolor{black}{t_{reg/ mfs} }$};
         
          \node[line width=0.5mm,rectangle, minimum height=0.5,minimum width=0.5cm,fill=white!70,rounded corners=1mm,draw=black, label]  at (1.5,0.5,2) {$\textcolor{black}{t_{reg/ mfs} }$};
          
          \node[line width=0.5mm,rectangle, minimum height=0.5,minimum width=0.5cm,fill=white!70,rounded corners=1mm,draw=black, label]  at (0.5,0.5,2) {$\textcolor{black}{t_{reg/ mfs} }$};

    \end{tikzpicture}
    \end{center}
    \caption{The places on an order-1 cuboid that, if a tunnel is placed there, a tile of type $t_{reg}$ ou $t_{mfs}$ appears.}
    \label{fig:treg}
\end{figure}
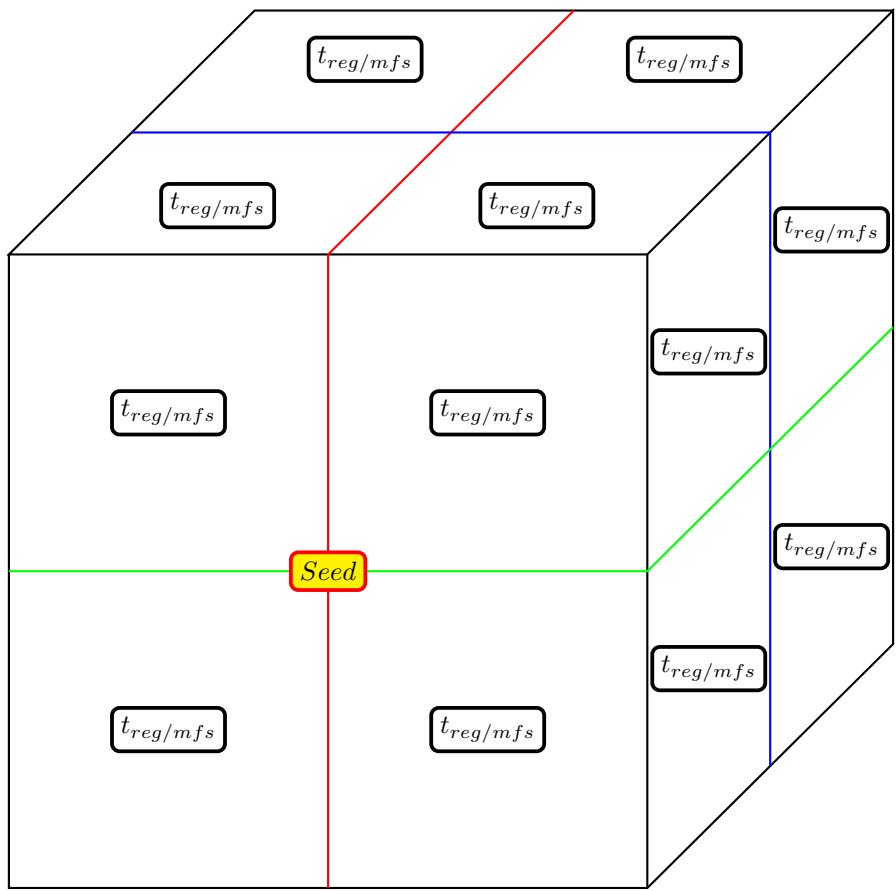

\begin{figure}[h!]
 
\begin{center}
\begin{tikzpicture}[scale=2.5]
  \draw[thick](2,2,0)--(0,2,0)--(0,2,2)--(2,2,2)--(2,2,0)--(2,0,0)--(2,0,2)--(0,0,2)--(0,2,2);
  \draw[thick](2,2,2)--(2,0,2);

  \draw[red,thick](1,2,0)--(1,2,2)--(1,0,2);
 
    \draw[blue,thick](0,2,1)--(2,2,1)--(2,0,1);

    \draw[green,thick](0,1,2)--(2,1,2)--(2,1,0);

  \node[line width=0.5mm,rectangle, minimum height=0.5cm,minimum width=0.5cm,fill=white!70,rounded corners=1mm,draw=red, label]  at (1,0.5,2) {$\textcolor{cyan}{t_{mfs} }$};
  
  \node[line width=0.5mm,rectangle, minimum height=0.5cm,minimum width=0.5cm,fill=white!70,rounded corners=1mm,draw=red, label]  at (1,1.5,2) {$\textcolor{cyan}{t_{mfs} }$};

  \node[line width=0.5mm,rectangle, minimum height=0.5cm,minimum width=0.5cm,fill=yellow,rounded corners=1mm,draw=red, label]  at (1,1,2) {$\textcolor{black}{Seed }$};

           \node[line width=0.5mm,rectangle, minimum height=0.5cm,minimum width=0.5cm,fill=white!70,rounded corners=1mm,draw=blue, label]  at (1.5,2,1) {$\textcolor{cyan}{t_{mfs} }$};
           
             \node[line width=0.5mm,rectangle, minimum height=0.5cm,minimum width=0.5cm,fill=white!70,rounded corners=1mm,draw=blue, label]  at (0.5,2,1) {$\textcolor{cyan}{t_{mfs} }$};

    \node[line width=0.5mm,rectangle, minimum height=0.5cm,minimum width=0.5cm,fill=white!70,rounded corners=1mm,draw=green, label]  at (1.5,1,2) {$\textcolor{cyan}{t_{mfs}}$};

    \node[line width=0.5mm,rectangle, minimum height=0.5cm,minimum width=0.5cm,fill=white!70,rounded corners=1mm,draw=green, label]  at (0.5,1,2) {$\textcolor{cyan}{t_{mfs} }$};

           \node[line width=0.5mm,rectangle, minimum height=0.5,minimum width=0.5cm,fill=white!70,rounded corners=1mm,draw=blue, label]  at (2,0.5,1) {$\textcolor{cyan}{t_{mfs} }$};
           
             \node[line width=0.5mm,rectangle, minimum height=0.5,minimum width=0.5cm,fill=white!70,rounded corners=1mm,draw=blue, label]  at (2,1.5,1) {$\textcolor{cyan}{t_{mfs} }$};

    \end{tikzpicture}

\end{center}
    \caption{The places on an order-1 cuboid where, if a tunnel is placed there, a tile of type $t_{mfs}$ appears.}
    \label{fig:tmfs}
\end{figure}
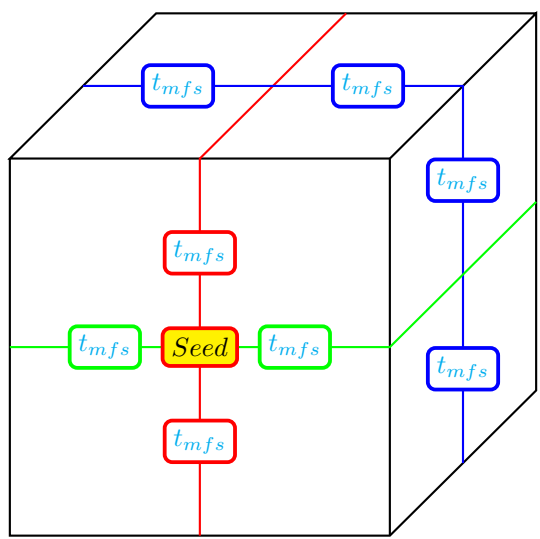

\begin{figure}[h!]
 
\begin{center}
\begin{tikzpicture}[scale=2.5]
  \draw[thick](2,2,0)--(0,2,0)--(0,2,2)--(2,2,2)--(2,2,0)--(2,0,0)--(2,0,2)--(0,0,2)--(0,2,2);
  \draw[thick](2,2,2)--(2,0,2);

  \draw[red,thick](1,2,0)--(1,2,2)--(1,0,2);
    \draw[green,thick](0,1,2)--(2,1,2)--(2,1,0);
 
  \node[line width=0.5mm,rectangle, minimum height=0.5cm,minimum width=0.5cm,fill=yellow,rounded corners=1mm,draw=red, label]  at (1,1,2) {$\textcolor{black}{Seed }$};
  
    \node[line width=0.5mm,rectangle, minimum height=0.5cm,minimum width=0.5cm,fill=white!70,rounded corners=1mm,draw=red, label]  at (1,2,1.55) {$\textcolor{orange}{t_{ibc} }$};
    
      \node[line width=0.5mm,rectangle, minimum height=0.5,minimum width=0.5cm,fill=white!70,rounded corners=1mm,draw=red, label]  at (1,2,0.45) {$\textcolor{orange}{t_{ibc} }$};

        \node[line width=0.5mm,rectangle, minimum height=0.5cm,minimum width=0.5cm,fill=white!70,rounded corners=1mm,draw=red, label]  at (1,2,1) {$\textcolor{orange}{T_{ibc} }$};
    
      \node[line width=0.5mm,rectangle, minimum height=0.5cm,minimum width=0.5cm,fill=white!70,rounded corners=1mm,draw=green, label]  at (2,1,1.5) {$\textcolor{orange}{t_{ibc} }$};
      
       \node[line width=0.5mm,rectangle, minimum height=0.5cm,minimum width=0.5cm,fill=white!70,rounded corners=1mm,draw=green, label]  at (2,1,0.5) {$\textcolor{orange}{t_{ibc} }$};
       
        \node[line width=0.5mm,rectangle, minimum height=0.5mm,minimum width=0.5cm,fill=white!70,rounded corners=1mm,draw=green, label]  at (2,1,0.95) {$\textcolor{orange}{t_{ibc} }$};
    
    \end{tikzpicture}

\end{center}
    \caption{The places on $C$ where $t_{ibc}$ displays the presence of a tunnel on $C$. 
  Note that the tunnel appears by $t_{ibc}$ also when the seed is inside the tunnel, since up to topological isomorphism, it is the same case}
    \label{fig:tibc}
\end{figure}
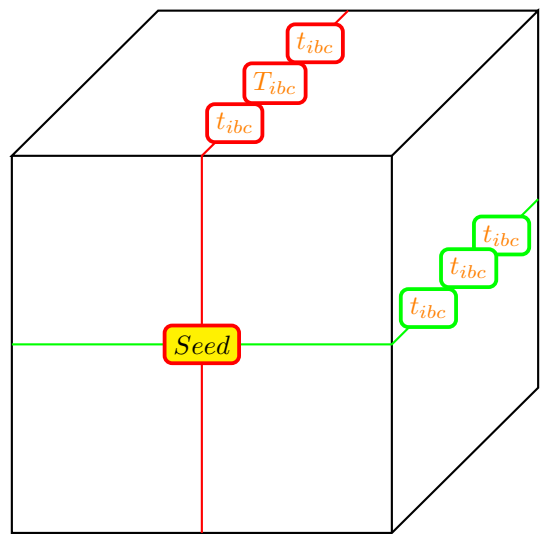

\begin{figure}[h!]
 
\begin{center}
\begin{tikzpicture}[scale=4.2]
  \draw[thick](2,2,0)--(0,2,0)--(0,2,2)--(2,2,2)--(2,2,0)--(2,0,0)--(2,0,2)--(0,0,2)--(0,2,2);
  \draw[thick](2,2,2)--(2,0,2);

  \draw[red,thick](1,2,0)--(1,2,2)--(1,0,2);
 
    \draw[blue,thick](0,2,1)--(2,2,1)--(2,0,1);

    \draw[green,thick](0,1,2)--(2,1,2)--(2,1,0);

  \node[line width=0.5mm,rectangle, minimum height=0.5cm,minimum width=0.5cm,fill=white!70,rounded corners=1mm,draw=red, label]  at (1,0.5,2) {$\textcolor{cyan}{t_{mfs} }$};
  
  \node[line width=0.5mm,rectangle, minimum height=0.5cm,minimum width=0.5cm,fill=white!70,rounded corners=1mm,draw=red, label]  at (1,1.5,2) {$\textcolor{cyan}{t_{mfs} }$};

  \node[line width=0.5mm,rectangle, minimum height=0.5cm,minimum width=0.5cm,fill=yellow,rounded corners=1mm,draw=red, label]  at (1,1,2) {$\textcolor{black}{Seed }$};
  
    \node[line width=0.5mm,rectangle, minimum height=0.5cm,minimum width=0.5cm,fill=white!70,rounded corners=1mm,draw=red, label]  at (1,2,1.55) {$\textcolor{orange}{T_{ibc} }$};
    
      \node[line width=0.5mm,rectangle, minimum height=0.5,minimum width=0.5cm,fill=white!70,rounded corners=1mm,draw=red, label]  at (1,2,0.45) {$\textcolor{orange}{T_{ibc} }$};

        \node[line width=0.5mm,rectangle, minimum height=0.5cm,minimum width=0.5cm,fill=white!70,rounded corners=1mm,draw=red, label]  at (1,2,1) {$\textcolor{orange}{T_{ibc} }$};
        
           \node[line width=0.5mm,rectangle, minimum height=0.5cm,minimum width=0.5cm,fill=white!70,rounded corners=1mm,draw=blue, label]  at (1.5,2,1) {$\textcolor{cyan}{t_{mfs} }$};
           
             \node[line width=0.5mm,rectangle, minimum height=0.5cm,minimum width=0.5cm,fill=white!70,rounded corners=1mm,draw=blue, label]  at (0.5,2,1) {$\textcolor{cyan}{t_{mfs} }$};
             
              \node[line width=0.5mm,rectangle, minimum height=0.5cm,minimum width=0.5cm,fill=white!70,rounded corners=1mm,draw=black, label]  at (0.5,2,0.40) {$\textcolor{black}{t_{reg/ mfs} }$};
              
               \node[line width=0.5mm,rectangle, minimum height=0.5cm,minimum width=0.5cm,fill=white!70,rounded corners=1mm,draw=black, label]  at (1.5,2,0.40) {$\textcolor{black}{t_{reg/ mfs} }$};
               
               \node[line width=0.5mm,rectangle, minimum height=0.5cm,minimum width=0.5cm,fill=white!70,rounded corners=1mm,draw=black, label]  at (1.5,2,1.60) {$\textcolor{black}{t_{reg/ mfs} }$};
               
                 \node[line width=0.5mm,rectangle, minimum height=0.5cm,minimum width=0.5cm,fill=white!70,rounded corners=1mm,draw=black, label]  at (0.5,2,1.60) {$\textcolor{black}{t_{reg/ mfs} }$};
  
    \node[line width=0.5mm,rectangle, minimum height=0.5cm,minimum width=0.5cm,fill=white!70,rounded corners=1mm,draw=green, label]  at (1.5,1,2) {$\textcolor{cyan}{t_{mfs}}$};

    \node[line width=0.5mm,rectangle, minimum height=0.5cm,minimum width=0.5cm,fill=white!70,rounded corners=1mm,draw=green, label]  at (0.5,1,2) {$\textcolor{cyan}{t_{mfs} }$};

      \node[line width=0.5mm,rectangle, minimum height=0.5cm,minimum width=0.5cm,fill=white!70,rounded corners=1mm,draw=green, label]  at (2,1,1.5) {$\textcolor{orange}{T_{ibc} }$};
      
       \node[line width=0.5mm,rectangle, minimum height=0.5cm,minimum width=0.5cm,fill=white!70,rounded corners=1mm,draw=green, label]  at (2,1,0.5) {$\textcolor{orange}{T_{ibc} }$};
       
        \node[line width=0.5mm,rectangle, minimum height=0.5mm,minimum width=0.5cm,fill=white!70,rounded corners=1mm,draw=green, label]  at (2,1,0.95) {$\textcolor{orange}{T_{ibc} }$};

          \node[line width=0.5mm,rectangle, minimum height=0.5mm,minimum width=0.5cm,fill=white!70,rounded corners=1mm,draw=black, label]  at (2,0.5,0.5) {$\textcolor{black}{t_{reg/mfs }}$};
          
          \node[line width=0.5mm,rectangle, minimum height=0.5mm,minimum width=0.5cm,fill=white!70,rounded corners=1mm,draw=black, label]  at (2,1.5,0.5) {$\textcolor{black}{t_{reg/mfs }}$};
          
            \node[line width=0.5mm,rectangle, minimum height=0.5mm,minimum width=0.5cm,fill=white!70,rounded corners=1mm,draw=black, label]  at (2,1.5,1.5) {$\textcolor{black}{t_{reg/mfs }}$};
            
             \node[line width=0.5mm,rectangle, minimum height=0.5mm,minimum width=0.5cm,fill=white!70,rounded corners=1mm,draw=black, label]  at (2,0.5,1.5) {$\textcolor{black}{t_{reg/mfs }}$};

         \node[line width=0.5mm,rectangle, minimum height=0.5,minimum width=0.5cm,fill=white!70,rounded corners=1mm,draw=black, label]  at (0.5,1.5,2) {$\textcolor{black}{t_{reg/ mfs} }$};
         
           \node[line width=0.5mm,rectangle, minimum height=0.5,minimum width=0.5cm,fill=white!70,rounded corners=1mm,draw=black, label]  at (1.5,1.5,2) {$\textcolor{black}{t_{reg/ mfs} }$};
         
          \node[line width=0.5mm,rectangle, minimum height=0.5,minimum width=0.5cm,fill=white!70,rounded corners=1mm,draw=black, label]  at (1.5,0.5,2) {$\textcolor{black}{t_{reg/ mfs} }$};
          
          \node[line width=0.5mm,rectangle, minimum height=0.5,minimum width=0.5cm,fill=white!70,rounded corners=1mm,draw=black, label]  at (0.5,0.5,2) {$\textcolor{black}{t_{reg/ mfs} }$};

           \node[line width=0.5mm,rectangle, minimum height=0.5,minimum width=0.5cm,fill=white!70,rounded corners=1mm,draw=blue, label]  at (2,0.5,1) {$\textcolor{cyan}{t_{mfs} }$};
           
             \node[line width=0.5mm,rectangle, minimum height=0.5,minimum width=0.5cm,fill=white!70,rounded corners=1mm,draw=blue, label]  at (2,1.5,1) {$\textcolor{cyan}{t_{mfs} }$};

    \end{tikzpicture}

\end{center}
    \caption{The places on a cuboid where, if there is a tunnel, a tile of $Y$ must appear in the assembly.}
    \label{fig:treg-tmfs-tibc}
\end{figure}
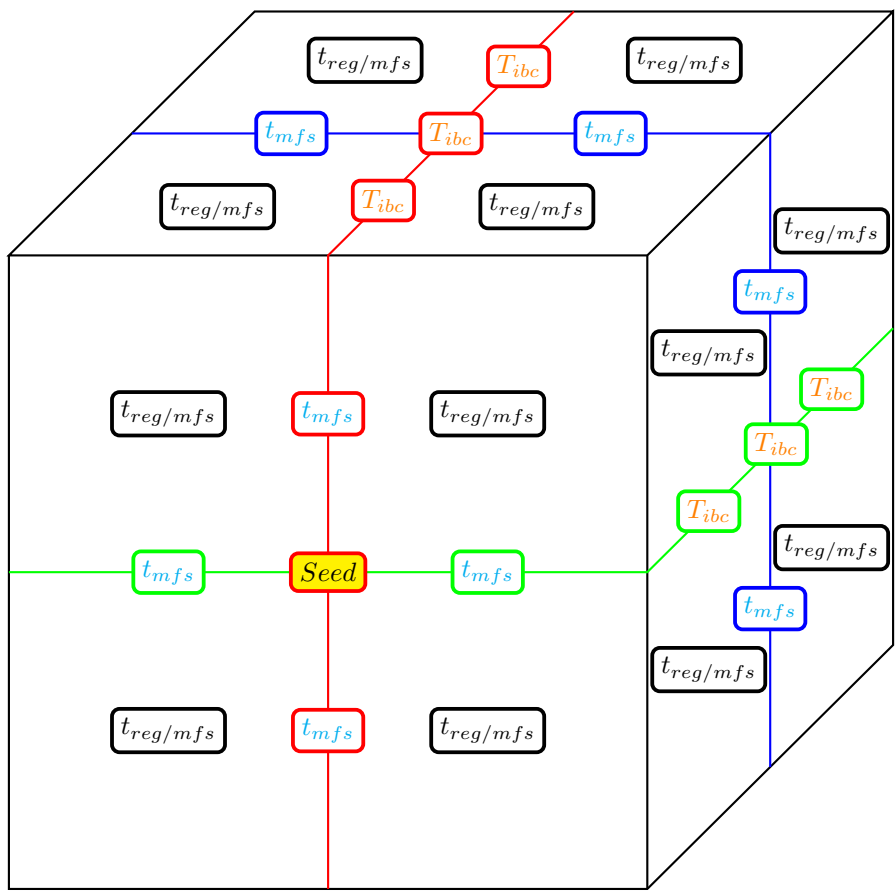

\medskip

\begin{lemma}\label{thm2}
Let $C$ be an order-1 cuboid. If one tile of Y=$\{t_{reg}\} \cup \{ t_{mfs} \}\cup T_{ibc} \subseteq T$ exists in a terminal assembly of $\mathcal{S_G}$ on $C$ starting from a seed in a normal placement, there is a tunnel on $C$.
\end{lemma}

\begin{proof}

Firstly, if a tile of type $t_{reg}$ is in the terminal assembly, its common labels with both inner filling tile types $t_{odd}$ and $t_{even}$ shows that at least two regions are connected i.e. there are at most $7$ distinct regions on $C$. Recall that by Section~\ref{genus0},  a terminal assembly of $\mathcal{S_G}$ on an order-1 cuboid with genus~0 partitions the cuboid into 8 distinct regions. Therefore, $C$ cannot have genus 0 and  there is a tunnel between these two regions. 

Secondly, if a tile of type $t_{mfs}$ exists in a terminal assembly on $C$,  two cases are possible. In one case there is a tunnel that intersect only $P_X$ with width and $t_{even}$ and $t_{odd}$ intersect perpendicularly each other and as a result $t_{mfs}$ appears in the assembly.
In the other case, two  ribbons of $R_Z$ must meet each other since the tiles whose labels correspond to the labels of $t_{mfs}$ are those of the $R_Z$ ribbons. Recall that the $R_X$ and $R_Y$ ribbons intersect at two places : one at the seed (since $R_Y$ grows out of $R_X$) and a second time, where the tiles of type
$t_{eu}$, 
$t_{ed}$, $t_{wu}$ or  $t_{wd}$
 appear in the assembly as the row tile number 1, in the second IBC system of the middle finding systems.
The two ribbons of $R_Z$, together with the parts of the middle finding system located between the second intersection of $R_X$ and $R_Y$ on the one hand, 
and $R_Z$ on the other hand, form a closed ribbon on the surface of $C$ (highlighted in green and blue Fig.~\ref{fig:planePZ}). This ribbon and $R_X$ pass through each other perpendicularly at only one place.  Since they pass through each other perpendicularly, it can be concluded that the cuboid $C$ cannot be topologically homeomorphic to the sphere, or in other words, be a genus 0 cuboid and a tunnel must exist.

Lastly, assume that a tile of one of the $T_{ibc}$ types appears in the assembly. Note that two labels $i_2$ and $i_4$ of every tiles of one of the $T_{ibc}$ types are the same as those of the  tiles of the $IBC1$ systems in $R_Y$ ribbons that are located opposite of each other (and no other tiles). Thus, the ribbons of $R_Y$ must collide, so that there is a tile of one of the $T_{ibc}$ types in the assembly and they do not reach the  $R_X$. 
Therefore, the ribbons $R_X$ and $R_Y$ have no intersection except at the seed. Since they pass through each other perpendicularly, as in the previous case, $C$ cannot be topologically homeomorphic to a sphere. Therefore, a tunnel must exist so that $R_X$ and $R_Y$ do not intersect in two places. 
\end{proof}

Furthermore, notice to the cases that  $\mathcal{S_G}$ assembles on an order-1 cuboid   $C\in O_1 ^c$ (the order-1 cuboids with concavity whose genus is $0$), or $C \in O_1^p$ (the order-1 cuboids with a pit whose genus is $0$). 
In these cases, the assembly's process is similar to the assembly on order-0 cuboids. The frame ribbons form completely by the assumption that the seed is located on a normal placement of $C$, the assembly's skeleton is formed completely and separates $C$ into $8$ distinct regions, and
the insides of the regions are tiled independently by inner filling lines of tiles of types $t_{odd}$ and $t_{even}$. 
However, in the case of $O_1 ^c$, the regions do not necessarily meet edge to edge, see Fig.~\ref{concavity} for an illustration.

After all that we have been through we are now ready to prove Theorem~\ref{maintheorem}. 
\begin{proof}
[Proof of Theorem~\ref{maintheorem}]
Let $C=C_0 \setminus C'_0 \in O_1$ be an order-1 cuboid  with the dimensions
at least 10 for $C'_0$ and  $\alpha$ be an assembly of the TAS $\mathcal {S_G} = (\Sigma, T, \sigma, str, \tau)$ such that its seed is placed at a normal placement.
Note that if $C_0$ is too small there is no normal placement. 
According to Corollary~\ref{thm1} and Lemma~\ref{thm2}, there is a tile type from $Y=\{t_{reg}\} \cup \{ t_{mfs} \}\cup T_{ibc} \subseteq T$ in all terminal assemblies of $\mathcal{S_G}$ on $C$ if and only if there is tunnel on $C$ (i.e. its genus is $1$).
\end{proof}

\section{Conclusion }\label{future}
We have introduced our new model, SFTAM, to perform tile self-assembly on 3D surfaces. We have shown that we can use self-assembly to determine the genus of a given surface. For this, we have worked on a simple and special family $O_1$ of polycubes, the order-1 cuboids.
		
It would be interesting to extend our results to a larger family of polycubes. Here, the middle finding system was used to detect the tunnel on the order-1 cuboids. However, for more complicated surfaces one needs to ensure that some part of the construction does go through the tunnel, and that it can be differentiated from the tiles it meets on the other side. The idea of having regions distinct identities can be reused in this context, but the Middle Finding System needs to be supplemented or replaced.


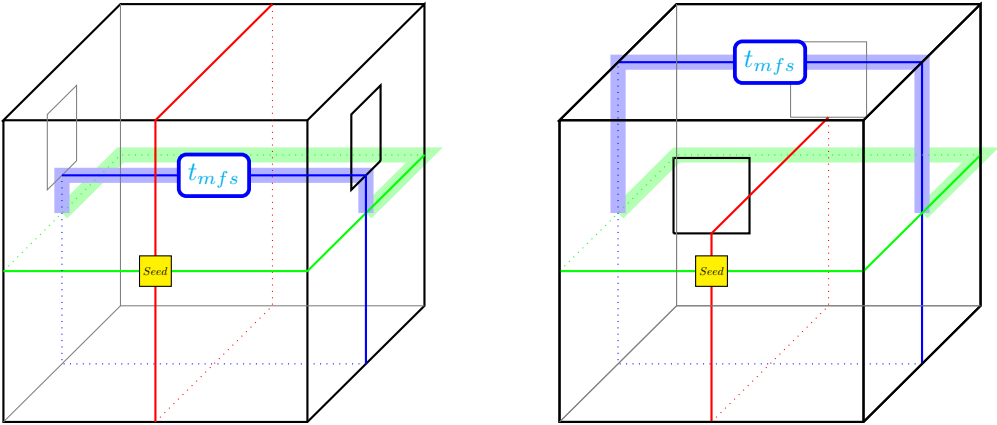
\begin{figure}[h!]
    \centering
    
\begin{subfigure}[t]{.45\textwidth}
\begin{tikzpicture}[scale=2]
 \draw[line width=2mm,green!30](0,1,1)--(0,1,0)--(2,1,0)--(2,1,1);
 \draw[line width=2mm,blue!30](0,1,1)--(0,1.25,1)--(2,1.25,1)--(2,1,1);

  \draw[thick](2,2,0)--(0,2,0)--(0,2,2)--(2,2,2)--(2,2,0)--(2,0,0)--(2,0,2)--(0,0,2)--(0,2,2);
  \draw[thick](2,2,2)--(2,0,2);
  \draw[gray](2,0,0)--(0,0,0)--(0,2,0);
  \draw[gray](0,0,0)--(0,0,2);

  \draw[red, dotted](1,0,2)--(1,0,0)--(1,2,0);

   \draw[green,thick](0,1,2)--(2,1,2)--(2,1,0);
  \draw[green, dotted](2,1,0)--(0,1,0)--(0,1,2);

  \draw[thick](2,1.25,0.75)--(2,1.75, 0.75)--(2,1.75,1.25)--(2,1.25,1.25)--(2,1.25,0.75);
 \draw[gray](0,1.25,0.75)--(0,1.75, 0.75)--(0,1.75,1.25)--(0,1.25,1.25)--(0,1.25,0.75);
  
    \draw[blue,thick](0,1.25,1)--(2,1.25,1)--(2,0,1);
  \draw[blue,dotted](2,0,1)--(0,0,1)--(0,1.25,1);

    \draw[red,thick](1,2,0)--(1,2,2)--(1,0,2);
  
       \node[line width=0.5mm,rectangle, minimum height=0.5,minimum width=0.5cm,fill=white!70,rounded corners=1mm,draw=blue, label]  at (1,1.25,1){$\textcolor{cyan}{t_{mfs} }$};

   \node[inner ysep=10pt,rectangle,draw,fill=yellow,scale=0.15mm,label]  at (1,1,2) {$\textcolor{black}{Seed}$};
 \end{tikzpicture}
 \end{subfigure}
 \begin{subfigure}[t]{.45\textwidth}
   \begin{tikzpicture}[scale=2]
  \draw[line width=2mm,green!30](0,1,1)--(0,1,0)--(2,1,0)--(2,1,1);
   \draw[line width=2mm,blue!30](0,1,1)--(0,2,1)--(2,2,1)--(2,1,1);
  \draw[thick](2,2,0)--(0,2,0)--(0,2,2)--(2,2,2)--(2,2,0)--(2,0,0)--(2,0,2)--(0,0,2)--(0,2,2);
  \draw[thick](2,2,2)--(2,0,2);
  \draw[gray](2,0,0)--(0,0,0)--(0,2,0);
  \draw[gray](0,0,0)--(0,0,2);
  
  \draw[thick](0.75, 1.25, 2)--(0.75, 1.75,2)--(1.25, 1.75, 2)--(1.25, 1.25, 2)--(0.75, 1.25, 2);
  \draw[gray](0.75, 1.25, 0)--(0.75, 1.75,0)--(1.25, 1.75, 0)--(1.25, 1.25, 0)--(0.75, 1.25, 0);

  \draw[red,thick](1,1.25,0)--(1,1.25,2)--(1,0,2);
  
  \draw[red, dotted](1,0,2)--(1,0,0)--(1,1.25,0);
  
    \draw[blue,thick](0,2,1)--(2,2,1)--(2,0,1);
  \draw[blue, dotted](2,0,1)--(0,0,1)--(0,2,1);
  
       \node[line width=0.5mm,rectangle, minimum height=0.5,minimum width=0.5cm,fill=white!70,rounded corners=1mm,draw=blue, label]  at (1,2,1){$\textcolor{cyan}{t_{mfs} }$};

    \draw[green,thick](0,1,2)--(2,1,2)--(2,1,0);
  \draw[green, dotted](2,1,0)--(0,1,0)--(0,1,2);

\node[inner ysep=10pt,rectangle,draw,fill=yellow,scale=0.15mm,label]  at (1,1,2) {$\textcolor{black}{Seed}$};

\draw[thick](2,2,0)--(0,2,0)--(0,2,2)--(2,2,2)--(2,2,0)--(2,0,0)--(2,0,2)--(0,0,2)--(0,2,2);
  \draw[thick](2,2,2)--(2,0,2);
  \draw[gray](2,0,0)--(0,0,0)--(0,2,0);
  \draw[gray](0,0,0)--(0,0,2);
    \end{tikzpicture}
    \end{subfigure}

    \caption{The closed ribbon formed by parts of the middle finding system (green) and the two ribbons of $R_Z$ (blue), when two $R_Z$ ribbons meet each other instead of reaching $R_X$. They meet the red ribbon $R_X$ only once.}
    \label{fig:planePZ}
\end{figure}


\begin{thebibliography}{000}
		

\bibitem{Abel} Z.Abel, N. Benbernou, M. Damian, E. D. Demaine, M. L. Demaine, R.Flatland, S. Kominers, and Robert Schweller. Shape replication through self-assembly and RNase enzymes. Proceedings of the \emph{21st Annual ACM-
SIAM Symposium on Discrete Algorithms (SODA
2010)}, pages 1045–1064, Austin, Texas, January
17–19 2010.
		
		\bibitem{polycubes} G. Aloupis, P. Bose, S. Collette, E. D. Demaine, M. L. Demaine, L. Dou\"{i}eb, V. Dujmovi\'{c}, J. Iacono, S. Langerman and P. Morin. Common unfoldings of polyominoes and polycubes. Proceedings of the  \emph{International Conference on Computational Geometry, Graphs and Applications}, CGGA'10. \emph{Lecture Notes in Computer Science} 7033:44--54, 2010.
		
		\bibitem{crystals} R. D. Barish, R. Schulman, P. W. K. Rothemund and E. Winfree. An information-bearing seed for nucleating algorithmic self-assembly. \emph{Proceedings of the National Academy of Sciences} 106(15):6054--6059, 2009.
		
		\bibitem{maze} M. Cook, T. St\'erin and D. Woods. Small tile sets that compute while solving mazes. Proceedings of the \emph{27th International Conference on DNA Computing and Molecular Programming (DNA 27)}, \emph{LIPIcs} 205, 8:1--8:20, 2021.
		
		\bibitem{FTAM} J. Durand-Lose, J. Hendricks, M. J. Patitz, I. Perkins and M. Sharp. Self-assembly of 3-D structures using 2-D folding tiles. \emph{Natural Computing} 19, 337--355, 2020.
		
		
		\bibitem{seeman} M. B. Jones, N. C. Seeman and C. A. Mirkin. Programmable Materials and the Nature of the DNA Bond. \emph{Science} 347:840--840, 2015.

        \bibitem{3D2001} M. Y. Kao and V. Ramachandran. DNA self-assembly for constructing 3D boxes. Proceedings of the \emph{12th International Symposium on Algorithms and Computation}, ISAAC'01, Lecture Notes in Computer Science 2223:429--441, 2001.
		
		\bibitem{crystals2} W. Liu, H., Zhong, R., Wang and N. C. Seeman. Crystalline two-dimensional DNA-origami arrays. \emph{Angewandte Chemie International Edition} 50(1):264--267, 2011.

		\bibitem{Survey} M. J. Patitz. An introduction to tile-based self-assembly and a survey of recent results.\emph{Natural Computing} 13, 195--224, 2014.

        \bibitem{2D} P. W. K. Rothemund. Folding DNA to create nanoscale shapes and patterns. \emph{Nature} 440(7082):297--302, 2006.

        \bibitem{BinaryCounter} P. W. K. Rothemund and E. Winfree. The Program-Size Complexity of Self-Assembled Squares. Proceedings of the \emph{Thirty-Second Annual ACM Symposium on Theory of Computing}, STOC 2000, 459--468, 2000.

   \bibitem{Hydrogel} V. A. Sontakke and Y. Yokobayashi. Programmable Macroscopic Self-Assembly of DNA-Decorated Hydrogels. \emph{Journal of the American Chemical Society} 144 (5), 2149-2155,  2022.
   
    	\bibitem{SeemanBIB}	DNA Nanotechnology: Bibliography from Ned Seeman's Laboratory. \url{http://seemanlab4.chem.nyu.edu/nanobib.html}


        \bibitem{wang} H. Wang. Proving theorems by pattern recognition -- II. \emph{Bell System Technical Journal} 40 (1):1--41, 1961.

        \bibitem{WinfreeThesis} E. Winfree. \emph{Algorithmic Self-Assembly of DNA.} Ph.D.thesis, California Institute of Technology, 1998.
        
        


         
        \bibitem{Shape-assisted} J.F. Woods,  L. Gallego, P. Pfister,  M. Maaloum,   A.V.  Jentzsch, and  M. Rickhaus. \emph{Shape-assisted self-assembly.} Nat. Commun.  13, 368 , 2022.
        
           \bibitem{Fabrication} S.Zhung,  \emph{ Fabrication of novel biomaterials through molecular self-assembly.} Nat Biotechnol 21, 1171–1178 , 2003.
        
        
        \bibitem{origami} R. Zhuo, F. Zhou, X. He, R. Sha, N. C. Seeman and P. M. Chaikin. Litters of self-replicating origami cross-tiles. \emph{Biophysics and Computational Biology} 116 (6) 1952--1957, 2019.
        
	\end{thebibliography}
\end{document}